\documentclass[article,12pt,]{article}
\usepackage{amsmath}
\usepackage{graphicx}
\usepackage{enumerate}
\usepackage{natbib}
\usepackage{url} % not crucial - just used below for the URL 
\usepackage{amsmath,amssymb,amsthm}
\usepackage{dcolumn}
\usepackage{subcaption}
\makeatletter
\newcommand*{\addFileDependency}[1]{% argument=file name and extension
\typeout{(#1)}
\@addtofilelist{#1}
\IfFileExists{#1}{}{\typeout{No file #1.}}
}
\makeatother

\usepackage{lscape}
\newcommand{\blind}{1}

\newcommand{\bm}[1]{\mbox{\boldmath$ #1$}}
\newcommand{\mboxa}[1]{\mbox{\normalfont#1}}
\renewcommand{\vec}{\mboxa{vec}}

\newcommand{\tr}{\mboxa{tr}}
\newcolumntype{d}[1]{D{.}{.}{#1}}
\newenvironment{remark}[1][Remark.]{\begin{trivlist}
\item[\hskip\labelsep {\bfseries#1}]}{\end{trivlist}}%\renewcommand{\labelenumi}{(\alph{enumi})}

\usepackage{color}

\newtheorem{theorem}{Theorem}

\newtheorem{assumption}{Assumption}
\newcommand{\beq}{\begin{equation}}
\newcommand{\eeq}{\end{equation}}
\newcommand{\beqw}{\begin{equation*}}
\newcommand{\eeqw}{\end{equation*}}

\newcommand{\eps}{\epsilon}

\newcommand{\xx}{\mbox{\boldmath$x$}}

\newcommand{\dd}{\mbox{\boldmath$d$}}

\newcommand{\pp}{\mbox{\boldmath$p$}}

\newcommand{\uu}{\mbox{\boldmath$u$}}
\newcommand{\vv}{\mbox{\boldmath$v$}}

\def\Cb{{\mathbb C}}

\def\Eb{{\mathbb E}}

\def\Gb{{\mathbb G}}
\def\Hb{{\mathbb H}}
\def\Jb{{\mathbb J}}
\def\Kb{{\mathbb K}}
\def\Lb{{\mathbb L}}
\def\Mb{{\mathbb M}}

\def\Pb{{\mathbb P}}
\def\Rb{{\mathbb R}}
\def\Ub{{\mathbb U}}
\def\Vb{{\mathbb V}}
\def\Wb{{\mathbb W}}

\def\Ac{{\mathcal A}}
\def\Rc{{\mathcal R}}
\def\Kc{{\mathcal K}}

\def\Uc{{\mathcal U}}

\def\Fc{{\mathcal F}}

\def\Hc{{\mathcal H}}
\def\Mc{{\mathcal M}}
\def\Nc{{\mathcal N}}
\def\Pc{{\mathcal P}}

\begin{document}

\def\spacingset#1{\renewcommand{\baselinestretch}%
{#1}\small\normalsize} \spacingset{1}

%%%%%%%%%%%%%%%%%%%%%%%%%%%%%%%%%%%%%%%%%%%%%%%%%%%%%%%%%%%%%%%%%%%%%%%%%%%%%%

\if1\blind
{
\title{\bf High-\textcolor{black}{d}imensional Sparse Multivariate Stochastic Volatility Models}
\author{Benjamin Poignard\thanks{
Benjamin Poignard, Graduate School of Economics, Osaka University, Toyonaka, Osaka 560-0043, and Riken-AIP, Chuo-ku, Tokyo 103-0027, Japan (email: bpoignard@econ.osaka-u.ac.jp; Corresponding author); Manabu Asai, Faculty of Economics, Soka University, Hachioji, Tokyo 192-8577, Japan (e-mail: m-asai@soka.ac.jp). 
\textcolor{black}{The authors are most grateful to Yoshihisa Baba, Shiqing Ling, the associate editor, and \textcolor{black}{the} two anonymous reviewers for \textcolor{black}{their} helpful comments and suggestions.}
This \textcolor{black}{research} is partially supported by the Japan Society for the Promotion of Science (19K01594, 19K23193). 
}\hspace{.2cm}\\
    Graduate School of Economics, Osaka University and
    Riken-AIP, Japan\\
    and \\
Manabu Asai\\
    Faculty of Economics, Soka University, Hachioji, Tokyo, Japan}
\maketitle
} \fi

\if0\blind
{
\bigskip
\bigskip
\bigskip
\begin{center}
    {\LARGE\bf High-\textcolor{black}{d}imensional Sparse Multivariate Stochastic Volatility Models}
\end{center}
\medskip
\vspace*{1cm}
} \fi

\vspace*{-1cm}
\bigskip
\begin{abstract}
Although multivariate stochastic volatility models usually produce more accurate forecasts compared to \textcolor{black}{the} MGARCH models, their estimation techniques such as Bayesian MCMC typically suffer from the curse of dimensionality. We propose a fast {and efficient} estimation approach for MSV based on a penalized OLS framework. Specifying the MSV model as a multivariate state space model, we carry out a two-step penalized procedure. We provide the asymptotic properties of the two-step estimator and the oracle property of the first-step estimator when the number of parameters diverges. The performances of our method are illustrated through simulations and financial data. 
\end{abstract}

\noindent%

\vfill

%\newpage
\spacingset{1.5} % DON'T change the spacing!

\section{Introduction}

Over \textcolor{black}{the} past decades, various covariance models have been developed for describing dynamic structures for multivariate economic and financial time series.
Within the Multivariate GARCH (MGARCH) family, the dynamic conditional correlation (DCC) model of Engle (2002) and Tse and Tsui (2002), the BEKK model of Baba et al. (1985) and Engle and Kroner (1995), %the orthogonal GARCH models of Alexander and Chibumba (1997) and Alexander (2001), 
and their variants are commonly used: see the survey of Bauwens, Laurent, and Rombouts (2006), for instance.
{As for} the multivariate stochastic volatility (MSV) family, the MSV model of Harvey, Ruiz, and Shephard (1994) was extended, among others, by the factor model of Chib, Nardari, and Shephard (2006) and the dynamic correlation model of Asai and McAleer (2009b): see Ghysels, Harvey, and Renault (1996), Asai, McAleer, and Yu (2006), and Chib, Omori, and Asai (2009) for various univariate and multivariate SV models.
Based on a thorough empirical analysis, Chib, Nardari, and Shephard (2006) highlighted that the MSV models usually outperform MGARCH-based models in terms of out-of-sample forecasts.

{Several} methods for estimating \textcolor{black}{the} MSV models have been developed. In their seminal work, Harvey, Ruiz, and Shephard (1994) derived a state space form based on the vector of the logarithm of squared returns. 
Based on this state space setting, they performed a Kalman-based filtering technique to evaluate and optimize the quasi log likelihood function. In the recent literature, a commonly used method is the Bayesian Markov Chain Monte Carlo \textcolor{black}{(MCMC)}, as described, \textcolor{black}{for example}, in Chib, Omori, and Asai (2009) and Kastner, Fr{$\ddot{\mbox{u}}$}hwirth-Schnatter, and Lopes (2017), among others.
An alternative estimation approach is the Monte Carlo Likelihood (MCL) method suggested by Durbin and Koopman (1997, 2001) and applied by Asai, Caporin, and McAleer (2015) and Asai and McAleer (2009a).
However, empirical applications in the literature are typically limited to low-dimensional random vectors when \textcolor{black}{methods based on} MCMC or MCL are performed, due to the severe costs in terms of computations, or the intricate choice of suitable priors (for the MCMC case), among others. In the same vein, MGARCH specifications also suffer from the so-called “curse of dimensionality” since the complexity is of order $O(p^2)$ \textcolor{black}{in general}, where $p$ corresponds to the problem dimension, as the specification of a general multivariate dynamic model often induces an explosion of the number of free parameters. Moreover, tricky conditions are required \textcolor{black}{for} the model parameters to satisfy the positive-definiteness of the variance-covariance process.

Another key hurdle of the aforementioned methods is the high non-linearity of the models, which requires the use of likelihood-based estimation techniques. Therefore, strongly reduced versions of such multivariate models are most often considered as soon as $p$ is larger than four or five. {The factor-model-based approach may be a solution to shrink the number of parameters. In particular,} Kastner, Fr{$\ddot{\mbox{u}}$}hwirth-Schnatter, and Lopes (2017) considered factors in their stochastic volatility {framework and provided a joint specification of a large number of covarying time series using a small number of latent factors}. However, this factor-based method requires the identification of the corresponding factors together with the treatment of the rotational \textcolor{black}{indeterminacy} inherent to factor models.

The objective of this \textcolor{black}{study} consists in \textcolor{black}{modeling} high-dimensional variance-covariance matrices within the \textcolor{black}{MSV} framework in a flexible manner and breaking the curse of dimensionality without relying on standard \textcolor{black}{procedures based on} MCMC or MCL. To do so, we introduce a vector autoregressive and moving-average (VARMA) representation for the MSV model {in the same spirit as} Harvey, Ruiz, and Shephard (1994) and apply an OLS-based two-step estimation approach extending the idea of Hannan and Rissanen (1982) and Hannan and Kavalieris (1984).
{More precisely, \textcolor{black}{as} a first step, we carry out an OLS estimation of a large dimensional VAR model with a sufficiently large number of lags to approximate the VARMA model. For the purpose of parsimony and to avoid over-fitting, we enforce the nullity of potentially numerous model coefficients using a penalization procedure on the model coefficients}. Our study shares a similar spirit with Poignard and Fermanian (\textcolor{black}{2021}), who provided a framework for high-dimensional variance-covariance within the MGARCH family: they derived some parameterizations to directly generate positive-definite covariance matrices based on multivariate ARCH processes allowing for a linear representation with respect to the parameters. \textcolor{black}{However}, our work differs from theirs in two main respects: our analysis lies within the \textcolor{black}{MSV} family; we consider a general penalization framework {for efficient estimation}, which includes a broad range of potentially non-convex penalty functions.

The main contributions of our method are as follows: using a {penalized} OLS framework, we can directly generate positive-definite variance-covariance matrices without relying on \textcolor{black}{methods like} MCMC/MCL and manage high-dimensional matrix processes; the large sample properties of the two-step estimator are provided; {in particular, we prove the oracle property of the first step estimator with a diverging dimension in the sense of Fan and Li (2001), which ensures the correct identification of the underlying set of nonzero coefficients.}

The remainder of the paper is \textcolor{black}{organized} as follows. In Section \ref{framework_setting}, we describe the framework and the new forecasting procedure based on a \textcolor{black}{penalized} OLS estimation framework. Section \ref{asymptotic_theory} contains the large sample properties of the \textcolor{black}{penalized} two-step OLS estimator. Section \ref{empirical_illustration} reports simulation-based experiment results for in-sample estimates of covariance matrices together with out-of-sample forecasting results based on real financial portfolios. Finally, Section \ref{conclusion} concludes the paper. All proofs and intermediary results are in the Appendix.

\textbf{\emph{Notations.}} Throughout this paper, we denote the cardinality of a set $E$ by $\text{card}(E)$. For a vector $\vv \in \Rb^d$, the $\ell_p$ norm is $\|\vv\|_p = \big(\sum^p_{k=1} |\vv_k|^p  \big)^{1/p}$ for $p > 0$, and $\|\vv\|_{\infty} = \underset{i}{\max}|\vv_i|$. Let the subset \textcolor{black}{be} $\Ac \subseteq \{1,\cdots,d\}$; then, $\vv_{\Ac} \in \Rb^{\text{card}(\Ac)}$ is the vector $\vv$ restricted to $\Ac$. $\Mc_{m \times n}(\Rb)$ denotes the space of $m \times n$ matrices with coefficients in $\Rb$. For a matrix $A$, $\|A\|_{\textcolor{black}{F}}$ is the \textcolor{black}{Frobenius} norm. We write $A^\top$ (resp. $\vv^\top$) to denote the transpose of the matrix $A$ (resp. the vector $\vv$). We write $\text{vec}(A)$ to denote the vectorization operator that stacks the columns of $A$ on top of one another into a vector. We denote by $\text{vech}(A)$ the $p(p+1)/2$ vector that stacks the columns of the lower triangular part of the square and symmetric matrix $A$. \textcolor{black}{$\lambda_{\min}(A)$ (resp. $\lambda_{\max}(A)$) denotes the minimum (resp. maximum) eigenvalue of $A$.}  \textcolor{black}{We write $\tr(A)$ to denote the trace of the square matrix $A$}. The $I_p$ matrix is the $p$-dimensional identity matrix. For a function $f: \Rb^d \rightarrow \Rb$, we denote \textcolor{black}{the gradient or subgradient of $f$ by $\nabla f$ and the Hessian of $f$ by $\nabla^2 f$}. We denote by $(\nabla^2 f)_{\Ac \Ac}$ the Hessian of $f$ restricted to the block $\Ac$. We write $\Ac^c$ to denote the complement of the set $\Ac$.

\vspace*{0.5cm}

\section{Penalized OLS framework for MSV}\label{framework_setting}

\subsection{Framework}

We consider a $p$-dimensional vectorial stochastic process $(y_t)_{t=1,\cdots,T}$ \textcolor{black}{and denote the vector of its model parameters by $\theta$. We then consider an MSV} decomposition given as 
\begin{align}
& y_t = D_t \varepsilon_t, \quad \varepsilon_t \sim iid(0, \Gamma), \label{eq:yt} \\
& h_{t+1} = \mu + \Phi (h_t -\mu) + \eta_t, \quad \eta_t \sim \Nc_{\Rb^p}(0, \Sigma_{\eta}), \label{eq:ht}
\end{align}
where $\Gamma$ is a $p \times p$ correlation matrix, $\varepsilon_t = (\varepsilon_{1t},\ldots,\varepsilon_{pt})^\top$ is a $p \times 1$ random vector, which is independently and identically distributed (i.i.d.), centered with variance-covariance $\Gamma$, $h_t = (h_{1t},\ldots,h_{pt})^\top$ is a $p \times 1$ vector of log-volatility, $D_t = \mbox{diag}\big(\exp( h_{1t}/2), \ldots, \exp( h_{pt}/2) \big)$ is a diagonal matrix of volatility, $\mu =(\mu_1,\ldots,\mu_p)^\top$ is a $p \times 1$ vector, $\Phi$ is a $p \times p$ matrix, and $\Sigma_{\eta}$ is a $p \times p$ covariance matrix of $\eta_t$. The MSV model (\ref{eq:yt}) and (\ref{eq:ht}) reduces to the MSV model of Harvey, Ruiz, and Shepard (1994) when $\Phi$ is diagonal and $\varepsilon_{it}$ follows a $t$ distribution.

\textcolor{black}{Subsequently}, we define $\textcolor{black}{y_t^{\ell}} = (\log( y_{1t}^2),\ldots, \log( y_{pt}^2))^\top$. Following Harvey, Ruiz, and Shepard (1994), the MSV model can be formulated as a state space model:
\begin{align}
& \textcolor{black}{y_t^{\ell}}  = c + \alpha_t + \zeta_t, \label{eq:ss1}\\
& \alpha_{t+1} = \Phi \alpha_t + \eta_t, \label{eq:ss2}
\end{align}
where $c=(c_1,\ldots,c_p)^\top$, $\zeta_t = (\zeta_{1t},\ldots,\zeta_{pt})^\top$, and $\alpha_t = h_t -\mu$ with 
$c_i = \mu_i + \Eb[\log (\varepsilon_{it}^2)]$ and $\zeta_{it} = \log( \varepsilon_{it}^2) - \Eb[\log (\varepsilon_{it}^2)]$.
Assuming a $t$ distribution for $\varepsilon_{it}$, Harvey, Ruiz, and Shepard (1994) specified the covariance matrix of $\zeta_t$ as $\Sigma_{\zeta}$. Note that $\Eb[\zeta_t]=0$ by definition.
Based on the state space form, these authors suggested a quasi-maximum likelihood estimation of the MSV model using the Kalman filter. Alternative methods were proposed such as the Bayesian MCMC technique of Chib, Nardari, and Shephard (2006) \textcolor{black}{and} the Monte Carlo Likelihood (MCL) method of Durbin and Koopman (1997, 2001). A significant drawback of these methods is the computational cost and, thus, the curse of dimensionality: most of the applications are restricted to small vector sizes and/or reduced forms are fostered.

In this paper, we aim \textcolor{black}{to tackle} this issue for \textcolor{black}{the} MSV models using a penalized OLS estimation method. Although the MSV model (\ref{eq:yt}) and (\ref{eq:ht}) might be a basic model, the following advantages with respect to \textcolor{black}{the} MGARCH models can be highlighted: (i) relatively stable estimates and forecasts for variance-covariance matrices; (ii) simpler restrictions for stationarity conditions; \textcolor{black}{and} (iii) no intricate matrix parameterization and/or parameter restrictions to generate positive-definite matrices. 
\textcolor{black}{Regarding (i), see the theoretical comparison of Taylor (1994, Section 5) and the empirical results of Dan{\'i}elsson (1998) and Ding and Vo (2012), for instance. As for (ii) and (iii), see Bauwens, Laurent, and Rombouts (2006) and Chib, Omori, and Asai (2009) for \textcolor{black}{the} MGARCH and MSV models, respectively.}

\textcolor{black}{
As in Harvey, Ruiz, and Shephard (1994) and Kim, Shephard, and Chib (2002), we consider the log of squared returns. Harvey, Ruiz, and Shephard (1994) suggested the quasi-maximum likelihood (QML) estimation based on the Kalman filter, by treating the distribution of $\log \chi^2(1)$ as a normal distribution.
Ruiz (1994) analyzed the asymptotic properties of the QML estimator for the univariate case. Since their QML estimation depends on the numerical optimization algorithm, Shephard (1993) and So, Li, and Lam (1997) developed simulated and standard expectation-maximization algorithms, respectively. However, the inefficiency of the QML estimator comes from the fact that $\log \chi^2(1)$ is highly right-skewed. To fix this issue, Kim, Shephard, and Chib (2002) approximated $\log \chi^2(1)$ by a mixture of normal distributions to carry out a Bayesian MCMC estimation. Instead, our estimation procedure improves the efficiency \textcolor{black}{using} the \textcolor{black}{penalized} OLS regression, as previously detailed.}

\subsection{Our proposed approach}

\textcolor{black}{In this section, we propose a new procedure for estimating high-dimensional stochastic volatility models.
Our approach starts from the measurement equation (\ref{eq:ss1}).
Instead of the state space form, we derive the VARMA representation of $\textcolor{black}{y_t^{\ell}}$ \textcolor{black}{to} apply the ideas of Hannan and Rissanen (1982) and Hannan and Kavalieris (1984) under the framework of a penalized OLS estimation.
Our approach consists of four steps and can be summarized as \textcolor{black}{follows:}}
\begin{itemize}
\item[\textbf{Step 1.}] \textcolor{black}{Consider a penalized OLS estimation to approximate the error terms in the VARMA representation;}
\item[\textbf{Step 2.}] \textcolor{black}{Using the approximated errors, obtain a regression-based estimator of $(c,\Phi)$};
\item[\textbf{Step 3.}] \textcolor{black}{Conditional on} the VARMA estimators, use an \textcolor{black}{ad hoc} estimator for $\Sigma_{\zeta}$ such that the corresponding estimator is positive-definite; 
\item[\textbf{Step 4.}] \textcolor{black}{Obtain the estimator of $\Gamma$.}
\end{itemize}

\medskip

\textcolor{black}{Let us now \textcolor{black}{detail} this four step procedure.} Since $\textcolor{black}{y_t^{\ell}}$ is the sum of a VAR(1) process and an i.i.d. noise by (\ref{eq:ss1}), the discussion of Granger and Morris (1976) suggests that $\textcolor{black}{y_t^{\ell}}$ has a VARMA(1,1) representation. 
By equations (\ref{eq:ss1}) and (\ref{eq:ss2}), we obtain
\[
\textcolor{black}{y_t^{\ell}}  = (I-\Phi)c + \Phi \textcolor{black}{y_{t-1}^{\ell}} + (\zeta_t + \eta_{t-1}) - \Phi \zeta_{t-1},
\]
which {can alternatively be written as
\begin{equation}
\textcolor{black}{y_t^{\ell}}  = (I-\Phi)c + \Phi \textcolor{black}{y_{t-1}^{\ell}}  + u_t + \Xi u_{t-1},
\label{eq:gm0}
\end{equation}
with $(u_t)$ \textcolor{black}{being} a $p$-dimensional white noise vector with moments $\Eb[u_t]=0, \; \text{Var}(u_t) = \Sigma_u, \; \Eb[u_t u_s^\top] = \mathbf{0} \mbox{ for } t \neq s$,} where $\Xi$ and $\Sigma_u$ are obtained by matching moments of $w_t = (\zeta_t + \eta_{t-1}) - \Phi \zeta_{t-1}$ and $w_t^* = u_t + \Xi u_{t-1}$. Using $\Eb[w_t w_t^\top]=\Eb[w_t^* w_t^{*\top}]$ and $\Eb[w_t w_{t-1}^\top]=\Eb[w_t^* w_{t-1}^{*\top}]$, the relationship between $(\Xi, \Sigma_u)$ and other parameters is given as follows:
\begin{align}
\Sigma_{\eta} + \Sigma_{\zeta} + \Phi \Sigma_{\zeta} \Phi^\top &= \Sigma_u + \Xi \Sigma_u \Xi^\top, 
\label{eq:gm1} \\
- \Phi \Sigma_{\zeta}  &= \Xi \Sigma_u.
\label{eq:gm2}
\end{align}
{$\Sigma_{\eta}$ and $\Sigma_{\zeta}$ can be deduced from $\Phi$, $\Xi$, and $\Sigma_u$ based on equations (\ref{eq:gm1}) and (\ref{eq:gm2}).
\textcolor{black}{Let $(x_t)$ denote the mean-subtracted process $x_t=y_t^{\ell} - \Eb[y_t^{\ell}]$.}
Assuming a stable and invertible model, $(x_t)$ has an AR($\infty$) representation:}
\begin{equation}\label{VAR_inf}
x_t = \sum_{i=1}^{\infty} \Psi_i x_{t-i} + u_t,
\end{equation}
\textcolor{black}{with $u_t$ defined in equation (\ref{eq:gm0}).}
Based on a penalized OLS estimation, we can obtain \textcolor{black}{an approximation of $u_t$ in the first step, denoted as $\widehat{u}^{(m)}_t$. The latter approximation depends on $m$: \textcolor{black}{we} empirically need to specify $m$ sufficiently large as a surrogate of $\infty$ in the summation in (\ref{VAR_inf}). Thus, for the sake of parsimony and to avoid the over-fitting issue, we assume sparsity among the $\Psi_i$\textcolor{black}{s}.} In the second step, we calculate the OLS estimator of $(\widehat{c}^*,\widehat{\Phi},\widehat{\Xi})$ by regressing $x_t$ on a constant, $x_{t-1}$, and \textcolor{black}{$\widehat{u}^{(m)}_{t-1}$}. For the third step, we start from the decomposition of the unconditional variance-covariance matrix of $x_t$, which is given by
\begin{equation}
\Sigma_x = \Sigma_{\alpha} + \Sigma_{\zeta}, 
\label{eq:deco}
\end{equation}
where $\Sigma_x = \Eb[\textcolor{black}{x_t x_t^\top}]$, $\Sigma_{\zeta} = \Eb[\zeta_t \zeta_t^\top]$, and $\Sigma_{\alpha} = \Eb[\alpha_t \alpha_t^\top]$ with
\[
\vec (\Sigma_{\alpha}) = [ I_{p^2} - (\Phi \otimes \Phi)]^{-1} \vec (\Sigma_{\eta}). 
\]
{Denoting the sample covariance matrix of $x_t$ and \textcolor{black}{$\widehat{u}^{(m)}_t$} by $S_x$ and $S_{\widehat{u}^{\textcolor{black}{(m)}}}$, respectively,
we obtain an estimator of $\Sigma_{\zeta}$ as}
\[
S_{\zeta} = - \frac{1}{2} \left[ \widehat{\Phi}^{-1} \widehat{\Xi} S_{\widehat{u}^{\textcolor{black}{(m)}}}
+ S_{\widehat{u}^{\textcolor{black}{(m)}}} \widehat{\Xi}^\top \widehat{\Phi}^{\top -1} \right],
\]
by \textcolor{black}{the sample analogous of the mean of $\Sigma_{\zeta}$ obtained by equation (\ref{eq:gm2}) and its transpose.}
As there is no guarantee for $S_{\zeta}$ and $S_x-S_{\zeta}$ to be \textcolor{black}{positive-definite} by the approach, we consider \textcolor{black}{ad hoc} estimators for $\Sigma_{\zeta}$ and $\Sigma_{\alpha}$ based on decomposition (\ref{eq:deco}). Finally, in the fourth step, we estimate $\Gamma$ by a correlation matrix of $y_t$. \\
\textcolor{black}{To summarize, our procedure can be broken down as follows:}
\begin{description}
\item[Step 1.] We approximate (\ref{VAR_inf}) as
\begin{equation}\label{approxm}
x_t = \sum_{i=1}^m \Psi_i x_{t-i} + u^{(m)}_t,
\end{equation}
with $u^{(m)}_t = u_t + \underset{i>m}{\sum}\Psi_iX_{t-i}$: under suitable parameter conditions, $\underset{i>m}{\sum}\Psi_iX_{t-i}$ is actually negligible when $m$ is large enough. Such conditions can be set in the same vein as {Assumptions (15.2.2)-(15.2.4) of L\"{u}tkepohl (2006) or Assumption 2.1(b) of Chang, Park, and Song (2006) for \textcolor{black}{the} VAR models; as Assumption 2 of Chang and Park (2002) for AR models: if we consider a univariate process, \textcolor{black}{based} on their latter assumption, $\Eb[|u^{(m)}_t-u_t|^r]=o(m^{-r})$ assuming the existence of the $r$-th moment of $u_t$. Under the sparsity assumption for the VAR($m$) coefficients, we consider the \textcolor{black}{penalized} OLS problem}
\begin{equation*}
\widehat{\Psi}_{1:m} = \mathop{\mbox{arg min}}_{ \Psi_{1:m} } 
 \big\{ \frac{1}{2T} \sum_{t=1}^T || x_t - \sum_{i=1}^m \Psi_i x_{t-i} ||^2_2
+ \textbf{pen}(\frac{\lambda_T}{T},\vec (\Psi_{1:m}))  \big\},
\end{equation*}
where $\textbf{pen}(\frac{\lambda_T}{T},\cdot) : \Rb^d \to \Rb$ is a coordinate-separable penalty applied to the coefficients $\Psi_{1:m}=[ \Psi_1 \; \cdots \Psi_m] \in \Mc_{p \times pm}(\Rb)$, $\lambda_T$ is the regularization parameter which depends on the sample size and enforces a particular type of sparse structure in the solution $\widehat{\Psi}_{1:m}$. \textcolor{black}{In vector form, $\text{vec}(\Psi_{1:m}) \in \Rb^{d}, d=mp^2$. In the asymptotic analysis detailed in Subsection \ref{first_step}, the dimension $d$ potentially diverges with the sample size $T$. In particular, this diverging property includes the case ``$m$ large and $p$ fixed",} which is pertinent when the objective is to suitably approximate $(u_t)$ by $(u^{(m)}_t)$. \\
\textcolor{black}{As the number of parameters $d$ increases with the sample size, we assume that the true parameter value is sparse, which refers to the condition that only $k < d$ elements of the true parameter are nonzero but allows the identities of these elements to be unknown. In other words, the true parameter contains a large number of zero coefficients}. Moreover, when $m$ is large, the sparse property is pertinent in the context of time series with autoregressive components. \textcolor{black}{Indeed, the most recent observations are likely to have} a higher-level effect on the current $(x_t)$ \textcolor{black}{in contrast to} older observations. \textcolor{black}{Consequently}, it is natural to assume that the parameters in $\Psi_i$ decay with $i$ and become negligible. \textcolor{black}{Since the set of non-zero coefficients is unknown, we rely on the penalty function $\textbf{pen}(\cdot,\cdot)$ to estimate it.} Importantly, the penalty function is non-differentiable at the origin to foster sparsity in the \textcolor{black}{estimator}. Furthermore, an additional merit for imposing sparsity is its ability to fix the so-called over-fitting issue: such a problem occurs when too many parameters must be estimated in light of the sample size, \textcolor{black}{which results} in poor out-of-sample performances. \textcolor{black}{Sparsity-based inference methods potentially fix this problem, as emphasized by, e.g., Belloni et al. (2013) or Ng (2013).}

\item[Step 2.] Let {$\widehat{u}^{(m)}_t = x_t - \overset{m}{\underset{i=1}{\sum}} \widehat{\Psi}_i x_{t-i}$.} \textcolor{black}{Conditional on} $\widehat{\Psi}_{1:m}$, we consider the regression
{\[
\textcolor{black}{y_t^{\ell}} = c^* + \Phi \textcolor{black}{y_{t-1}^{\ell}} + \Xi \widehat{u}^{(m)}_{t-1} + v_t,
\]}
where the parameters are $(c^*,\Phi,\Xi)$, and since we replace $u_{t-1}$ by {$\widehat{u}^{(m)}_{t-1}$}, $(v_t)$ is the error term for this auxiliary regression. The second step objective function is
\beqw
(\widehat{c}^*,\widehat{\Phi},\widehat{\Xi})|\widehat{\Psi}_{1:m} = \mathop{\mbox{arg min}}_{ (c^*,\Phi,\Xi) }  \big\{ \frac{1}{2T} \sum_{t=1}^T ||\textcolor{black}{y_t^{\ell}} -  \big(c^* + \Phi\textcolor{black}{y_{t-1}^{\ell}} + \Xi {\widehat{u}^{(m)}_{t-1}}  \big)||^2_2\big\},
\eeqw
such that we can obtain the estimator of $c$ by $\widehat{c} = (I -\widehat{\Phi})^{-1} \widehat{c}^*$. In this step, the second step parameter dimension is $p(1+2p)$.
\item[Step 3.] The estimators of $\Sigma_{\zeta}$ and $\Sigma_{\alpha}$ are deduced as
\begin{equation}\label{adhoc_estim}
\widehat{\Sigma}_{\zeta} = r S_x, \quad \widehat{\Sigma}_{\alpha} = (1-r) S_x,
\end{equation}
where $r$ is a constant satisfying $0 < r < 1$. This ad hoc method aims to treat the positive\textcolor{black}{-}definiteness of the estimators and to deal with the high-dimensionality issue of $\widehat{\Sigma}_{\zeta}$.
While we consider a naive decomposition based on equation (\ref{eq:deco}) for the former, we set $r = (\pi^2/2)(p^{-1} \textcolor{black}{\tr (S_{x})})^{-1}$ in (\ref{adhoc_estim}).
Here, $ \pi^2/2$ is the value of $\Eb[\zeta_{it}^2]$ when $\varepsilon_{it}$ follows the standard normal distribution. 
The ad hoc estimators yield $\tr (\widehat{\Sigma}_{\zeta}) = r \tr(S_x) = p \pi^2/2$ and $\tr (\widehat{\Sigma}_{\alpha}) =  \tr(S_x) - p \pi^2/2$.
Using such approach, we are able to estimate $\tr({\Sigma}_{\zeta})$ and $\tr({\Sigma}_{\alpha})$ with accuracy and consistency, respectively.
More importantly, the computational cost is negligible, compared to alternative estimators (e.g., \textcolor{black}{the} GMM type method) that would require a numerical optimization \textcolor{black}{with constraints on the positive-definiteness of $S_{\zeta}$ and $S_x-S_{\zeta}$.} 

\item[Step 4.] Estimate $\Gamma$ by a correlation matrix of $y_t$.
\end{description}

When the tuning parameter $\lambda_T$ shrinks to zero, Steps 1 and 2 reduce to the standard OLS estimation for low-dimensional VARMA models considered by Hannan and Rissanen (1982) and Hannan and Kavalieris (1984). Step 1 corresponds to a multivariate version of the AR($\infty$) representation of a log-GARCH model. %of Geweke (1986) and Pantula (1986), which is a special case of the exponential GARCH model of Nelson (1991). Thus, for a significantly large $m$ chosen a priori, (\ref{approxm}) would be a relevant approximation of a log-GARCH type process. Our approach avoids the use of computationally demanding methods such as MCMC and MCL.
Although Harvey, Ruiz, and Shephard (1994) applied the Kalman filter, its computational cost is non-negligible for large $p$, since the cost {evolves according to $O(T p^2)$} for storing covariance matrices of a $p \times 1$ state vector for all $t=1,\ldots,T$. 
For the estimators in \textbf{Step 3}, we may improve them by considering moment-matching methods using equations (\ref{eq:gm1}), (\ref{eq:gm2}), and $(\ref{eq:deco})$ with restrictions on the positive-definiteness of the estimators of $\Sigma_{\zeta}$ and $\Sigma_{\alpha}$.
However, we use the above fast {and efficient} method described in \textbf{Step 3} without the need of a numerical optimization procedure.
Finally, the fourth step can easily be adapted to a sparse correlation matrix setting, especially when the size $p/T$ is not negligible.

\medskip

We now introduce our setting for generating the volatility process. For a low-dimensional case, we can calculate the minimum mean square linear estimator (MMSLE) of $\alpha_t$ based on the full sample \textcolor{black}{$\bm{y}^\ell=(y_1^{\ell \top},\ldots,y_T^{\ell \top})^\top \in \Rb^{pT}$} by the state space smoothing algorithm.
In the high-dimensional case, we consider the multivariate version of Harvey (1998)\textcolor{black}{'s approach} with the vector form of (\ref{eq:ss1}) as follows:
\[
\textcolor{black}{\bm{y}^{\ell}} = \bm{c}^\dagger + \bm{\alpha} + \bm{\zeta},
\]
where $\bm{c}^\dagger = (\iota_T \otimes c) \in \Rb^{pT}$, $\bm{\alpha}=(\alpha_1^\top,\ldots,\alpha_T^\top)^\top \in \Rb^{pT}$, and $\bm{\zeta}=(\zeta_1^\top,\ldots,\zeta_T^\top)^\top\in \Rb^{pT}$.
By the model structure, the covariance matrix of $\bm{x}$ is given by
\[
V_x = V_{\alpha} + V_{\zeta},
\]
where 
\[
V_{\alpha} = \left( \begin{array}{cccccc}
\Sigma_{\alpha} & \Sigma_{\alpha} \Phi^\top & \Sigma_{\alpha} (\Phi^\top)^2 & \cdots & \Sigma_{\alpha} (\Phi^\top)^{T-2} & \Sigma_{\alpha} (\Phi^\top)^{T-1} \\
\Phi \Sigma_{\alpha} & \Sigma_{\alpha} & \Sigma_{\alpha} \Phi^\top & \cdots & \Sigma_{\alpha} (\Phi^\top)^{T-3} & \Sigma_{\alpha} (\Phi^\top)^{T-2} \\
\Phi^2 \Sigma_{\alpha} & \Phi \Sigma_{\alpha} & \Sigma_{\alpha} & \cdots & \Sigma_{\alpha} (\Phi^\top)^{T-4} & \Sigma_{\alpha} (\Phi^\top)^{T-3} \\
\vdots & \vdots & \vdots & \ddots & \vdots & \vdots \\
\Phi^{T-2} \Sigma_{\alpha} & \Phi^{T-3} \Sigma_{\alpha} & \Phi^{T-4} \Sigma_{\alpha} & \cdots & \Sigma_{\alpha} & \Sigma_{\alpha}  \Phi^\top \\
\Phi^{T-1} \Sigma_{\alpha} & \Phi^{T-2} \Sigma_{\alpha} & \Phi^{T-3} \Sigma_{\alpha} & \cdots & \Phi \Sigma_{\alpha} & \Sigma_{\alpha} 
\end{array} \right),
\]
and $V_{\zeta} = (I_T \otimes \Sigma_{\zeta})$.
Then\textcolor{black}{,} the MMSLE \textcolor{black}{can be calculated as follows:}
\begin{equation*}
\widetilde{\bm{\alpha}} = V_{\alpha} V_x^{-1} (\textcolor{black}{\bm{y}^\ell} - \bm{c}^\dagger) + \bm{c}^\dagger.
%= V_{\alpha} V_x^{-1} \bm{x} + (I_{Tp} - V_{\alpha} V_x^{-1} ) \bm{c}^\dagger.
\end{equation*}
As in Harvey (1998), the covariance matrix is deduced from the relationship $H_t = D_t \Gamma D_t$ such that the sample variance of the standardized variable of $y_{it}$ equals to one.
We consider the estimator as $\widetilde{H}_t = \widetilde{D}_t \widehat{\Gamma} \widetilde{D}_t$, where
\begin{equation*}
\widetilde{D}_t = \mbox{diag} \big( \widetilde{d}_{1t},\ldots, \widetilde{d}_{pt} \big), 
\quad
\widetilde{d}_{it} = \bar{d}_i \exp \left( \widetilde{x}_{it}/2 \right), 
\quad
\bar{d}_i  = \sqrt{ T^{-1} \sum_{t=1}^T y_{it}^2  \exp \left(  -\widetilde{x}_{it} \right) },
\end{equation*}
for $i=1,\ldots,p$. The standardized variables are defined as $\widetilde{z}_{it} = y_{it} / \widetilde{d}_{it}$, which implies $T^{-1} \sum_{t=1}^T \widetilde{z}_{it}^2 = 1$ by definition. We call our proposed parameterization ``penalized OLS-MSV''.

\subsection{Volatility forecasting}

We now provide the forecasts for variance-covariance based on our proposed method. The MMSLE for the $l$th-step-ahead forecast of $\alpha_T$ is given by
\begin{equation*}
\widehat{\alpha}_{T+l} = R_l V_x^{-1} ( \textcolor{black}{\bm{y}^\ell} - \bm{c}^{\dagger}) + c,
\end{equation*}
where $R_l = \left[ \Phi^{T+l-1} \Sigma_{\alpha} \;\; \Phi^{T+l-2} \Sigma_{\alpha} \; 
\cdots \; \Phi^{l} \Sigma_{\alpha} \right]$. The $l$th-step-ahead forecast of the covariance matrix is given by
\begin{equation*}
\widehat{H}_t = \widehat{D}_t \widehat{\Gamma} \widehat{D}_t,
\end{equation*}
where for $i=1,\ldots,p$,
\[
\widehat{D}_{T+l} = \mbox{diag} \big(\widehat{d}_{1,T+l},\ldots, \widehat{d}_{p,T+l} \big), 
\quad
\widehat{d}_{i,T+l} = \bar{d}_i \exp \left( \widehat{x}_{i,T+l}/2 \right).
\]
\textcolor{black}{
By the structure of $R_l$ and $V_x$, the inconsistency on the off-diagonal elements of the third-step estimator may affect the forecasts. We assess its applicability via the Monte
Carlo experiments performed in Section \ref{empirical_illustration}.}

Using the covariance forecasts, we can construct time-varying portfolios for the asset returns, in which \textcolor{black}{the} portfolio weights are determined by \textcolor{black}{past information}.
\textcolor{black}{To assess the relevance of the volatility models in terms of forecasts, we can obtain the value-at-risk (VaR) threshold of the portfolio return. This threshold is the negative of the 100$q$-th percentile of the portfolio return distribution, with $q$ small such as $q=0.01$, and may be used in the test procedure of Candelon et al. (2011). More precisely, the VaR threshold at time $t+1$ is given by $- \tau_q \sqrt{\widehat{h}_{t+1}^p}$, where $\widehat{h}_{t+1}^p=w_{t+1}^\top \widehat{H}_{t+1} w_{t+1}$, $w_{t+1}$ is the $p\times 1$ vector of portfolio weights at time $t+1$, and $\tau_q$ is the 100$q$-th percentile of the standard normal distribution or a historically simulated distribution of $\{ (w_t^\top y_t)/\sqrt{\widehat{h}_t^p}\}$.
Alternatively,} we can compare forecasting models directly based on the minimum-variance portfolio via the test suggested by Engle and Colacito (2006).

\section{Asymptotic properties}\label{asymptotic_theory}

In this section, we provide the asymptotic properties of the penalized two-step estimator. We show that the first step estimator satisfies the oracle property for the SCAD and MCP cases and when the number of parameters diverges with the sample size. Conditional on this first step sparse estimator, we derive the conditions for consistency and asymptotic normality of the second step estimator.

\subsection{First step penalized estimator $\widehat{\Psi}_{1:m}$}\label{first_step}

\textcolor{black}{In \textbf{Step 1}, we estimate the parameter $\theta = \text{vec}(\Psi_{1:m})$ with $d=mp^2$, the dimension that can diverge with the sample size $T$. Consequently, both dimension $d=d_T$ and parameter $\theta=\theta_T$ are indexed \textcolor{black}{hereafter} by $T$ to highlight the dependence of \textcolor{black}{$d$ and, thus, $\theta$}, with respect to $T$. More formally, we consider a sequence of parametric models
$\Pc_T:=\{\Pb_{\theta_T}, \,\theta_T \in \Theta_{1,T}\}$, $\Theta_{1,T}\subset \Rb^{d_T}$. We \textcolor{black}{denote the non-penalized loss function} by $\Gb_T: \Rb^{pT} \times \Theta_{1,T} \rightarrow \Rb$: the value $\Gb_T(\underline{y};\theta_T)$ with $\theta_T \in \Theta_{1,T}$ evaluates the quality of the ``fit'' for the realizations of $y_t$ for every $t=1,\cdots,T$ and under $\Pb_{\theta_T}$. The loss $\Gb_T(\underline{y};\theta)$ is associated to a continuous function $\ell : \Rb^{pT} \times \Theta_{1,T} \rightarrow \Rb$ that can be written as
\begin{equation*}
\Gb_T(\underline{y};\theta) := \cfrac{1}{T} \overset{T}{\underset{t=1}{\sum}} \frac{1}{2}\|x_t-\overset{m}{\underset{i=1}{\sum}}\Psi_ix_{t-i}\|^2_2 = \cfrac{1}{T} \overset{T}{\underset{t=1}{\sum}} \frac{1}{2}\|x_t-\Psi_{1:m} Z_{m,t-1}\|^2_2 := \cfrac{1}{T} \overset{T}{\underset{t=1}{\sum}} \ell(y_s,s\leq t;\theta),
\end{equation*}
where $x_t$ corresponds to the vector of continuous transforms of $\log(y^2_{it})$ (vector of mean-subtracted series), $\Psi_{1:m}=(\Psi_1,\cdots,\Psi_m)\in \Mc_{p \times pm}(\Rb)$, and $Z_{m,t-1}=(x^\top_{t-1},\cdots,x^\top_{t-m})^\top\in \Rb^{pm}$. For every $T$, we assume there exists a unique pseudo-true parameter value $\theta_{0,T}$: for every $T$, the function $\theta_T \mapsto \Eb[\ell(y_s,s\leq t;\theta_T)]$ is uniquely minimized on $\Theta_{1,T}$ at $\theta_T = \theta_{0,T}$ and the first-order conditions are satisfied, \textcolor{black}{that is}, $\Eb[ \nabla_{\theta_T}\Gb_T(\underline{y};\theta_{0,T})] = 0$. In light of the possibly explosive number of parameters for a given $T$, $\theta_{0,T}$ is assumed sparse so that the size of the true support $k_T = \text{card}(\Ac_T)$, with $\Ac_T:=\{i=1,\cdots,d_T: \theta_{0,i,T}\neq 0\}$, also diverges with $T$. To estimate the latter support, we rely on the penalty function $\textbf{pen}(\frac{\lambda_T}{T},.)$, which is assumed coordinate-separable, \textcolor{black}{that is}, $\textbf{pen}(\frac{\lambda_T}{T},\theta_{T}) = \sum^{d_T}_{i=1}\pp(\frac{\lambda_T}{T},|\theta_{i,T}|)$. Then, the penalized problem becomes
\begin{equation} \label{obj_crit_gen}
\widehat{\theta}_T = \underset{\theta_T \in \Theta_{1,T}}{\arg \; \min} \;  \big\{ \Gb_T(\underline{y};\theta_T) + \overset{d_T}{\underset{i=1}{\sum}}\pp(\frac{\lambda_T}{T},|\theta_{i,T}|) \big\}.
\end{equation}}
For the penalty function, we consider the convex penalty LASSO $\pp(\lambda,|\theta|) = \lambda |\theta|$ of \textcolor{black}{Tibshirani (1996)} and the non-convex penalties SCAD and MCP. The SCAD of Fan and Li (2001) is defined as
\beqw
\pp(\lambda,|\theta|) = \begin{cases}
\lambda |\theta|,&\text{for} \;  |\theta|\leq\lambda,\\
-\frac{1}{(2(a-1))}(\theta^2-2 a\lambda|\theta|+\lambda^2), &\text{for} \;\lambda\leq |\theta|\leq a \lambda,\\
(a+1)\lambda^2/2, &\text{for} \;  |\theta|> a\lambda,
\end{cases}
\eeqw
where $a>2$. The MCP due to Zhang et al. \textcolor{black}{(2010) is defined} for $b>0$ as
\beqw
\pp(\lambda,|\theta|)  = \lambda\Big[|\theta|-\frac{\theta^2}{2 b\lambda}\Big] \mathbf{1}_{ \big\{0 \leq |\theta|<b\lambda\big\}} + \lambda\frac{b}{2} \mathbf{1}_{ \big\{|\theta|\geq b\lambda\big\}}.
\eeqw

\medskip

All assumptions we relied on for the large sample analysis are reported in Section \ref{appendix_proof} of the Appendix. In particular, the sparsity assumption states that the true parameter vector is sparse, that is\textcolor{black}{,} the cardinality of the true sparse support \textcolor{black}{$\Ac_T$} is of size $k_T < d_T$. We assume stability of the VARMA(1,1) process $(x_t)$ to apply the large sample theory for stationary processes. Finally, we assume suitable regularity conditions for both the non-penalized loss and the penalty function.

\medskip

We first show the existence of the penalized estimator $\widehat{\theta}_T$ for the three aforementioned penalty cases.
\begin{theorem}\label{bound_prob_asym}
Under Assumptions \ref{assumption_sparsity}-\ref{assumption_second_derivative} given in Appendix \ref{appendix_proof}, assume that the penalty function satisfies Assumptions \ref{assumption_folded}-(i),(ii),(iii) in Appendix \ref{appendix_proof} for the SCAD and MCP cases and satisfies $\lambda_T=o(T)$ for the LASSO \textcolor{black}{case; then, under} the scaling \textcolor{black}{behavior} $d^2_T=o(T)$, there is a local optimum $\widehat{\theta}_{\textcolor{black}{T}}$ of (\ref{obj_crit_gen}) satisfying
\beqw
\|\widehat{\theta}_{\textcolor{black}{T}}-\theta_{0,\textcolor{black}{T}}\|_2 = O_p\big(\sqrt{d_T} ( T^{-1/2} + R_T ) \big),
\eeqw
where $R_T = A_{1,T}$ for the SCAD and MCP defined in Assumption \ref{assumption_folded}-(ii), and $R_T = \frac{\lambda_T}{T}$ for the LASSO.
\end{theorem}
\begin{remark}
For a suitable choice of $\lambda_T$, we would obtain a $\sqrt{T/d_T}$-consistent $\widehat{\theta}_{\textcolor{black}{T}}$. A diverging $d_T$ requires the use of an explicit norm: due to norm equivalences, some constants may appear \textcolor{black}{that} may depend on the size $d_T$ and, thus, on $T$.
\end{remark}

Our second result is dedicated to the oracle property: we show that the penalization procedure in problem (\ref{obj_crit_gen}) asymptotically recovers the true underlying sparse subset $\Ac_{\textcolor{black}{T}}$ and the nonzero estimated coefficients are normally distributed. We prove the oracle property for the SCAD and MCP only: these penalty functions are non-convex, a key property that enables to relax the incoherence/irrepresentable condition and/or avoid the specification of adaptive weights. \textcolor{black}{The} incoherence/irrepresentable condition - see inequality (3) of Zou (2006) regarding the irrepresentable condition; see Loh and Wainwright (2017) regarding the incoherence condition - is necessary to prove the oracle property for the LASSO: such condition is nontrivial and difficult to empirically verify. Rather than assuming the incoherence/irrepresentable condition, Zou (2006) proposed the adaptive LASSO: stochastic weights are specified in the LASSO penalization to alter the convergence rate of the regularization parameter $\lambda_T$; such weights depend on a first step $\sqrt{T/d_T}$-consistent estimator, typically an \textcolor{black}{non-penalized} OLS estimator: the adaptive LASSO is consequently a two-step procedure. In the same vein, Poignard (2020) specified adaptive weights in the Sparse Group LASSO penalty - $\ell_1+\ell_1/\ell_2$ penalty - since the convexity of the $\ell_1$ and $\ell_1/\ell_2$ norms prevents from satisfying the oracle property. The key advantage of non-convex penalization is the relaxation of the incoherence/irrepresentable condition and avoids a two-step procedure as in the adaptive LASSO.

\begin{theorem}\label{oracle_theorem}
Under Assumptions \ref{assumption_sparsity}-\textcolor{black}{\ref{assumption_array}} \textcolor{black}{given in Appendix \ref{appendix_proof}}, assume $d^3_T = o(T)$, $\lambda_T=o(T)$; then, the $\sqrt{\frac{T}{d_T}}$-consistent local estimator $\widehat{\theta}_{\textcolor{black}{T}}$ of Theorem \ref{bound_prob_asym} satisfies
\beqw
\begin{array}{llll}
\underset{T \rightarrow\infty}{\lim} \;\Pb(\widehat{\Ac}_{\textcolor{black}{T}} = \Ac_{\textcolor{black}{T}}) = 1, \;\;\text{and} &&\\
&&\\
\sqrt{T}Q_T\Vb^{-1/2}_{\Ac_{\textcolor{black}{T}}\Ac_{\textcolor{black}{T}}} \big(\widehat\theta_{\textcolor{black}{T}} - \theta_{0,\textcolor{black}{T}} \big)_{\Ac_{\textcolor{black}{T}}} \overset{d}{\underset{T \rightarrow\infty}{\longrightarrow}} \Nc_{\Rb^r} \big(0,\mathbb{C} \big),&&
\end{array}
\eeqw
where $\Vb_{\Ac_{\textcolor{black}{T}}\Ac_{\textcolor{black}{T}}} :=  \big(\Hb^{-1} \Mb \Hb^{-1} \big)_{\Ac_{\textcolor{black}{T}}\Ac_{\textcolor{black}{T}}}$, \textcolor{black}{$\Hb := \Eb[\partial^2_{\theta_{k,\textcolor{black}{T}}\theta_{l,\textcolor{black}{T}}}\ell(y_s, s \leq t;\theta_{0,\textcolor{black}{T}})]_{1 \leq k,l \leq d_T}$ and $\Mb := \Eb[\nabla_{\theta_{\textcolor{black}{T}}}\ell(y_s,s\leq t;\theta_{0,\textcolor{black}{T}}) \nabla_{\theta^\top_{\textcolor{black}{T}}}\ell(y_s,s\leq t;\theta_{0,\textcolor{black}{T}})]$}, where $\nabla_{\theta_{\textcolor{black}{T}}}\ell(y_s,s\leq t;\theta_{0,\textcolor{black}{T}})= \big(Z_{m,t-1}\otimes \{x_t-\Psi_{0,1:m} Z_{m,t-1}\} \big)$ and $\nabla^2_{\theta_{\textcolor{black}{T}} \theta^\top_{\textcolor{black}{T}}}\ell(y_s, s \leq t;\theta_{0,\textcolor{black}{T}}) =  \big(Z_{m,t-1}Z^\top_{m,t-1}\otimes I_p \big)$, and $Q_T$ is a $r \times \normalfont\text{card}(\Ac_{\textcolor{black}{T}})$ matrix satisfying $ Q_T Q^\top_T \overset{\Pb}{\underset{T \rightarrow \infty}{\longrightarrow}} \mathbb{C}$ with $\Cb$ \textcolor{black}{being} a $r \times r$ symmetric \textcolor{black}{positive-definite} and deterministic matrix.
\end{theorem}
{\begin{remark} This result deserves a few comments:
\begin{itemize}
\item[(i)] The scaling behavior $(d_T,T)$ is given as $d^3_T = o(T)$\textcolor{black}{. T}his is because the third order term in the Taylor expansion vanishes for the least squares loss. If we consider a non-linear-based non-penalized loss, this rate would become $d^5_T=o(T)$, as in Fan and Peng (2004) or Poignard (2020).
\item[(ii)] The cardinality of the true support $\Ac_{T}$ denoted by $k_T$ also diverges with the sample size\textcolor{black}{. Thus,} the dimension of $(\widehat\theta_{\textcolor{black}{T}} - \theta_{0,\textcolor{black}{T}})_{\Ac_{\textcolor{black}{T}}}$ is diverging. This motivates the introduction of the matrix $Q_T$ to obtain a finite dimensional Gaussian distribution.
\end{itemize}
\end{remark}}

{The first step estimator is deduced from the truncated VAR process (\ref{approxm}): a VAR($m$) is fitted to obtain $(\widehat{u}^{(m)}_t)$ as an approximation of $(u_t)$. In this context, the specification of a diverging number of parameters is relevant to correctly approximate $(u_t)$. \textcolor{black}{When} $p$ is fixed and $m:=m_T \rightarrow \infty$, our scaling condition $d^3_T=o(T)$ for the oracle property is identical to condition (15.2.5) of Proposition 15.1. (result of Lewis and Reinsel, 1985) of L\"{u}tkepohl (2006). However, our setting does not enable to simultaneously distinguish $m_T \rightarrow \infty$ \textcolor{black}{and} $p:=p_T \rightarrow \infty$. Furthermore, we may derive a finite sample and explicit upper bound for the approximation error $\frac{1}{T}\overset{T}{\underset{t=1}{\sum}}\|\widehat{u}^{(m)}_t-(u_t+\underset{i > m}{\sum}\Psi_ix_{t-i})\|_2$ in the same spirit as in Proposition 2.4. of Wilms, Basu, Bien, and Matteson (2021), who considered an approximation for VARMA(p,q) processes and relied on a LASSO penalization. The derivation of such bound would require additional assumptions on the non-convex penalty functions - such as the $\mu$-amenable assumption as in Loh and Wainwright (2017) - and the derivation of an exponential bound over $\nabla_{\theta_{\textcolor{black}{T}}}\Gb_T(\underline{y};\theta_{\textcolor{black}{T}})$. We leave this topic for future research.}

\subsection{Second step estimator $(\widehat{c}^*,\widehat{\Phi},\widehat{\Xi})$}

We consider the large sample properties of the second step estimator $\widehat{\gamma} = (\widehat{c}^{*\top},\text{vec}(\widehat{\Phi})^\top,\text{vec}(\widehat{\Xi})^\top)^\top$, which is of size $d_2 = p(1+2p)$ assumed fixed. \textcolor{black}{Conditional on} $\widehat{\theta}_{\textcolor{black}{T}}$, we consider a second step loss function $\Lb_T$ from $\Rb^{pT} \times \Theta_2$ to $\Rb$ with $\Theta_2 \subset \Rb^{p(1+2p)}$ compact, and $\Lb_T(\underline{y};\widehat{\theta}_{\textcolor{black}{T}},\gamma)$ is the empirical loss associated to a continuous function $f: \Rb^{pT} \times \Theta_{1,\textcolor{black}{T}} \times \Theta_2 \rightarrow \Rb$, \textcolor{black}{that is},
\beqw
\Lb_T(\underline{y};\widehat{\theta}_{\textcolor{black}{T}},\gamma) =  \frac{1}{T} \overset{T}{\underset{t=1}{\sum}}\frac{1}{2}\|\textcolor{black}{y_t^{\ell}}-c^*-\Phi\textcolor{black}{y_{t-1}^{\ell}}-\Xi {\widehat{u}^{(m)}_{t-1}}\|^2_2  := \frac{1}{T} \overset{T}{\underset{t=1}{\sum}} \| \textcolor{black}{y_t^{\ell}}-\Gamma K_{t-1}(\widehat{\theta}_{\textcolor{black}{T}})\|^2_2:=\frac{1}{T} \overset{T}{\underset{t=1}{\sum}} f(y_s,s\leq t; \widehat{\theta}_{\textcolor{black}{T}},\gamma),
\eeqw
where $K_{t-1}(\widehat{\theta}_{\textcolor{black}{T}}) = (1,\textcolor{black}{y_{t-1}^{\ell\top}},{\widehat{u}^{(m)\top}_{t-1}})^\top \in \Rb^{1+2p}$ and $\Gamma = (c^*,\Phi,\Xi) \in \Mc_{p\times (1+2p)}(\Rb)$. The dependence with respect to the first step estimator is through {$\widehat{u}^{(m)}_t$}. The problem of interest is
\beqw
\widehat{\gamma} = \underset{\gamma\in\Theta_2}{\arg\;\min}\;\big\{\Lb_T(\underline{y};\widehat{\theta}_{\textcolor{black}{T}},\gamma) \big\}.
\eeqw
The function $\gamma \rightarrow \Eb[f(y_s,s\leq t; \theta_{0,\textcolor{black}{T}},\gamma)]$ is assumed to be uniquely minimized at $\gamma = \gamma_0$\textcolor{black}{; the} true parameter vector $\gamma_0 = (c^{*\top}_0,\text{vec}(\Phi_0)^\top,\text{vec}(\Xi_0)^\top)^\top$. 
\begin{theorem}\label{consistency_second}
Under the conditions of Theorem \ref{bound_prob_asym}, under Assumptions \ref{assumption_first_2ndstep}- \textcolor{black}{\ref{assumption_third_cross_1_2ndstep}} \textcolor{black}{in Appendix \ref{appendix_proof}}, the sequence of second step estimators $\widehat{\gamma}$ satisfies 
\beqw
\|\widehat{\gamma}-\gamma_0\| = O_p(\frac{1}{\sqrt{T}}).
\eeqw
\end{theorem}
\begin{remark}
To derive this explicit convergence rate, the moment \textcolor{black}{conditions in} Assumptions \ref{assumption_second_cross_2ndstep} and \textcolor{black}{\ref{assumption_third_cross_1_2ndstep}} are key to control for the first step estimator when the number of parameters is diverging. The framework can potentially be extended to a diverging $d_2:=d_{2,T}$, at the expense of more complicated assumptions. To keep our asymptotic arguments simple, we assumed $d_2$ \textcolor{black}{to be} fixed.  
\end{remark}
We now derive the asymptotic distribution of the second step estimator $\widehat{\gamma}$ \textcolor{black}{conditional on} $\widehat{\theta}_{\textcolor{black}{T}}$ \textcolor{black}{whose elements belong to $\Ac_{\textcolor{black}{T}}$}. 
\begin{theorem}\label{distr_second}
Under the assumptions of Theorem \ref{oracle_theorem}, assume $d^3_T=o(T)$ and $\widehat{\theta}_{\textcolor{black}{T}}$ is the oracle estimator, under the conditions of Theorem \ref{consistency_second} and under \textcolor{black}{Assumptions \ref{assumption_varcov_2ndstep}-\ref{assumption_third_cross_2_2ndstep}} in Appendix \ref{appendix_proof}, then
\beqw
\sqrt{T}
\big(\widehat{\gamma}-\gamma_0\big)
\underset{T \rightarrow\infty}{\overset{d}{\longrightarrow}} \Nc_{\Rb^{d_2}} \big(\mathbf{0},\Vb_{\gamma} \big),
\eeqw
with $\Vb_{\gamma} = \textcolor{black}{\Ub^{-1}}\textcolor{black}{\Upsilon_{\gamma \Ac_T}}\Vb_{\Ac_{\textcolor{black}{T}}\Ac_{\textcolor{black}{T}}} \textcolor{black}{\Upsilon^\top_{\gamma \Ac_T}}  \textcolor{black}{\Ub^{-1}} +  \textcolor{black}{\Ub^{-1}} \textcolor{black}{\Wb} \textcolor{black}{\Ub^{-1}} \textcolor{black}{+\Ub^{-1} \Upsilon_{\gamma \Ac_T} \Hb^{-1}_{\Ac_T\Ac_T}\Jb_{\Ac_T\gamma}\Ub^{-1} + \Ub^{-1} \Jb^\top_{\Ac_T\gamma}  \Hb^{-1}_{\Ac_T\Ac_T}\Upsilon^\top_{\gamma \Ac_T}\Ub^{-1}}$,
where $\Vb_{\Ac_{\textcolor{black}{T}}\Ac_{\textcolor{black}{T}}},\textcolor{black}{\Hb_{\Ac_T\Ac_T}}$ are defined in Theorem \ref{oracle_theorem}, $\textcolor{black}{\Ub:=\Eb[ \big(I_p\otimes K_{t-1}(\theta_{0,\textcolor{black}{T}})K_{t-1}(\theta_{0,\textcolor{black}{T}})^\top \big)]}$, \\$\textcolor{black}{\Wb:=\Eb[\big(K_{t-1}(\theta_{0,\textcolor{black}{T}})\otimes \big\{\textcolor{black}{y_{t}^{\ell}}-\Gamma_0K_{t-1}(\theta_{0,\textcolor{black}{T}})\big\}\big)\big(K_{t-1}(\theta_{0,\textcolor{black}{T}})\otimes \big\{\textcolor{black}{y_{t}^{\ell}}-\Gamma_0K_{t-1}(\theta_{0,\textcolor{black}{T}})\big\}\big)^\top]}$,\\ $\textcolor{black}{\Upsilon_{\gamma \Ac_T} := \Eb[\partial^2_{\gamma_l \theta_{k,T}}f(y_s,s\leq t;\theta_{0,\textcolor{black}{T}},\gamma_0)]_{1\leq l \leq d_2, k \in \Ac_T}}$ \textcolor{black}{and $\Jb_{\Ac_T\gamma} = \Eb[\nabla_{\theta_{\Ac_T}}\ell(y_s,s\leq t;\theta_{0,T})\nabla_{\gamma^\top}f(y_s,s\leq t;\theta_{0,T},\gamma_0)]$}.
\end{theorem}

\begin{remark} The following comments can be emphasized:
\begin{itemize}
\item[(i)] The effect of the first step estimator is explicitly provided. It affects the variance of the second-step estimator through $\textcolor{black}{\Ub^{-1}}\textcolor{black}{\Upsilon_{\gamma \Ac_T}}\Vb_{\Ac_{\textcolor{black}{T}}\Ac_{\textcolor{black}{T}}} \textcolor{black}{\Upsilon^\top_{\gamma \Ac_T}}  \textcolor{black}{\Ub^{-1}}$ \textcolor{black}{and $\Ub^{-1} \Upsilon_{\gamma \Ac_T} \Hb^{-1}_{\Ac_T\Ac_T}\Jb_{\Ac_T\gamma}\Ub^{-1}$}.
\item[(ii)] The key difficulty is to establish that $\nabla^2_{\gamma\theta^\top_{\textcolor{black}{T}}}\Lb_T(\underline{y};\widehat{\theta}_{\textcolor{black}{T}},\gamma_0)$ converges in probability to some deterministic counterpart while controlling for the diverging dimension of $\widehat{\theta}_{\textcolor{black}{T}}$\textcolor{black}{, justifying} the \textcolor{black}{moment conditions of Assumption \ref{assumption_third_cross_2_2ndstep}}.
\end{itemize}
\end{remark}

The third step estimator is accurate in the sense that $p^{-1}\tr(\widehat{\Sigma}_{\zeta})$ always takes the true value $\Eb[p^{-1}\tr(\widehat{\Sigma}_{\zeta})]=\pi^2/2$ under the Gaussian assumption. To improve the estimator, we can consider the structure $\Sigma_{\zeta} = \sigma_{\zeta}^2 P_{\zeta}$, where $P_{\zeta}$ is a correlation matrix and $\sigma_{\zeta}^2$ is the variance for non-Gaussian assumption.
Neglecting the computational costs, we may estimate \textcolor{black}{positive-definite} $\Sigma_{\zeta}$ and $\Sigma_{\alpha}$ given $\widehat{\Phi}$ and $\widehat{\Xi}$ under the restrictions discussed below equation (\ref{eq:deco}).

\section{Empirical analysis}\label{empirical_illustration}

\subsection{Simulation experiment}

In this section, we empirically investigate the ability of the proposed penalization method to better capture complex variance-covariance processes. We simulate the $p$-dimensional stochastic process $(\eps_t)$ based on two data generating processes (DGP\textcolor{black}{s}): the multivariate ARCH and the BEKK processes. For the multivariate ARCH with $q^*$ lags - M-ARCH($q^*$) in the rest of the paper - case, we consider the DGP
\beqw
\left\{\begin{array}{llll}
\eps_t& = & H^{1/2}_t\eta_t,\\
H_t& = &\Omega + \overset{q^*}{\underset{k = 1}{\sum}} (I_p\otimes\eps^\top_{t-k} ) A_k (I_p\otimes\eps_{t-k}),
\end{array}\right.
\eeqw
where $q^*$ is the number of lagged matrices being functions of $\eps_{t-k}$ and the $p^2 \times p^2$ square matrices $A_k$ satisfy the stationarity conditions of Theorem 2 of Boussama (2006) together with the positivity condition given by Gouri\'eroux (1997). We generate the diagonal elements of $A_k$ from a uniform distribution $\Uc([0.01,0.05])$ and the off-diagonal ones from $\Uc([-0.01,0.01])$ under the ordering constraint $\forall k \geq 2, \forall i,j, \, |A_{k,ij}| \leq |A_{k-1,ij}|$. Note that these coefficients are more constrained (i.e.\textcolor{black}{,} closer to zero) when the dimension $p$ increases. As for the matrix $\Omega$, the diagonal and off-diagonal elements are simulated from $\Uc([0.1,0.2])$ and $\Uc([-0.01,0.01])$, respectively. As for the BEKK process, the DGP is based on 
\beqw
\left\{\begin{array}{llll}
\eps_t& = & H^{1/2}_t\eta_t,\\
H_t& = &\Omega + A \eps_{t-1}\eps^\top_{t-1} A^\top + B H_{t-1} B^\top,
\end{array}\right.
\eeqw
where $A,B$ are $p \times p$ matrices, satisfying the stationarity constraint $\|D^+_p\{ \big(A \otimes A \big) +  \big(B \otimes B \big)\} D_p\|_s < 1$, \textcolor{black}{and $D_p$ and $D^+_p$ are the duplication matrix and elimination matrix, respectively} (see Subsection 11.3 ``Stationarity of VEC and BEKK Models" of Francq and Zako\"ian (2010) for the stationarity condition and Remark 11.1 for the definition of the latter matrices). The entries of $A$ and $B$ are generated from the uniform distribution $\Uc([-0.8,0.8])$. The matrix $\Omega$ is generated as in the M-ARCH($q^*$) case. Unlike the M-ARCH($q^*$) case, the BEKK dynamic includes an autoregressive component through $B$, which motivated the use of larger lags when estimating our proposed parameterizations. In both proposed dynamics, we initialize the observations $(\eps_{k},\cdots,\eps_{1})$ with centered and unit variance multivariate Gaussian distribution, where $k = q^*$ in the M-ARCH model and $k = 1$ in the BEKK \textcolor{black}{model}. \textcolor{black}{Further}, conditional on the past $k$ observations, we generate $H_t$ and, thus, $\eps_t$ according to a centered multivariate Gaussian distribution with variance-covariance $H_t$.

\medskip

We consider the problem sizes, $p = 15, 50, 100$, and $T = 800$ observations for each of them. For the M-ARCH($q^*$)-based data generating process, we considered $q^*=2$ when $p=15$ and $q^*=1$ when $p=50, 100$. \textcolor{black}{Subsequently}, we propose to compare the true variance-covariance processes - BEKK and M-ARCH($q^*$) - and the estimated ones through our proposed MSV model and the scalar DCC together with the constant correlation model (CCC). The estimation of the DCC model is based on the classic two-step Gaussian QMLE, where the marginal conditional volatility processes are specified as GARCH(1,1) and a correlation targeting procedure is applied in the second step, providing an estimated trajectory $\widehat{H}^{\text{dcc}}_t$. The CCC is estimated thanks to a joint estimation of the GARCH(1,1) parameters and correlation parameters through a Gaussian QML, which provides an estimated process $\widehat{H}^{\text{ccc}}_t$. More details on the DCC and CCC can be found in Appendix \ref{competing_MGARCH}.

\medskip

Regarding our proposed variance-covariance dynamic, the penalized OLS-MSV, denoted as $\widehat{H}^{\text{ols,scad}}_t$ for the SCAD OLS-MSV, $\widehat{H}^{\text{ols,mcp}}_t$ for the MCP OLS-MSV, and the non-penalized version of the OLS-MSV denoted as $\widehat{H}^{\text{ols}}_t$, we considered two settings depending on the DGP. In the M-ARCH($q^*$) case, we set the number of lags in Step 1 in (\ref{approxm}) as $m = 10$ when $p=15$ and set \textcolor{black}{it} as $m=5$ for a dimension $p=50, 100$. In the BEKK case, due to the autoregressive nature of the latter dynamic, more lags were specified: we selected $m=30$ (resp. $m=15$, resp. $m=5$) when $p=15$ (resp. $p=50$, resp. $p=100$). For both DGPs, the correlation matrix of Step 4 is estimated as the sample correlation matrix estimator. In the SCAD and MCP cases, the coefficients $a$ and $b$ are set as $3.5$ - a value close to the optimal one as in Fan and Li (2001) - and $3$\textcolor{black}{, respectively}.

\medskip

We compare the true variance-covariance and the estimated variance-covariance processes through the aforementioned models. To do so, we specify a matrix distance, namely, the Frobenius norm, defined as $ || A - B ||_F := \sqrt{\tr((A - B)^\top(A - B))} $.
We compute the previous norm for each $t$ and for $A = H_t$ and 
$B \in  \big\{\widehat{H}^{\text{dcc}}_t,\widehat{H}^{\text{ccc}}_t,\widehat{H}^{\text{ols}}_t,\widehat{H}^{\text{ols,scad}}_t,\widehat{H}^{\text{ols,mcp}}_t \big\}$. We take the average of those quantities over $T=800$ periods of time. Since we repeat this experiment $100$ times, this provides an average gap for all those simulations. 

\medskip

By a cross-validation (CV) procedure - see, \textcolor{black}{for example}, Hastie \textcolor{black}{et} al. (2015, Chap. 2) - we selected the regularization parameter and emphasize that the standard CV developed for i.i.d. data \textcolor{black}{cannot} be used in our time series framework. To fix this issue, we used the hv-CV procedure devised by Racine (2000), which consists in leaving a gap between the test sample and the training sample, on both sides of the test sample.

\medskip

The average difference results are reported in Table \ref{simulation_average_march} for the M-ARCH($q^*$)\textcolor{black}{-based} DGP and Table \ref{simulation_average_bekk} for the BEKK-based DGP. First, our proposed method provides better in-sample results in terms of accuracy compared to standard MGARCH models. The results are closer to each other in the BEKK\textcolor{black}{-based} DGP case, essentially due to the presence of an autoregressive term, which is a priori in \textcolor{black}{favor} of the DCC/CCC model. Interestingly, our results \textcolor{black}{emphasize} the gain in considering a penalized MSV, especially when the dimension increases.

\subsection{Application to real data}

To assess the relevance of the proposed penalized method, we propose a real data experiment, where we focus on direct out-of-sample evaluation methods, which allow for \textcolor{black}{pair-wise} comparisons. They test whether some of the variance-covariance models provide better forecasts in terms of portfolio volatility behavior. Following the methodology of Engle and Colacito (2006), we develop a mean-variance portfolio approach to test the $H_t$ forecasts. Intuitively, if a conditional covariance process is misspecified, the minimum variance portfolio should emphasize such a shortcoming, compared to other models. \textcolor{black}{Here}, consider an investor who allocates a fixed amount between $p$ stocks, according to a minimum-variance strategy and independently at each time $t$:
\beq\label{portprob}
\min_{w_t} \; w^{\top}_tH_tw_t, \;\;\text{s.t.} \;\;\iota^{\top} w_t = 1,
\eeq
where $w_t$ is the $p \times 1$ vector of portfolio weights chosen at (the end of) time $t-1$, $\iota$ is a $p \times 1$ vector of $1$, and $H_t$ is the estimated conditional covariance matrix of the asset returns at time $t$. The solution of (\ref{portprob}) is given by the global minimum variance portfolio
$w_t = H^{-1}_t \iota/\iota^{\top}H^{-1}_t\iota$. Engle and Colacito (2006) \textcolor{black}{showed} that the realized portfolio volatility is the smallest when the variance-covariance matrices are correctly specified. \textcolor{black}{Consequently}, if wealth is allocated using two different dynamic models $i$ and $j$, whose predicted covariance matrices are $(H^{i}_t)$ and $(H^{j}_t)$, the strategy providing the smallest portfolio variance will be considered as the best. To do so, we consider a sequence of minimum variance portfolio weights $(w_{i,t})$ and $(w_{j,t})$, depending on the model. \textcolor{black}{Further}, we consider a distance based on the difference of the squared returns of the two portfolios, defined as
$u_{ij,t} = \left\{w^{\top}_{i,t} \epsilon_t\right\}^2 - \left\{w^{\top}_{j,t} \epsilon_t\right\}^2$.
The portfolio variances are the same if the predicted covariance matrices are the same. Thus, we test the null hypothesis
$\Hc_0: \; \mathbb{E}\left[u_{ij,t}\right] = 0$ by the Diebold and Mariano (1995) test. It consists of a least squares regression using HAC standard errors, given by
$u_{ij,t} =\alpha + \epsilon_{u,t}$, $\mathbb{E}[\epsilon_{u,t}]=0$, and we test $ \mathbf{H0}: \alpha = 0$.
If the mean of $u_{ij,t}$ is significantly positive (resp. negative), the forecasts given by the covariance matrices of model $j$ (resp. $i$) are preferred. \textcolor{black}{Following Engle and Colacito (2006), we compute the test statistic
\beqw
\widehat{\text{DM}}_{ij} = \frac{\sqrt{h}\,\overline{u}_{ij}}{\sqrt{\widehat{\text{Var}}(\sqrt{h}\,\overline{u}_{ij})}}, \;\overline{u}_{ij} =\frac{1}{h} \overset{h}{\underset{t=1}{\sum}}u_{ij,t}, 
\eeqw
with $h$ \textcolor{black}{being} the number of one-period ahead forecasts and $\sqrt{\widehat{\text{Var}}(\sqrt{h}\,\overline{u}_{ij})}$ is a \textcolor{black}{heteroscedasticity} and autocorrelation consistent
estimator of the asymptotic variance of $\sqrt{h}\,\overline{u}_{ij})$. In particular, $\widehat{\text{DM}}_{ij} \overset{d}{\longrightarrow} \Nc_{\Rb}(0,1)$ under $\textbf{H0}$.}

\medskip

We run the latter test to compare the scalar DCC (DCC), the Orthogonal GARCH (O-G), the BEKK (BEKK), and our OLS-MSV method (MSV) together with its penalized counterpart (denoted as MSV-SCA and MSV-MCP for the SCAD and MCP, respectively). We also consider the adaptive LASSO, denoted as MSV-AL, as an additional convex penalization technique developed by Zou (2006). In that case, we selected $\delta=3$ (resp. $\delta=4$) for the low-dimensional (resp. high-dimensional) portfolio as the power entering the stochastic weights and the first step estimator is the \textcolor{black}{non-penalized} OLS estimator, following Zou (2006). The definitions of the BEKK and O-GARCH processes are in Appendix \ref{competing_MGARCH}. No variance-targeting was applied in the BEKK.

\medskip

We consider two different data sets. \textcolor{black}{First is} a low-dimensional portfolio of daily financial returns composed of the MSCI stock index for the following 23 countries over the period December 1998 - March 2018: Australia, Austria, Belgium, Canada, Denmark, Finland, France, Germany, Greece, Hong Kong, Ireland, Italy, Japan, Netherlands, New Zealand, Norway, Portugal, Singapore, Spain, Sweden, Switzerland, the United-Kingdom, \textcolor{black}{and} the United-States. The second portfolio corresponds to the daily stock returns of the S\&P 100, where we considered firms that have been continuously included in the index from December 2015 until January 2020, excluding AbbVieInc., Dow Inc., Facebook, Inc., General Motors, Kraft Heinz, Kinder Morgan, and PayPal Holdings\textcolor{black}{. This} leaves $94$ assets and, thus, corresponds to the high-dimensional portfolio. The matrix $H_t$ in (\ref{portprob}) is deduced from the aforementioned dynamics that have been estimated on the sub-sample December 1998 - November 2015 (resp. February 2010 - January 2018), \textcolor{black}{that is 4000 (resp. $2000$) observations} for the MSCI (resp. S\&P 100) portfolio. Once the latter process is estimated in-sample, out-of-sample predictions \textcolor{black}{have been} plugged into the program~(\ref{portprob}) between December 2015 and March 2018 (resp. February 2018 - January 2020) for the MSCI (resp. S\&P 100) portfolio. For both data sets, the MSV-based models are estimated with $m=10$ lags, \textcolor{black}{and as an alternative, with $m=20$ lags. The calibration of $m$ is set in light of the scaling $d^3_T=o(T)$ in Theorem \ref{oracle_theorem} to satisfy the oracle property, or equivalently $d_T = O(T^c)$ with $0<c<1/3$\textcolor{black}{; that is}, there exists $L>0$ a finite constant such that $d_T \leq L  \, T^c$. For $p$ fixed and $T$ the in-sample size, then, $T^{c} \approx 16$ (resp. $13$) with $c=1/3-\eps$, $\eps\rightarrow 0$ for the MSCI (resp. S\&P) portfolio, providing a calibration order for the number of lags. The matrix forecast comparisons are provided in Tables \ref{DM-MSCI} and \ref{DM-SP}. First, the results \textcolor{black}{emphasize} that the proposed penalized OLS-MSV method outperforms the MGARCH-based competitors in both portfolio cases. For the low-dimensional portfolio, no clear-cut results are in favor of penalized MSV over non-penalized MSV. \textcolor{black}{However}, interestingly, fostering sparsity yields much better forecasting performances for the high-dimensional portfolio in the adaptive LASSO and SCAD (with $20$ lags), at least. Furthermore, the calibration $m=20$ for the low-dimensional portfolio does not provide better forecasts. The results are different in the high-dimensional portfolio: for each penalized MSV model, a larger $m$ with sparse estimation provides better performances; there is a gain in penalizing the MSV process over the non-penalized version of the MSV.}

\medskip

The results based on the Diebold-Mariano test are limited since they are \textcolor{black}{pair-wise} comparisons. \textcolor{black}{It} is not possible to ensure that an optimal test is clearly identified. To tackle this issue, Hansen, Lunde, and Nason (2003, 2011) proposed the Model Confidence Set (MCS) method, which is a testing framework for the null hypothesis of equivalence across subsets of models. Starting with a full set of candidate models, the MCS method sequentially trims the elements of this set, thus reducing the number of viable models. To be more precise, this approach performs an iterative selection procedure testing the null hypothesis of equal forecasting ability among all models included in a set $\Mc$ (the starting set containing all candidate models) for a given loss function. The null hypothesis is $\mathbf{H0:} \; \Eb[u_{ij,t}]=0$, $i>j$ for any $i,j \in \Mc$. To test $\mathbf{H0}$, Hansen, Lunde, and Nason (2003) proposed the following two statistics:
\beqw
t_R = \underset{i,j \in\Mc}{\max} \; |\frac{\overline{u}_{ij}}{\sqrt{\widehat{\text{Var}}(\overline{u}_{ij})}}|, \;\;\text{and} \;\;t_{SQ} = \underset{i,j \in\Mc, i>j}{\sum} \;\big(\frac{\overline{u}_{ij}}{\sqrt{\widehat{\text{Var}}(\overline{u}_{ij})}} \big)^2, \;\overline{u}_{ij} =\frac{1}{h} \overset{h}{\underset{t=1}{\sum}}u_{ij,t},
\eeqw
\textcolor{black}{where} $h$ is the number of one-period ahead forecasts and $\sqrt{\widehat{\text{Var}}(\overline{u}_{ij})}$ is the bootstrap estimate of the variance of $\overline{u}_{ij}$. The $p$-values of the test statistics are obtained using a bootstrap method. For a given confidence level, if $\mathbf{H0}$ is rejected, the worst performing model is excluded from the set $\Mc$, where such a model is identified using the \textcolor{black}{following} rule:
\beqw
j = \underset{j \in\Mc}{\arg\;\max} \;\big(\underset{i \in\Mc, i\neq j}{\sum}\overline{u}_{ij} \big) \big(\widehat{\text{Var}}(\underset{i \in\Mc, i\neq j}{\sum}\overline{u}_{ij}) \big)^{-1/2},
\eeqw
where the variance is obtained again using a bootstrap approach. Table \ref{MCS} reports the MCS results for both MSCI and S$\&$P 100 portfolios when applied to the loss function $u_{ij,t}$. We used the statistic $t_{SQ}$ to compute the $p$-values for three confidence levels ($5\%, \; 10\%$, and $20\%$). These $p$-values inform about the included/excluded models for a given confidence level. If a $p$-value is larger than the fixed confidence level, the corresponding model is included in the MCS test of statistically equivalent models. The higher the $p$-value, the better the model is in terms of prediction accuracy. For the MSCI portfolio results in Table \ref{MCS-MSCI}, we can draw the following remarks: our MSV specifications are always included for any confidence level, contrary to the standard MGARCH models; among the MSV specifications, the penalized processes provide better forecasting performances. For the high-dimensional S$\&$P 100 \textcolor{black}{portfolio} in Table \ref{MCS-SP}, only the adaptive LASSO (with $m=10$) and SCAD (with $m=20$) penalized models are included in the test for all confidence levels. These results support our findings in the Diebold-Mariano test. 

\section{Conclusion}\label{conclusion}

The focus of this \textcolor{black}{study} was \textcolor{black}{on} the estimation of high-dimensional MSV models. Our main contribution consisted in proposing an estimation framework that does not rely on standard MCMC/MCL methods but instead on a penalized OLS framework for state-space estimation. The corresponding large sample properties of the two-step estimator are derived. In particular, we considered a sparse first step estimator for a broad range of penalty functions when the number of parameters is diverging. We derived an explicit convergence rate of the second step estimator and its large sample distribution. The performances of our proposed method compared to standard MGARCH models are illustrated through simulated experiments together with an out-of-sample analysis for prediction accuracy, where our method clearly outperformed the competing MGARCH models. These results also emphasized the gain of penalization, which manages \textcolor{black}{the over-fitting problem}.

Various issues and extensions can be further considered. Our proposed model could be extended to accommodate a factor structure, long memory, and/or asymmetry, as discussed in Asai, McAleer, and Yu (2006) and Chib, Omori, and Asai (2009). 
\textcolor{black}{Besides} the factor setting considered by Chib, Nardari, and Shephard (2006) and So and Choi (2009), we can consider the rotation of the variables as in Harvey, Ruiz, and Shephard (1994) and Hafner and Preminger (2009). For the long memory property, So and Kwok (2006) extended Harvey, Ruiz, and Shephard (1994)'s model\textcolor{black}{; hence}, we can consider applying their work.
For including asymmetric effects, we may extend Harvey and Shephard (1996)\textcolor{black}{'s approach}; see Asai and McAleer (2009a) for the multivariate case. Another direction \textcolor{black}{would include} modeling directly the variance-covariance matrix $H_t$ without relying on the decomposition $D_t\Gamma D_t$. To do so, a $\log$-type dynamic on $H_t$ could be considered and the estimation could be managed through the development of a suitable state-\textcolor{black}{space-based} setting.

% %\newpage

% \medskip
% \begin{center}
% {\large\bf SUPPLEMENTARY MATERIAL}
% \end{center}

% The supplementary material contains a technical appendix, encompassing intermediary results and all assumptions and proof derivations.

% %\vspace*{-5mm}
% \begin{description}

% \item[Appendix:] Technical results, \textcolor{black}{derivative formulas}, proofs, \textcolor{black}{and} some competing M-GARCH models.

% \end{description}

\newpage
\spacingset{0.7}

\section*{References}

\begin{small}
\begin{description}

\item Abadir, K.M. and Magnus, J.R. (2005). {\sl Matrix algebra.} Cambridge University Press.

\item Alexander, C. (2001). Orthogonal GARCH, In: Alexander, C. (Ed.), {\it Mastering Risk}. Financial Times-Prentice Hall, London, pp. 21--28.

\item Asai, M., Caporin, M., and McAleer, M. (2015). Forecasting Value-at-Risk Using Block Structure Multivariate Stochastic Volatility. {\it International Review of Economics \& Finance}, {\bf 40}, 40--50.

%\item Asai, M. and McAleer, M. (2006). Asymmetric Multivariate Stochastic Volatility. {\it Econometric Reviews}, {\bf 25}, 453--473.

\item Asai, M., and McAleer, M. (2009a). Multivariate Stochastic Volatility, Leverage and News Impact Surfaces. {\it Econometrics Journal}, {\bf 12}, 292--309.

\item Asai, M., and McAleer, M. (2009b). The Structure of Dynamic Correlations in Multivariate Stochastic Volatility Models. {\it Journal of Econometrics}, {\bf 150}, 182--192.

\item Asai, M., McAleer, M., and Yu, J. (2006). Multivariate Stochastic Volatility: A Review. {\sl Econometric Reviews}, {\bf 25}, 145--175.

\item Baba, Y., Engle, R.F., Kraft, D., and Kroner, K. (1985). Multivariate Simultaneous Generalized ARCH. Unpublished Paper, University of California, San Diego. [Published as Engle and Kroner (1995)]

%\item Bauwens, L., Grigoreyva, L., and Ortega, J. P. (2016).
%Estimation and Empirical Performance of Non-Scalar Dynamic Conditional Correlation Models. %{\it Computational Statistics \& Data Analysis}, {\bf 100}, 17--36.

\item Bauwens L., Laurent, S., and Rombouts, J.K.V. (2006). Multivariate GARCH Models: A Survey. {\sl Journal of Applied Econometrics}, {\bf 21}, 79--109.

%\item Billingsley, P. (1961). The Lindeberg-Levy Theorem for Martingales. {\it Proceedings of the American Mathematical Society}, {\bf 12}, 788--792.

\item Billingsley, P. (1995). {\sl Probability and measure}. New York: John Wiley\&Sons.

\item Belloni A., Chernozhukov V. and Hansen C. (2013). Inference for high-dimensional sparse econometric models. {\it Advances in Economics and Econometrics. 10th World Congress, Econometric Society}, {\bf 3}, Cambridge: Cambridge University Press, 245--295.

%\item Belloni, A., Chernozhukov, V., and Hansen, C. (2011). Inference for High-dimensional Sparse Econometric Models. {\it Advances in Economics and Econometrics; World Congress of Econometric Society 2010.}

\item Boussama, F. (2006). Ergodicit\'edes cha\^{i}nes de Markov \`a valeursdansunevari\'et\'ealg\'ebrique: application aux mod\`elesGARCHmultivari\'es. {\it ComptesRendus de l'Acad\'emiedes Sciences Paris}, 343, 275--278.

\item \textcolor{black}{Candelon, B., Colletaz, G., Hurlin, C., and Tokpavi, S. (2011). Backtesting Value-at-Risk: A GMM Duration-based Approach. {\it Journal of Financial Econometrics}, {\bf 9}, 314--343.}

\item Chang, Y. and Park, J.Y. (2002). On the Asymptotics of ADF Tests for Unit Roots. {\it Econometric Reviews}, {\bf 21}(4), 431--447

\item Chang, Y., Park, J.Y., and Song, K. (2006). Bootstrapping Cointegrating Regressions. {\it Journal of Econometrics}, {\bf 133}, 703--739.

\item Chib, S., Nardari, F., and Shephard, N. (2006). Analysis of High Dimensional Multivariate Stochastic Volatility Models. {\it Journal of Econometrics,} {\bf 134}, 341--371.

\item Chib, S., Omori, Y., and Asai, M. (2009). Multivariate Stochastic Volatility, In: Andersen, T.G., R.A. Davis, J.P. Kreiss, and T. Mikosch (Eds.). {\sl Handbook of Financial Time Series}. New York: Springer-Verlag, pp.365--400.

\item Dan{\'i}elsson, J. (1998). Multivariate Stochastic Volatility Models: Estimation and a Comparison with VGARCH Models. {\it Journal of Empirical Finance}, {\bf 5}, 155--173.

\item Diebold F, and Mariano R. (1995). Comparing predictive accuracy. {\it Journal of Business \& Economic Statistics}, {\bf 13}, 253–263.

\item Durbin, J. and Koopman, S.J. (1997). Monte Carlo Maximum Likelihood Estimation for Non-Gaussian State Space Models. {\it Biometrika}, {\bf 84}, 669--684.

\item Ding, L. and Vo, M. (2012). Exchange Rates and Oil Prices: A Multivariate Stochastic Volatility Analysis. {\it Quarterly Review of Economics and Finance}, {\bf 52}, 15--37.

\item Durbin, J. and Koopman, S.J. (2001). {\sl Time Series Analysis by State Space Methods}. Oxford: Oxford University Press.

%\item Engle, R.F. (1982). Autoregressive Conditional Heteroskedasticity with Estimates of the Variance of United Kingdom Inflation. {\sl Econometrica}, {\bf 50}, 987--1007.

\item Engle, R.F. (2002). Dynamic Conditional Correlation. {\it Journal of Business \& Economic Statistics}, {\bf 20}: 339--350.

\item Engle, R.F. and Colacito, R. (2006). Testing and Valuing Dynamic Correlations for Asset Allocation. {\it Journal of Business \& Economic Statistics}, {\bf 24}, 238--253.

\item Engle, R.F. and Kroner, K.F. (1995). Multivariate Simultaneous Generalized ARCH. {\it Econometric Theory} {\bf 11}, 122--150.

%\item Engle, R.F., Ledoit, O., and Wolf, M. (2017). Large Dynamic Covariance Matrices. {\sl Journal of Business \& Economic Statistics}, {\bf 37}, 363--375.

\item Fan, J. and Li, R. (2001). Variable Selection via Nonconcave Penalized Likelihood and its Oracle Properties. {\it Journal of the American Statistical Association}, {\bf 96} (456), 1348--1360.

%\item Fan, J., Fan, Y., and Lv, J. (2008). Large Dimensional Covariance Matrix Estimation Using a Factor Model. {\it Journal of Econometrics}, {\bf 147}, 186--197.

\item Fan, J. and Peng, H. (2004). Nonconcave Penalized Likelihood with a Diverging Number of Parameters. {\it The Annals of Statistics}, {\bf 32} (3), 928--961.

\item Francq, C. and Zako\"{i}an, J.-M. (2010). {\it GARCH Models Structure, Statistical Inference and Financial Applications}. Chichester, West Sussex: John Wiley and Sons. 

%\item Fu, W. J. (1998). Penalized Regressions: The Bridge Versus the LASSO. {\it Journal of Computational and Graphical Statistics}, {\bf 7}, 397–416.

\item Gouri\'eroux, C. (1997). {\it ARCH Models and Financial Applications}. Springer.

\item Granger, C.W.J. and Morris, M. (1976). Time Series Modeling and Interpretation. {\it Journal of the Royal Statistical Society, Series A}, {\bf 139}, 246--257.

\item Ghysels, E., Harvey, A.C., and Renault, E. (1996). Stochastic Volatility, In: Rao, C. R. and G.S. Maddala (Eds.) {\it Statistical Models in Finance (Handbook of Statistics)}. Amsterdam: North-Holland, pp. 119--191. 

\item Hafner, C.M. and Preminger, A. (2009). Asymptotic Theory for a Factor GARCH Model. {\it Econometric Theory}, {\bf 25}, 336--363.

\item Hannan, E.J. and Kavalieris, L. (1984). Multivariate Linear Time Series Models. {\it Advances in Applied Probability}, {\bf 16}, 492--561. 

\item Hannan, E. J. and Rissanen, J. (1982). Recursive Estimation of Mixed Autoregressive-Moving Average Order. {\it Biometrika}, {\bf 69}, 81--94. 

\item Hansen, P. R., Lunde, A., and Nason, J.M. (2003). Choosing the Best Volatility Models: The Model Confidence Set Approach. {\it Oxford Bulletin of Economics and Statistics}, {\bf 65}, 839-–8-61.

\item Hansen, P. R., Lunde, A., and Nason, J.M. (2011). The Model Confidence Set. {\it Econometrica}, {\bf 79} (2), 453--497.

\item Harvey, A. (1998). Long Memory in Stochastic Volatility,
In: Knight, J. and S. Satchell (Eds.), {\em Forecasting Volatility in Financial Markets}. Oxford: Butterworth-Haineman, 307--320.

\item Harvey, A. C., Ruiz, E., and Shephard, N. (1994). Multivariate Stochastic Variance Models. {\it Review of Economic Studies}, {\bf 61}, 247--264.

\item Harvey, A. C. and Shephard, N. (1996). Estimation of an Asymmetric Stochastic Volatility Model for Asset Returns. {\it Journal of Business \& Economic Statistics}, {\bf 14}, 429--434.

\item Hastie, T., Tibshirani, R., and Wainwright, M. (2015). {\sl Statistical Learning with Sparsity: The LASSO and Generalizations.} Monographs on Statistics and Applied Probability 143. Chapman and Hall.

\item Kastner, G., Fr{$\ddot{\mbox{u}}$}-Schnatter, S., and Lopes, H.F. (2017). Efficient Bayesian Inference for Multivariate Factor Stochastic Volatility Models. {\sl Journal of Computational and Graphical Statistics}, {\bf 26}, 905--917.

\item Kim, S., Shephard, N., and Chib, S. (1998). Stochastic Volatility: Likelihood Inference and Comparison with ARCH Models. {\it Review of Economic Studies}, {\bf 65}, 361--393.

\item Lewis, R. and Reinsel, G. C. (1985). Prediction of Multivariate Time Series by Autoregressive Model Fitting. {\it Journal of Multivariate Analysis}, {\bf 16}, 393--411.

\item Loh, P.L. and Wainwright, M.J. (2017). Support Recovery Without Incoherence: A Case for Non-convex Regularization. {\it The Annals of Statistics}, {\bf 45} (6), 2455--2482.

\item L\"utkepohl, H. (2006). {\sl New Introduction to Multiple Time Series Analyais}. New York: Springer-Verlag.

%\item \textcolor{black}{McAleer, M. (2005). Automated Inference and Learning in Modeling Financial Volatility. {\it Econometric Theory}, {\bf 21}, 232--261.}

\item\textcolor{black}{Ng, S. (2013). Variable Selection in Predictive Regressions. {\it Handbook
of Economic Forecasting.} {\bf 2}, 752--789.}

\item Poignard, B. (2020). Asymptotic Theory of the Adaptive Sparse Group Lasso. {\it Annals of the Institute of Statistical Mathematics.} {\bf 72}, 297--328.

\item Poignard, B. and Fermanian, J.D. (2021). High-dimensional Penalized ARCH Processes. {\it Econometric Reviews}. {\bf 40} (1), 86--107.

\item Racine, J. (2000). Consistent Cross-validatory Model-selection for Dependent Data: HV-block Cross-validation. {\it Journal of Econometrics.} {\bf 99}, 39--61.

\item Ruiz, E. (1994). Quasi-maximum Likelihood Estimation of Stochastic Volatility Models. {\it Journal of Econometrics}, {\bf 63}, 289--306.

\item\textcolor{black}{Shephard, N. (1993). Fitting Nonlinear Time-series Models with Applications to Stochastic Variance Models. {\it Journal of Applied Econometrics}, {\bf 8}, S135--S152.}

\item\textcolor{black}{So, M. K. P. and Choi, C. Y. (2009). A Threshold Factor Multivariate Stochastic Volatility Model. {\it Journal of Forecasting}, {\bf 28}, 712--735.}

\item\textcolor{black}{So, M. K. P. and Kwok, S. W. (2006). A Multivariate Long Memory Stochastic Volatility Model. {\it Physica A: Statistical Mechanics and its Applications}, {\bf 362}, 450--464.}

\item\textcolor{black}{So, M. K. P., Li, W. K., and Lam, K. (1997). Multivariate Modelling of the Autoregressive Random Variance Process. {\sl Journal of Time Series Analysis}, {\bf 18}, 429--446.}

\item Shiryaev, A. N. (1991). {\it Probability}. Berlin: Springer.

\item\textcolor{black}{Taylor, S. J. (1994). Modeling Stochastic Volatility: A Review and Comparative Study. {\sl Mathematical Finance}, {\bf 4}, 183--204.}

\item Tibshirani, R. (1996). Regression Shrinkage and Selection via the LASSO. {\it Journal of the Royal Statistical Society. Series B}, {\bf 58} (1), 267--288.

\item Tse, Y.K. and Tsui, A.K.C. (2002). A Multivariate Generalized Autoregressive Conditional Heteroscedasticity Model with Time-varying Correlations. {\sl Journal of Business \& Economic Statistics}, {\bf 20}, 351--361.

\item Wilms, I., Basu, S., Bien, J., and Matteson D.S. (2021). Sparse Identification and Estimation of Large-scale Vector Autoregressive Moving Averages. To appear in {\it Journal of the American Statistical Association}.

\item Zhang, C.-H. (2010). Nearly Unbiased Variable Selection under Minimax Concave Penalty. {\it The Annals of Statistics}, {\bf 38}, 894--942.

\item Zou, H. (2006). The Adaptive LASSO and Its Oracle Properties. {\it Journal of the American Statistical Association}, {\bf 101}, No. 476, 1418--1429.

%\item[]  [dataset] MSCI (2018), Daily MSCI stock indices,
%https://www.msci.com/.
%\item[]  [dataset] Yahoo Finance (2020), Daily stock prices for companies listed in S\&P 100,
%https://finance.yahoo.com.

\end{description}
\end{small}

% \section*{Data availability statement}

% The data that support the findings will be available through MSCI at https://www.msci.com/ and Yahoo Finance at https://finance.yahoo.com following an embargo from the date of publication to allow for the commercialization of the research findings.

\newpage
\spacingset{1}

\begin{table}[h!]\centering\caption{\label{simulation_average_march} Average distance true/estimated covariance matrices - M-ARCH($q^*$) (100 replications)}
{\normalsize
\begin{tabular}{l | c | c | c | c | c }\hline&
$\widehat{H}^{\text{dcc}}_t$&$\widehat{H}^{\text{ccc}}_t$&$\widehat{H}^{\text{ols}}_t$&$\widehat{H}^{\text{ols,scad}}_t$&$\widehat{H}^{\text{ols,mcp}}_t$\\\hline
$p = 15$&$6.37$&$6.75$&$5.92$&$5.81$&$5.80$\\
$p = 50$&$17.10$&$18.38$&$16.80$&$15.87$&$15.85$\\
$p = 100$&$64.05$&$70.51$&$62.09$&$59.84$&$59.66$\\
\hline\end{tabular}
}
\end{table}	

\begin{table}[h!]\centering\caption{\label{simulation_average_bekk} Average distance true/estimated covariance matrices - BEKK (100 replications)}
{\normalsize
\begin{tabular}{l | c | c | c | c | c }\hline&
$\widehat{H}^{\text{dcc}}_t$&$\widehat{H}^{\text{ccc}}_t$&$\widehat{H}^{\text{ols}}_t$&$\widehat{H}^{\text{ols,scad}}_t$&$\widehat{H}^{\text{ols,mcp}}_t$\\\hline
$p = 15$&$21.75$&$22.02$&$21.60$&$20.57$&$21.18$\\
$p = 50$&$114.40$&$115.63$&$115.78$&$111.16$&$113.14$\\
$p = 100$&$292.28$&$295.99$&$291.73$&$285.88$&$287.97$
\\
\hline\end{tabular}
}
\end{table}	

\newpage

\begin{landscape}

\begin{table}[h]
\begin{subtable}[h]{0.45\textwidth}
\centering
\scalebox{0.8}{\begin{tabular}{c| c | c | c | c | c | c | c | c | c | c | c}\hline&\color{black}DCC&\color{black} O-G &\color{black}BEKK&\color{black}$\text{MSV}_{10}$&\color{black}$\text{MSV-AL}_{10}$&\color{black}$\text{MSV-SCA}_{10}$&\color{black}$\text{MSV-MCP}_{10}$&\color{black}$\text{MSV}_{20}$&\color{black}$\text{MSV-AL}_{20}$&\color{black}$\text{MSV-SCA}_{20}$&\color{black}$\text{MSV-MCP}_{20}$\\
\hline

\color{black} DCC&&\color{black}$\mathbf{-2.259}^c$&\color{black}$\mathbf{-4.302}^c$&\color{black}$\mathbf{7.408}^c$&\color{black}$\mathbf{7.541}^c$&\color{black}$\mathbf{7.188}^c$&\color{black}$\mathbf{7.441}^c$&\color{black}$\mathbf{7.215}^c$&\color{black}$\mathbf{7.496}^c$&\color{black}$\mathbf{7.466}^c$&\color{black}$\mathbf{7.399}^c$\\

\color{black}O-G  &\color{black}$\mathbf{2.259}^c$&&\color{black}$\mathbf{-4.852}^c$&\color{black}$\mathbf{6.207}^c$&\color{black}$\mathbf{6.241}^c$&\color{black}$\mathbf{6.056}^c$&\color{black}$\mathbf{6.228}^c$&\color{black}$\mathbf{6.196}^c$&\color{black}$\mathbf{6.279}^c$&\color{black}$\mathbf{6.298}^c$&\color{black}$\mathbf{6.270}^c$\\

\color{black}BEKK&\color{black}$\mathbf{4.302}^c$&\color{black}$\mathbf{4.852}^c$&&\color{black}$\mathbf{6.249}^c$&\color{black}$\mathbf{6.281}^c$&\color{black}$\mathbf{6.204}^c$&\color{black}$\mathbf{6.264}^c$&\color{black}$\mathbf{6.256}^c$&\color{black}$\mathbf{6.292}^c$&\color{black}$\mathbf{6.323}^c$&\color{black}$\mathbf{6.296}^c$\\

\color{black}$\text{MSV}_{10}$&\color{black}$\mathbf{-7.408}^c$&\color{black}$\mathbf{-6.207}^c$&\color{black}$\mathbf{-6.249}^c$&&\color{black}$0.667$&\color{black}$\mathbf{-1.535}^a$&\color{black}$1.257$&\color{black}$0.542$&\color{black}$0.641$&\color{black}$0.983$&\color{black}$1.276$\\

\color{black}$\text{MSV-AL}_{10}$&\color{black}$\mathbf{-7.541}^c$&\color{black}$\mathbf{-6.241}^c$&\color{black}$\mathbf{-6.281}^c$&\color{black}$-0.667$&&\color{black}$\mathbf{-1.843}^b$&\color{black}$-0.299$&\color{black}$-0.096$&\color{black}$-0.059$&\color{black}$0.237$&\color{black}$0.304$\\

\color{black}$\text{MSV-SCA}_{10}$&\color{black}$\mathbf{-7.188}^c$&\color{black}$\mathbf{-6.056}^c$&\color{black}$\mathbf{-6.204}^c$&\color{black}$\mathbf{1.535}^a$&\color{black}$\mathbf{1.843}^b$&&\color{black}$\mathbf{2.078}^b$&\color{black}$\mathbf{1.299}^a$&\color{black}$\mathbf{1.521}^a$&\color{black}$\mathbf{2.008}^b$&\color{black}$\mathbf{1.902}^b$\\

\color{black}$\text{MSV-MCP}_{10}$&\color{black}$\mathbf{-7.441}^c$&\color{black}$\mathbf{-6.228}^c$&\color{black}$\mathbf{-6.264}^c$&\color{black}$-1.257$&\color{black}$0.299$&\color{black}$\mathbf{-2.078}^b$&&\color{black}$0.157$&\color{black}$0.227$&\color{black}$0.619$&\color{black}$0.839$\\

\color{black}$\text{MSV}_{20}$&\color{black}$\mathbf{-7.215}^c$&\color{black}$\mathbf{-6.196}^c$&\color{black}$\mathbf{-6.256}^c$&\color{black}$-0.542$&\color{black}$0.096$&\color{black}$\mathbf{-1.299}^a$&\color{black}$-0.157$&&\color{black}$0.076$&\color{black}$0.547$&\color{black}$1.161$\\

\color{black}$\text{MSV-AL}_{20}$&\color{black}$\mathbf{-7.496}^c$&\color{black}$\mathbf{-6.279}^c$&\color{black}$\mathbf{-6.292}^c$&\color{black}$-0.641$&\color{black}$0.059$&\color{black}$\mathbf{-1.521}^a$&\color{black}$-0.227$&\color{black}$-0.076$&&\color{black}$0.609$&\color{black}$0.826$\\

\color{black}$\text{MSV-SCA}_{20}$&\color{black}$\mathbf{-7.466}^c$&\color{black}$\mathbf{-6.298}^c$&\color{black}$\mathbf{-6.323}^c$&\color{black}$-0.983$&\color{black}$-0.237$&\color{black}$\mathbf{-2.008}^b$&\color{black}$-0.619$&\color{black}$-0.547$&\color{black}$-0.609$&&\color{black}$0.101$\\

\color{black}$\text{MSV-MCP}_{20}$&\color{black}$\mathbf{-7.399}^c$&\color{black}$\mathbf{-6.270}^c$&\color{black}$\mathbf{-6.296}^c$&\color{black}$-1.276$&\color{black}$-0.304$&\color{black}$\mathbf{-1.902}^b$&\color{black}$-0.839$&\color{black}$-1.161$&\color{black}$-0.826$&\color{black}$-0.101$&\\

\hline
\end{tabular}}\caption{\color{black} MSCI portfolio}\label{DM-MSCI}
\end{subtable}
\hfill\vspace*{0.8cm}\\
\begin{subtable}[h]{0.45\textwidth}
\centering
\scalebox{0.8}{\begin{tabular}{c| c | c | c | c | c | c | c | c | c | c | c}\hline&\color{black}DCC&\color{black} O-G &\color{black}BEKK&\color{black}$\text{MSV}_{10}$&\color{black}$\text{MSV-AL}_{10}$&\color{black}$\text{MSV-SCA}_{10}$&\color{black}$\text{MSV-MCP}_{10}$&\color{black}$\text{MSV}_{20}$&\color{black}$\text{MSV-AL}_{20}$&\color{black}$\text{MSV-SCA}_{20}$&\color{black}$\text{MSV-MCP}_{20}$\\
\hline

\color{black} DCC&&\color{black}$\mathbf{-4.785}^c$&\color{black}$\mathbf{-4.698}^c$&\color{black}$\mathbf{3.694}^c$&\color{black} $\mathbf{7.572}^c$&\color{black} $\mathbf{3.652}^c$&\color{black}$\mathbf{3.645}^c$&\color{black}$\mathbf{2.791}^c$&\color{black}$\mathbf{5.056}^c$&\color{black}$\mathbf{4.587}^c$&\color{black}$\mathbf{2.792}^c$\\

\color{black}O-G  &\color{black}  $\mathbf{4.785}^c$&&\color{black}$0.732$&\color{black}$\mathbf{5.646}^c$&\color{black} $\mathbf{8.562}^c$&\color{black} $\mathbf{5.605}^c$&\color{black} $\mathbf{5.609}^c$&\color{black}$\mathbf{4.996}^c$&\color{black}$\mathbf{6.496}^c$&\color{black}$\mathbf{6.202}^c$&\color{black}$\mathbf{4.995}^c$\\

\color{black}BEKK&\color{black}$\mathbf{4.698}^c$&\color{black} $-0.732$&&\color{black}$\mathbf{5.754}^c$&\color{black}$\mathbf{8.781}^c$&\color{black} $\mathbf{5.713}^c$&\color{black} $\mathbf{5.716}^c$&\color{black} $\mathbf{5.090}^c$&\color{black}$\mathbf{6.651}^c$&\color{black} $\mathbf{6.318}^c$&\color{black}   $\mathbf{5.088}^c$\\

\color{black}$\text{MSV}_{10}$&\color{black}  $\mathbf{-3.694}^c$&\color{black} $\mathbf{-5.646}^c$&\color{black} $\mathbf{-5.754}^c$&\color{black} &\color{black} $\mathbf{5.902}^c$&\color{black}$-1.117$&\color{black}$\mathbf{-1.732}^b$&\color{black} $\mathbf{-7.651}^c$&\color{black} $\mathbf{5.701}^c$&\color{black} $\mathbf{5.235}^c$&\color{black}$1.182$\\

\color{black}$\text{MSV-AL}_{10}$&\color{black}  $\mathbf{-7.572}^c$&\color{black} $\mathbf{-8.562}^c$&\color{black}$\mathbf{-8.781}^c$&\color{black} $\mathbf{-5.902}^c$&\color{black}  &\color{black}$\mathbf{-5.936}^c$&\color{black} $\mathbf{-5.928}^c$&\color{black} $\mathbf{-6.870}^c$&\color{black}$\mathbf{-4.704}^c$&\color{black} $\mathbf{-4.803}^c$&\color{black}$\mathbf{-6.828}^c$\\

\color{black}$\text{MSV-SCA}_{10}$&\color{black}  $\mathbf{-3.652}^c$&\color{black}$\mathbf{-5.605}^c$&\color{black} $\mathbf{-5.713}^c$&\color{black} $1.117$&\color{black} $\mathbf{5.936}^c$&\color{black}   &\color{black}$-0.2883$&\color{black} $\mathbf{-7.735}^c$&\color{black}  $\mathbf{5.798}^c$&\color{black} $\mathbf{5.427}^c$&\color{black}$\mathbf{-7.201}^c$\\

\color{black}$\text{MSV-MCP}_{10}$&\color{black}  $\mathbf{-3.645}^c$&\color{black} $\mathbf{-5.609}^c$&\color{black}$\mathbf{-5.716}^c$&\color{black} $\mathbf{1.732}^b$&\color{black} $\mathbf{5.928}^c$&\color{black} $0.288$&\color{black} &\color{black} $\mathbf{-7.602}^c$&\color{black} $\mathbf{5.746}^c$&\color{black} $\mathbf{5.277}^c$&\color{black} $\mathbf{-7.095}^c$\\

\color{black}$\text{MSV}_{20}$&\color{black}  $\mathbf{-2.791}^c$&\color{black} $\mathbf{-4.996}^c$&\color{black} $\mathbf{-5.090}^c$&\color{black} $\mathbf{7.651}^c$&\color{black} $\mathbf{6.870}^c$&\color{black} $\mathbf{7.735}^c$&\color{black}  $\mathbf{7.602}^c$&\color{black}   &\color{black}  $\mathbf{7.905}^c$&\color{black} $\mathbf{8.406}^c$&\color{black} $0.312$\\

\color{black}$\text{MSV-AL}_{20}$&\color{black}  $\mathbf{-5.056}^c$&\color{black} $\mathbf{-6.496}^c$&\color{black}$\mathbf{-6.651}^c$&\color{black} $\mathbf{-5.701}^c$&\color{black} $\mathbf{4.704}^c$&\color{black} $\mathbf{-5.798}^c$&\color{black}$\mathbf{-5.746}^c$&\color{black} $\mathbf{-7.905}^c$&\color{black} &\color{black} $\mathbf{-1.991}^b$&\color{black}$\mathbf{-7.781}^c$\\

\color{black}$\text{MSV-SCA}_{20}$&\color{black}  $\mathbf{-4.587}^c$&\color{black}$\mathbf{-6.202}^c$&\color{black} $\mathbf{-6.318}^c$&\color{black}$\mathbf{-5.235}^c$&\color{black} $\mathbf{4.803}^c$&\color{black} $\mathbf{-5.427}^c$&\color{black}$\mathbf{-5.277}^c$&\color{black} $\mathbf{-8.406}^c$&\color{black} $\mathbf{1.991}^b$&\color{black} &\color{black}$\mathbf{-8.366}^c$\\

\color{black}$\text{MSV-MCP}_{20}$&\color{black} $\mathbf{-2.792}^c$&\color{black} $\mathbf{-4.995}^c$&\color{black} $\mathbf{-5.088}^c$&\color{black} $-1.182$&\color{black} $\mathbf{6.828}^c$&\color{black} $\mathbf{7.201}^c$&\color{black} $\mathbf{7.095}^c$&\color{black} $-0.312$&\color{black} $\mathbf{7.781}^c$&\color{black}  $\mathbf{8.366}^c$&\color{black}

\\
\hline
\end{tabular}}\caption{\color{black}S\&P 100 portfolio}\label{DM-SP}\vspace*{0.5cm}
\end{subtable}
\caption{\color{black} This table reports the out-of-sample t-statistics of the Diebold-Mariano test for the MSCI (Table \ref{DM-MSCI}) and S\&P 100 (Table \ref{DM-SP}) portfolios that checks the equality between covariance matrix forecasts using the loss function $u_{ij,t}$ over the period December 2015 - March 2018 and February 2018 - January 2020, respectively. This loss function is defined as the difference between squared realized returns of alternative Multivariate Variance-Covariance models. When the null hypothesis of equal predictive accuracy is rejected, a positive number is evidence in favour of the model in the column. $a$, $b$, $c$: rejection of the null hypothesis at 10\%, 5\% and 1\% respectively. The MSV models are indexed by the number of lags $m$.}
\label{DM}
\end{table}

\end{landscape}

\newpage

\begin{table}[h]
\begin{subtable}[h]{0.45\textwidth}
\centering
\begin{tabular}{c | c | c | c} \hline
\color{black}Confidence level &\color{black}$\; 5\%\;$&\color{black} $10\%$&\color{black} $20\%$\\
\hline
\color{black}DCC&\color{black} $0.001$&\color{black}$0.001$&\color{black}\color{black}$0.001$\\
\color{black}O-G &\color{black}$0.001$&\color{black}$0.001$&\color{black}\color{black}$0.001$\\
\color{black}BEKK&\color{black}$0.008$&\color{black}$0.008$&\color{black} \color{black}$0.009$\\
\color{black}$\text{MSV}_{10}$&\color{black}$\mathbf{0.723}$&\color{black}$\mathbf{0.704}$&\color{black} $\mathbf{0.697}$\\
\color{black}$\text{MSV-AL}_{10}$&\color{black}$\mathbf{0.934}$&\color{black}$\mathbf{0.916}$&\color{black}$\mathbf{0.929}$\\
\color{black}$\text{MSV-SCA}_{10}$&\color{black}$\mathbf{0.103}$&\color{black}$\mathbf{0.177}$&\color{black}$0.111$\\
\color{black}$\text{MSV-MCP}_{10}$&\color{black}$\mathbf{0.723}$&\color{black}$\mathbf{0.704}$&\color{black}$\mathbf{0.697}$\\
\color{black}$\text{MSV}_{20}$&\color{black}$\mathbf{0.723}$&\color{black}$\mathbf{0.704}$&\color{black}$\mathbf{0.697}$\\
\color{black}$\text{MSV-AL}_{20}$&\color{black}$\mathbf{0.723}$&\color{black} $\mathbf{0.865}$&\color{black}$\mathbf{0.862}$\\
\color{black}$\text{MSV-SCA}_{20}$&\color{black}$\mathbf{0.934}$&\color{black}$\mathbf{0.916}$&\color{black}$\mathbf{0.929}$\\
\color{black}$\text{MSV-MCP}_{20}$&\color{black}$\mathbf{1.000}$&\color{black}$\mathbf{1.000}$&\color{black}$\mathbf{1.000}$\\\hline
\end{tabular}
\caption{\color{black}MCS $p$-values, MSCI portfolio} \label{MCS-MSCI}
\end{subtable}
\hfill
\begin{subtable}[h]{0.45\textwidth}
\centering
\begin{tabular}{c | c | c | c} \hline
\color{black}Confidence level &\color{black}$\; 5\%\;$&\color{black} $10\%$&\color{black} $20\%$\\
\hline
\color{black}DCC&\color{black}$0.001$&\color{black}$0.001$&\color{black}$0.002$\\
\color{black}O-G &\color{black}$0.001$&\color{black}$0.001$&\color{black}$0.001$\\
\color{black}BEKK&\color{black}$0.001$&\color{black}$0.001$&\color{black}$0.001$\\
\color{black}$\text{MSV}_{10}$&\color{black}$0.001$&\color{black}$0.001$&\color{black}$0.001$\\
\color{black}$\text{MSV-AL}_{10}$&\color{black}$\mathbf{1.000}$&\color{black}$\mathbf{1.000}$&\color{black}$\mathbf{1.000}$\\
\color{black}$\text{MSV-SCA}_{10}$&\color{black}$0.001$&\color{black}$0.001$&\color{black}$0.002$\\
\color{black}$\text{MSV-MCP}_{10}$&\color{black}$0.001$&\color{black}$0.001$&\color{black}$0.002$\\
\color{black}$\text{MSV}_{20}$&\color{black}$0.001$&\color{black}$0.001$&\color{black}$0.001$\\
\color{black}$\text{MSV-AL}_{20}$&\color{black}$\mathbf{0.06}$&\color{black}$0.007$&\color{black}$0.05$\\
\color{black}$\text{MSV-SCA}_{20}$&\color{black}$\mathbf{0.301}$&\color{black}$\mathbf{0.280}$&\color{black}$\mathbf{0.250}$\\
\color{black}$\text{MSV-MCP}_{20}$&\color{black}$0.001$&\color{black}$0.001$&\color{black} $0.001$\\\hline

\end{tabular}
\caption{\color{black}MCS $p$-values, S$\&$P 100 portfolio}\label{MCS-SP}
\end{subtable}
\caption{\color{black}These tables report the MCS results using the loss $u_{ij,t}$ and the statistic $t_{SQ}$ for different confidence levels. Bold values indicate that the models are included in the confidence set, that is they are statistically equivalent in terms of prediction accuracy. For each confidence level, the lowest $p$-value corresponds to the first model being excluded; the largest $p$-value corresponds to the best performing model. The MSV models are indexed by the number of lags $m$.}
\label{MCS}
\end{table}

\newpage

\appendix

\section{Intermediary results}\label{inter}

In this Section, we introduce some technical results \textcolor{black}{used} in our proofs. The dependent setting requires more sophisticated probabilistic tools to derive asymptotic results than the i.i.d. case. 
%The strict stationarity assumption and the martingale difference property of the gradient $\nabla_{\gamma}f(y_s,s\leq t;\theta_{0,\textcolor{black}{T}},\gamma_0)$ - \textcolor{black}{and of $\nabla_{\theta_{\textcolor{black}{T}}}\ell(y_s,s\leq t;\theta_{0,\textcolor{black}{T}})$ when the dimension is fixed} - allow for using the central limit theorem of Billingsley (1961). We remind this result stated as a corollary in Billingsley (1961).
%\begin{corollary}[Billingsley, 1961]
%If $(x_t,\Fc_t)$ is a stationary and ergodic sequence of square integrable martingale increments such that $\sigma^2_x=\normalfont\text{Var}(x_t) \neq 0$, then $\frac{1}{\sqrt{T}}\overset{T}{\underset{t=1}{\sum}}x_t \overset{d}{\underset{T \rightarrow \infty}{\longrightarrow}} \Nc_{\Rb}(0,\sigma^2_x)$.
%\end{corollary}
When we consider a diverging number of parameters, the empirical criterion can be viewed as a sequence of dependent arrays for which we need refined asymptotic results. Shiryaev (1991) proposed a version of the central limit theorem for dependent sequence of arrays, provided this sequence is a square integrable martingale difference satisfying the so-called Lindeberg condition. A similar theorem can be found in Billingsley (1995, Theorem 35.12, p.476). \textcolor{black}{We provide here Theorem 4 of Shiryaev (see Theorem 4, p.543, Shiryaev, 1991) that we use} to derive the asymptotic distribution of the penalized estimator (\ref{obj_crit_gen}) in Theorem \ref{oracle_theorem}.
\begin{theorem}[Shiryaev, 1991]\label{lindeberg_shiryaev}
Let a sequence of square integrable martingale differences $\xi^n=(\xi_{nk},\Fc^n_k), n \geq 1$, with $\Fc^n_k=\sigma(\xi_{ns}, s \leq k)$, satisfy the Lindeberg condition for any $0 < t \leq 1$, for $\eps>0$, given by
\beqw
\overset{\lfloor nt \rfloor}{\underset{k=0}{\sum}}\Eb\Big[\xi^2_{nk}\mathbf{1}_{|\xi_{nk}|>\eps}|\Fc^n_{k-1}\Big] \overset{\Pb}{\underset{n \rightarrow \infty}{\longrightarrow}} 0,
\eeqw
then if $\overset{\lfloor nt \rfloor}{\underset{k=0}{\sum}}\Eb\Big[\xi^2_{nk}|\Fc^n_{k-1}\Big] \overset{\Pb}{\underset{n \rightarrow \infty}{\longrightarrow}} \sigma^2_t$, or $\overset{\lfloor nt \rfloor}{\underset{k=0}{\sum}}\xi^2_{nk} \overset{\Pb}{\underset{n \rightarrow \infty}{\longrightarrow}} \sigma^2_t$, then $\overset{\lfloor nt \rfloor}{\underset{k=0}{\sum}}\xi_{nk} \overset{d}{\underset{n \rightarrow \infty}{\longrightarrow}} \Nc_{\Rb}(0,\sigma^2_t)$.
\end{theorem}

\section{\textcolor{black}{Derivative formulas}}\label{appendix_derivative}

\textcolor{black}{In this section, we derive the gradient, Hessian for both $\Gb_T(\underline{y};.)$ and $\Lb_T(\underline{y};.)$. We also provide the cross-derivatives of $\Lb_T(\underline{y};.)$.} 

\vspace*{0.3cm}

\noindent \textcolor{black}{\textbf{First and second order derivatives of $\Gb_T(\underline{y};\theta_{\textcolor{black}{T}})$ with respect to $\theta_T$.}} To derive the gradient function $\nabla_{\theta_{\textcolor{black}{T}}}\Gb_T(\underline{y};\theta_{\textcolor{black}{T}})$, we consider the the differential $\dd \Gb_T(\underline{y};\theta_{\textcolor{black}{T}})$ with respect to $\Psi_{1:m}$, which is
\begin{eqnarray*}
\lefteqn{\dd\Gb_T(\underline{y};\theta_{\textcolor{black}{T}})}\\ &=&\dd\big\{ \frac{1}{2T}\overset{T}{\underset{t=1}{\sum}}\big(x_t-\Psi_{1:m}Z_{m,t-1}\big)^\top \big(x_t-\Psi_{1:m}Z_{m,t-1}\big) \big\}\\
& = & \frac{1}{2T}\overset{T}{\underset{t=1}{\sum}} \tr\big(-x^\top_t (\dd\Psi_{1:m})Z_{m,t-1}-Z^\top_{m,t-1}(\dd\Psi_{1:m})^\top x_t+Z^\top_{m,t-1}(\dd\Psi_{1:m})^\top\Psi_{1:m}Z_{m,t-1}\\
& & +Z^\top_{m,t-1}\Psi^\top_{1:m}(\dd\Psi_{1:m})Z_{m,t-1}\big).
\end{eqnarray*}
Now using the matrix identification formulas of Abadir and Magnus (2005), we obtain for the score
\beqw
\nabla_{\theta_{\textcolor{black}{T}}} \Gb_T(\underline{y};\theta_{0,\textcolor{black}{T}})= \frac{1}{T}\overset{T}{\underset{t=1}{\sum}} \nabla_{\theta_{\textcolor{black}{T}}} \ell(y_s, s\leq t;\theta_{0,\textcolor{black}{T}})= -\frac{1}{T}\overset{T}{\underset{t=1}{\sum}} \big(Z_{m,t-1}\otimes \{x_t-\Psi_{0,1:m} Z_{m,t-1}\}\big).
\eeqw
As for the Hessian, we aim at extracting the form $\text{tr}(L (\dd X)^\top M (\dd X))$ for $L$ and $M$ constant matrices, with sizes $n \times n$ and $v \times v$ respectively, with a differential operator $\dd$ applied twice with respect to the $v \times n$ matrix $X$. In our case, applying the differential operator twice with respect to the $p \times pm$ matrix $\Psi_{1:m}$, we have
\beqw
\dd^2\Gb_T(\underline{y};\theta_{0,\textcolor{black}{T}}) = \cfrac{1}{T}\overset{T}{\underset{t=1}{\sum}} \text{tr}\big(Z^\top_{m,t-1}(\dd \Psi_{1:m})^\top(\dd\Psi_{1:m})Z_{m,t-1}\big) = \cfrac{1}{T}\overset{T}{\underset{t=1}{\sum}} \text{tr}\big(Z_{m,t-1}Z^\top_{m,t-1}(\dd \Psi_{1:m})^\top I_p(\dd\Psi_{1:m})\big). 
\eeqw
Hence, using exercise 13.49 of Abadir and Magnus (2006), we obtain by identification of the Hessian:
\beqw
\nabla^2_{\theta \theta^\top_{\textcolor{black}{T}}} \Gb_T(\underline{y};\theta_{0,\textcolor{black}{T}}) = \cfrac{1}{T}\overset{T}{\underset{t=1}{\sum}} (Z_{m,t-1} Z^\top_{m,t-1} \otimes I_p).
\eeqw

\vspace*{0.3cm}

\noindent \textcolor{black}{\textbf{First and second order derivatives of $\Lb_T(\underline{y};\theta_{\textcolor{black}{T}},\gamma)$ with respect to $\gamma$.}} Let us now consider the second step loss function. The score $\nabla_{\gamma}\Lb_T(\underline{y};\theta_{\textcolor{black}{T}},\gamma)$ and Hessian $\nabla^2_{\gamma\gamma^\top}\Lb_T(\underline{y};\theta_{\textcolor{black}{T}},\gamma)$ can be obtained in a similar manner, where $\gamma = \text{vec}(\Gamma)$. For the sake of clarification, we omit the $\theta$ argument in $\textcolor{black}{k_T}$. To compute the score, we use the differential operator $\dd$ with respect to $\Gamma$, so that we have
\begin{eqnarray*}
\lefteqn{\dd \Lb_T(\underline{y};\theta_{\textcolor{black}{T}},\gamma)}\\
&=& \frac{1}{2T}\overset{T}{\underset{t=1}{\sum}}\tr\big(-\textcolor{black}{y_t^{\ell\top}}(\dd\Gamma)K_{t-1}-K^\top_{t-1}(\dd\Gamma)^\top \textcolor{black}{y_t^{\ell}} + K^\top_{t-1}(\dd\Gamma)^\top \Gamma K_{t-1} + K^\top_{t-1}\Gamma^\top (\dd\Gamma)K_{t-1}\big)\\
& = & \frac{1}{2T}\overset{T}{\underset{t=1}{\sum}}\tr\big(-2 \textcolor{black}{y_t^{\ell\top}}(\dd\Gamma)K_{t-1} + 2 K^\top_{t-1}\Gamma^\top(\dd\Gamma)K_{t-1}\big).
\end{eqnarray*}
As a consequence, by identification, the first order derivative is
\beqw
\nabla_{\gamma}\Lb_T(\underline{y};\theta_{\textcolor{black}{T}},\gamma)=
 \frac{1}{T}\overset{T}{\underset{t=1}{\sum}}\Big[\big(K_{t-1}\otimes -\textcolor{black}{y_t^{\ell}}\big) + \big(K_{t-1}\otimes \Gamma K_{t-1}\big)\Big] = - \frac{1}{T}\overset{T}{\underset{t=1}{\sum}}\big(K_{t-1}\otimes \big\{\textcolor{black}{y_t^{\ell}}-\Gamma K_{t-1}\big\}\big).
\eeqw
The second order differential is now
\beqw
\dd^2\Lb_T(\underline{y};\theta_{\textcolor{black}{T}},\gamma)=
 \frac{1}{T}\overset{T}{\underset{t=1}{\sum}}\tr\big(K^\top_{t-1}(\dd\Gamma)^\top (\dd\Gamma)K_{t-1}\big).
\eeqw
By identification, we have
\beqw
\nabla^2_{\gamma\gamma^\top}\Lb_T(\underline{y};\theta_{\textcolor{black}{T}},\gamma) = \frac{1}{T}\overset{T}{\underset{t=1}{\sum}}\big(I_p \otimes K_{t-1}K^\top_{t-1}\big).
\eeqw

\vspace*{0.3cm}
 
We now focus on the second step loss function and its cross-derivatives, which are used in the moment conditions of Assumptions \ref{assumption_second_cross_2ndstep} and \ref{assumption_third_cross_1_2ndstep}. \\
\noindent \textbf{Cross-derivatives of $\Lb_T(\underline{y};\theta_{\textcolor{black}{T}},\gamma)$.} Let us investigate element-by-element, the first order derivative with respect to the elements in $\Gamma$ and then its derivative with respect to the elements in $\Psi_{1:m}$. This computation enables to explicit the moment conditions we assume in Assumptions \ref{assumption_third_cross_1_2ndstep} and \ref{assumption_third_cross_2_2ndstep}. The loss function $\Lb_T(\underline{y};\theta_{\textcolor{black}{T}},\gamma)$ is based on $f(y_s,s\leq t;\theta_{\textcolor{black}{T}},\gamma)$, which can be expanded as
\begin{eqnarray}
\lefteqn{f(y_s,s\leq t;\theta_{\textcolor{black}{T}},\gamma)} \nonumber\\
& = & \textcolor{black}{y_t^{\ell\top}} \textcolor{black}{y_t^{\ell}}-\textcolor{black}{y_t^{\ell\top}}c^* -\textcolor{black}{y_t^{\ell\top}}\Phi \textcolor{black}{y_{t-1}^{\ell}} - \textcolor{black}{y_t^{\ell\top}} \Xi \textcolor{black}{y_{t-1}^{\ell}} + \textcolor{black}{y_t^{\ell\top}}\Xi \Psi_{1:m}Z_{m,t-2} \nonumber\\
&  & - c^{*\top} \textcolor{black}{y_t^{\ell}}+c^{*\top}c^* +c^{*\top}\Phi \textcolor{black}{y_{t-1}^{\ell}} +c^{*\top} \Xi \textcolor{black}{y_{t-1}^{\ell}} -c^{*\top}\Xi \Psi_{1:m}Z_{m,t-2}\nonumber\\
& & - \textcolor{black}{y_{t-1}^{\ell\top}}\Phi^\top \textcolor{black}{y_t^{\ell}}+\textcolor{black}{y_{t-1}^{\ell\top}}\Phi^\top c^* +\textcolor{black}{y_{t-1}^{\ell\top}}\Phi^\top \Phi \textcolor{black}{y_{t-1}^{\ell}} + \textcolor{black}{y_{t-1}^{\ell\top}}\Phi^\top \Xi \textcolor{black}{y_{t-1}^{\ell}} -\textcolor{black}{y_{t-1}^{\ell\top}}\Phi^\top\Xi \Psi_{1:m}Z_{m,t-2} \nonumber\\
& & - \textcolor{black}{y_{t-1}^{\ell\top}}\Xi^\top \textcolor{black}{y_t^{\ell}}+\textcolor{black}{y_{t-1}^{\ell\top}}\Xi^\top c^* +\textcolor{black}{y_{t-1}^{\ell\top}}\Xi^\top \Phi \textcolor{black}{y_{t-1}^{\ell}} + \textcolor{black}{y_{t-1}^{\ell\top}}\Xi^\top \Xi \textcolor{black}{y_{t-1}^{\ell}} -\textcolor{black}{y_{t-1}^{\ell\top}}\Xi^\top\Xi \Psi_{1:m}Z_{m,t-2}\nonumber \\
& & + Z^\top_{m,t-2}\Psi^\top_{1:m}\Xi^\top \textcolor{black}{y_t^{\ell}}- Z^\top_{m,t-2}\Psi^\top_{1:m}\Xi^\top c^*\nonumber\\
& & -Z^\top_{m,t-2}\Psi^\top_{1:m}\Xi^\top \Phi \textcolor{black}{y_{t-1}^{\ell}} - Z^\top_{m,t-2}\Psi^\top_{1:m}\Xi^\top \Xi \textcolor{black}{y_{t-1}^{\ell}} + Z^\top_{m,t-2}\Psi^\top_{1:m}\Xi^\top \Xi\Psi_{1:m}Z_{m,t-2}.\label{expansion_f}
\end{eqnarray}
\begin{itemize}
    \item[(i)] \textbf{cross-derivatives of the form $\partial^2_{\gamma_k\theta_{l,\textcolor{black}{T}}}$, $\partial^3_{\gamma_k\theta_{l,\textcolor{black}{T}}\theta_{j,\textcolor{black}{T}}}$ and $\partial^3_{\gamma_k\theta_{l,\textcolor{black}{T}}\gamma_j}$.} Let us apply the second order partial derivative $\partial^2_{\gamma_k\theta_{l,\textcolor{black}{T}}}$ for $k=1,\cdots,p$ (parameters in $c^*$) and $l=1,\cdots,d_T$ (the $l$-th parameter element of $\Psi_{1:m}$). We obtain
\beqw
\partial^2_{\gamma_k\theta_{l,\textcolor{black}{T}}}f(y_s,s\leq t;\theta_{\textcolor{black}{T}},\gamma) = -(\partial_{\gamma_k}c^{*\top})\Xi (\partial_{\theta_{l,\textcolor{black}{T}}}\Psi_{1:m}) Z_{m,t-2}-Z^\top_{m,t-2}(\partial_{\theta_{l,\textcolor{black}{T}}}\Psi^\top_{1:m})\Xi^\top(\partial_{\gamma_k}c^*).
\eeqw
For $k=p+1,\cdots,p+p^2$ (elements in $\Psi$), we have for any $l=1,\cdots,d_T$:
\beqw
\partial^2_{\gamma_k\theta_{l,\textcolor{black}{T}}}f(y_s,s\leq t;\theta_{\textcolor{black}{T}},\gamma) = -\textcolor{black}{y_{t-1}^{\ell\top}}(\partial_{\gamma_k}\Phi^\top)\Xi (\partial_{\theta_{l,\textcolor{black}{T}}}\Psi_{1:m}) Z_{m,t-2}-Z^\top_{m,t-2}(\partial_{\theta_{l,\textcolor{black}{T}}}\Psi_{1:m})\Xi^\top(\partial_{\gamma_k}\Phi).
\eeqw
Finally, for $k=p+p^2+1,\cdots,d_2$ (elements in $\Xi$), we have for any $k=1,\cdots,d_T$:
\begin{eqnarray*}
\lefteqn{\partial^2_{\gamma_k\theta_{l,\textcolor{black}{T}}}f(y_s,s\leq t;\theta_{\textcolor{black}{T}},\gamma) }\\
& = & \textcolor{black}{y_{t}^{\ell}}(\partial_{\gamma_k}\Xi)(\partial_{\theta_{l,\textcolor{black}{T}}}\Psi_{1:m})Z_{m,t-2}-c^{*\top}(\partial_{\gamma_k}\Xi)(\partial_{\theta_{l,\textcolor{black}{T}}}\Psi_{1:m})Z_{m,t-2}-\textcolor{black}{y_{t-1}^{\ell\top}}\Phi^\top(\partial_{\gamma_k}\Xi)(\partial_{\theta_{l,\textcolor{black}{T}}}\Psi_{1:m})Z_{m,t-2}\\ & & -2\textcolor{black}{y_{t-1}^{\ell\top}}\Xi^\top(\partial_{\gamma_k}\Xi)(\partial_{\theta_{l,\textcolor{black}{T}}}\Psi_{1:m})+Z^\top_{m,t-2}(\partial_{\theta_{l,\textcolor{black}{T}}}\Psi^\top_{1:m})(\partial_{\gamma_k}\Xi^\top)x_t-Z^\top_{m,t-2}(\partial_{\theta_{l,\textcolor{black}{T}}}\Psi^\top_{1:m})(\partial_{\gamma_k}\Xi^\top)c^* \\
& & - Z^\top_{m,t-2}(\partial_{\theta_{l,\textcolor{black}{T}}}\Psi^\top_{1:m})(\partial_{\gamma_k}\Xi^\top)\Phi \textcolor{black}{y_{t-1}^{\ell}}-2Z^\top_{m,t-2}(\partial_{\theta_{l,\textcolor{black}{T}}}\Psi^\top_{1:m})(\partial_{\gamma_k}\Xi^\top)\Xi \textcolor{black}{y_{t-1}^{\ell}}\\
& & + 4 Z^\top_{m,t-2}(\partial_{\theta_{l,\textcolor{black}{T}}}\Psi^\top_{1:m})(\partial_{\gamma_k}\Xi^\top)\Xi\Psi_{1:m}Z_{m,t-2}.
\end{eqnarray*}
Thus, the third order partial derivative $\partial^3_{\gamma_k\theta_{l,\textcolor{black}{T}}\theta_{j,\textcolor{black}{T}}}f(y_s,s\leq t;\theta_{\textcolor{black}{T}},\gamma)=0$ for $k \leq p+p^2$ (elements in $c^*$ and $\Psi$). For $k\geq p+p^2+1$ (elements in $\Xi$), then
\beqw
\partial^3_{\gamma_k\theta_{l,\textcolor{black}{T}}\theta_{j,\textcolor{black}{T}}}f(y_s,s\leq t;\theta_{\textcolor{black}{T}},\gamma) = 4 Z^\top_{m,t-2}(\partial_{\theta_{l,\textcolor{black}{T}}}\Psi^\top_{1:m})(\partial_{\gamma_k}\Xi^\top)\Xi(\partial_{\theta_{j,\textcolor{black}{T}}}\Psi_{1:m})Z_{m,t-2}.
\eeqw
Let us consider the cross derivative $\partial^3_{\gamma_k\theta_{l,\textcolor{black}{T}}\gamma_j}f(y_s,s\leq t;\theta_{\textcolor{black}{T}},\gamma)$. Using the formulas of $\partial^2_{\gamma_k\theta_{l,\textcolor{black}{T}}}f(y_s,s\leq t;\theta_{\textcolor{black}{T}},\gamma)$, we have for $k=1,\cdots,p$, $l=1,\cdots,d_T$ and any $j=1,\cdots,p+p^2$, 
\beqw
\partial^3_{\gamma_k\theta_{l,\textcolor{black}{T}}\gamma_j}f(y_s,s\leq t;\theta_{\textcolor{black}{T}},\gamma) =0,
\eeqw
and for $j=p+p^2+1,\cdots,d_2$,
\beqw
\partial^3_{\gamma_k\theta_{l,\textcolor{black}{T}}\gamma_j}f(y_s,s\leq t;\theta_{\textcolor{black}{T}},\gamma) =-(\partial_{\gamma_k}c^{*\top})(\partial_{\gamma_j}\Xi)(\partial_{\theta_{l,\textcolor{black}{T}}}\Psi_{1:m})Z_{m,t-2}-Z^\top_{m,t-2}(\partial_{\theta_{l,\textcolor{black}{T}}}\Psi^\top_{1:m})(\partial_{\gamma_j}\Xi^\top)(\partial_{\gamma_k}c^*).
\eeqw
When $k=p+1,\cdots,p+p^2$, $l=1,\cdots,d_T$ and any $j=1,\cdots,p+p^2$, 
\beqw
\partial^3_{\gamma_k\theta_{l,\textcolor{black}{T}}\gamma_j}f(y_s,s\leq t;\theta_{\textcolor{black}{T}},\gamma) =0,
\eeqw
and for $j=p+p^2+1,\cdots,d_2$, then
\beqw
\partial^3_{\gamma_k\theta_{l,\textcolor{black}{T}}\gamma_j}f(y_s,s\leq t;\theta_{\textcolor{black}{T}},\gamma) = -\textcolor{black}{y_{t-1}^{\ell\top}}(\partial_{\gamma_k}\Phi^\top)(\partial_{\gamma_j}\Xi) (\partial_{\theta_{l,\textcolor{black}{T}}}\Psi_{1:m}) Z_{m,t-2}-Z^\top_{m,t-2}(\partial_{\theta_{l,\textcolor{black}{T}}}\Psi_{1:m})(\partial_{\gamma_j}\Xi^\top)(\partial_{\gamma_k}\Phi).
\eeqw
Finally, for $k=p+p^2+1,\cdots,d_2$, $l=1,\cdots,d_T$ and $j=1,\cdots,p$, then
\beqw
\partial^3_{\gamma_k\theta_{l,\textcolor{black}{T}}\gamma_j}f(y_s,s\leq t;\theta_{\textcolor{black}{T}},\gamma) =-(\partial_{\gamma_j}c^{*\top}) (\partial_{\gamma_k}\Xi)(\partial_{\theta_{l,\textcolor{black}{T}}}\Psi_{1:m})Z_{m,t-2}-Z^\top_{m,t-2}(\partial_{\theta_{l,\textcolor{black}{T}}}\Psi^\top_{1:m})(\partial_{\gamma_k}\Xi^\top)(\partial_{\gamma_j}c^*).
\eeqw
For $k=p+p^2+1,\cdots,d_2$, $l=1,\cdots,d_T$ and $j=p+1,\cdots,p+p^2$,
{\small{\beqw
\partial^3_{\gamma_k\theta_{l,\textcolor{black}{T}}\gamma_j}f(y_s,s\leq t;\theta_{\textcolor{black}{T}},\gamma) = -\textcolor{black}{y_{t-1}^{\ell\top}}(\partial_{\gamma_j}\Phi^\top)(\partial_{\gamma_k}\Xi)(\partial_{\theta_{l,\textcolor{black}{T}}}\Psi_{1:m})Z_{m,t-2}-Z^\top_{m,t-2}(\partial_{\theta_{l,\textcolor{black}{T}}}\Psi^\top_{1:m})(\partial_{\gamma_k}\Xi^\top)(\partial_{\gamma_j}\Phi)\textcolor{black}{y_{t-1}^{\ell}}.
\eeqw}}
And for $k=p+p^2+1,\cdots,d_2$, $l=1,\cdots,d_T$ and $j=p+p^2,\cdots,d_2$, then
\begin{eqnarray*}
\lefteqn{\partial^3_{\gamma_k\theta_{l,\textcolor{black}{T}}\gamma_j}f(y_s,s\leq t;\theta_{\textcolor{black}{T}},\gamma) = -2\textcolor{black}{y_{t-1}^{\ell\top}}(\partial_{\gamma_j}\Xi^\top)(\partial_{\gamma_k}\Xi)(\partial_{\theta_{l,\textcolor{black}{T}}}\Psi_{1:m})}\\
&  & -2Z^\top_{m,t-2}(\partial_{\theta_{l,\textcolor{black}{T}}}\Psi^\top_{1:m})(\partial_{\gamma_k}\Xi^\top)(\partial_{\gamma_j}\Xi) \textcolor{black}{y_{t-1}^{\ell}} + 4 Z^\top_{m,t-2}(\partial_{\theta_{l,\textcolor{black}{T}}}\Psi^\top_{1:m})(\partial_{\gamma_k}\Xi^\top)(\partial_{\gamma_j}\Xi)\Psi_{1:m}Z_{m,t-2}.
\end{eqnarray*}

\item[(ii)] \textbf{cross-derivatives of the form $\partial^2_{\gamma_k\gamma_l}$ and $\partial^3_{\gamma_k\gamma_l\theta_{j,\textcolor{black}{T}}}$.} Using equation (\ref{expansion_f}), $\forall 1 \leq k \leq p, p+1\leq l \leq p+p^2, 1 \leq j \leq d_T$ and $\forall 1 \leq l \leq p, p+1\leq k \leq p+p^2, 1 \leq j \leq d_T$ , we have 
\beqw
\partial^3_{\gamma_k \gamma_l \theta_{j,\textcolor{black}{T}}}f(y_s,s\leq t;\theta_{\textcolor{black}{T}},\gamma) = 0,
\eeqw
and $\forall 1 \leq k \leq p, p+p^2+1\leq l \leq d_2, 1 \leq j \leq d_T$,
\beqw
\partial^3_{\gamma_k \gamma_l \theta_{j,\textcolor{black}{T}}}f(y_s,s\leq t;\theta_{\textcolor{black}{T}},\gamma) = -(\partial_{\gamma_k}c^{*})^\top (\partial_{\gamma_l}\Xi)(\partial_{\theta_{j,\textcolor{black}{T}}}\Psi_{1:m})Z_{m,t-2}-Z^\top_{m,t-2}(\partial_{\theta_{j,\textcolor{black}{T}}}\Psi_{1:m})^\top (\partial_{\gamma_l}\Xi)^\top (\partial_{\gamma_k}c^*).
\eeqw
By symmetry of the Hessian $\nabla^2_{\gamma\gamma^\top}f(y_s,s\leq t;\theta_{\textcolor{black}{T}},\gamma)$, we would obtain the same quantity for $\partial^3_{\gamma_l \gamma_k \theta_{j,\textcolor{black}{T}}}f(y_s,s\leq t;\theta_{\textcolor{black}{T}},\gamma)$. Now $\forall p+1\leq k \leq p+p^2, p+p^2+1 \leq l \leq d_2, 1 \leq j \leq d_T$,
{\small{\beqw
\partial^3_{\gamma_k \gamma_l \theta_{j,\textcolor{black}{T}}}f(y_s,s\leq t;\theta_{\textcolor{black}{T}},\gamma) = -\textcolor{black}{y_{t-1}^{\ell\top}}(\partial_{\gamma_k}\Phi)^\top (\partial_{\gamma_l}\Xi)(\partial_{\theta_{j,\textcolor{black}{T}}}\Psi_{1:m})Z_{m,t-2}-Z^\top_{m,t-2}(\partial_{\theta_{j,\textcolor{black}{T}}}\Psi_{1:m})^\top(\partial_{\gamma_l}\Xi)^\top (\partial_{\gamma_k}\Phi)\textcolor{black}{y_{t-1}^{\ell}}.
\eeqw}}
Again, by symmetry, we would obtain the same for $\partial^3_{\gamma_l \gamma_k \theta_{j,\textcolor{black}{T}}}f(y_s,s\leq t;\theta_{\textcolor{black}{T}},\gamma)$. Finally, $\forall p+p^2+1 \leq k,l \leq d_2, 1 \leq j \leq d_T$, we have
\begin{eqnarray*}
\lefteqn{\partial^3_{\gamma_k \gamma_l \theta_{j,\textcolor{black}{T}}}f(y_s,s\leq t;\theta_{\textcolor{black}{T}},\gamma) }\\
& = & -\textcolor{black}{y_{t-1}^{\ell\top}}(\partial_{\gamma_k}\Xi)^\top (\partial_{\gamma_l}\Xi) (\partial_{\theta_{j,\textcolor{black}{T}}}\Psi_{1:m})Z_{m,t-2}-\textcolor{black}{y_{t-1}^{\ell\top}}(\partial_{\gamma_l}\Xi)^\top (\partial_{\gamma_k}\Xi) (\partial_{\theta_{j,\textcolor{black}{T}}}\Psi_{1:m})Z_{m,t-2}\\
& & -Z_{m,t-2}(\partial_{\theta_{j,\textcolor{black}{T}}}\Psi_{1:m})^\top (\partial_{\gamma_k}\Xi)^\top (\partial_{\gamma_l}\Xi)\textcolor{black}{y_{t-1}^{\ell}}-Z^\top_{m,t-2}(\partial_{\theta_{j,\textcolor{black}{T}}}\Psi_{1:m})(\partial_{\gamma_l}\Xi)^\top(\partial_{\gamma_k}\Xi)\textcolor{black}{y_{t-1}^{\ell}} \\
& & + Z_{m,t-2}(\partial_{\theta_{j,\textcolor{black}{T}}}\Psi_{1:m})^\top (\partial_{\gamma_k}\Xi)^\top(\partial_{\gamma_l}\Xi)\Psi_{1:m}Z_{m,t-2} + Z_{m,t-2}\Psi^\top_{1:m} (\partial_{\gamma_k}\Xi)^\top(\partial_{\gamma_l}\Xi)(\partial_{\theta_{j,\textcolor{black}{T}}}\Psi_{1:m})Z_{m,t-2} \\
& & + Z_{m,t-2}(\partial_{\theta_{j,\textcolor{black}{T}}}\Psi_{1:m})^\top (\partial_{\gamma_l}\Xi)^\top(\partial_{\gamma_k}\Xi)\Psi_{1:m}Z_{m,t-2} + Z_{m,t-2}\Psi^\top_{1:m} (\partial_{\gamma_l}\Xi)^\top(\partial_{\gamma_k}\Xi)(\partial_{\theta_{j,\textcolor{black}{T}}}\Psi_{1:m})Z_{m,t-2}.
\end{eqnarray*}

\end{itemize}

\section{\textcolor{black}{Proofs}}\label{appendix_proof}

\textcolor{black}{We first provide the assumptions} we relied on in the large sample analysis of the first step penalized estimator $\widehat{\theta}_{\textcolor{black}{T}}$.

\medskip

{The following assumption is the so-called sparsity assumption. As the true support $\Ac_{\textcolor{black}{T}}$ is unknown, we rely on penalized M-estimation to obtain an estimation of this set.
\begin{assumption}\label{assumption_sparsity}
$\normalfont \text{card}(\Ac_{\textcolor{black}{T}}) = k_{T} < d_T$ with $\Ac_T = \normalfont \text{supp}(\theta_{0,\textcolor{black}{T}}) := \{i=1,\cdots,d_T: \theta_{0,i,\textcolor{black}{T}} \neq 0\}$.
\end{assumption}}

{The next assumption ensures the stability and invertibility of the process $(x_t)$.
\begin{assumption}\label{assumption_stationarity}
Let the AR and MA operators defined as $\Phi(L)=I_p-\Phi L$ and $\Xi(L)=I_p+\Xi L$ with $L$ the lag operator. We assume that the model is stable, i.e. $\normalfont\text{det}(\Phi(z))\neq 0$, and invertible, i.e. $\normalfont\text{det}(\Xi(z))\neq 0$ for any $|z|\leq 1$, $z \in \Cb$.
\end{assumption}}

{To ensure suitable regularity conditions on the non-penalized loss, we assume the following:
\begin{assumption}\label{assumption_var_cov}
\textcolor{black}{The $d_T \times d_T$ matrices $\Hb := \Eb[\nabla^2_{\theta_{T} \theta^\top_T}\ell(y_s, s \leq t;\theta_{0,\textcolor{black}{T}})]$ and $\Mb := \Eb[\nabla_{\theta_{\textcolor{black}{T}}}\ell(y_s,s\leq t;\theta_{0,\textcolor{black}{T}}) \nabla_{\theta^\top_{\textcolor{black}{T}}}\ell(y_s,s\leq t;\theta_{0,\textcolor{black}{T}})]$} exist and are positive-definite. There exist $b_1,b_2$ with $0 < b_1 < b_2 < \infty$ and $c_1,c_2$ with $0 < c_1 < c_2 < \infty$ such that
\beqw
b_1 < \lambda_{\min}(\Mb) < \lambda_{\max}(\Mb) < b_2, \;\; \text{and} \;\; c_1 < \lambda_{\min}(\Hb) < \lambda_{\max}(\Hb) < c_2.
\eeqw
Let $\Vb = \Hb^{-1} \Mb \Hb^{-1}$. Then there exist $a_1,a_2$ with $0 < a_1 < a_2 < \infty$ such that for any $T$, we have $a_1 < \lambda_{\min}(\Vb) < \lambda_{\max}(\Vb) < a_2$.
\end{assumption}
\begin{assumption}\label{assumption_var_cov2}
For every $d_T$ and thus $T$, $\Eb[\big\{\nabla_{\theta_{\textcolor{black}{T}}}\ell(y_s, s\leq t;\theta_{0,\textcolor{black}{T}}) \nabla_{\theta^\top_{\textcolor{black}{T}}}\ell(y_s, s\leq t;\theta_{0,\textcolor{black}{T}})\big\}^2] < \infty$.
\end{assumption}
\begin{assumption}\label{assumption_gradient}
Let $\partial_{\theta_{k,\textcolor{black}{T}}}\ell(y_s,s\leq t;\theta_{0,\textcolor{black}{T}}) = -\big(Z_{m,t-1}\otimes \{x_t-\Psi_{0,1:m} Z_{m,t-1}\}\big)_k$, for any $k=1,\cdots,d_T$. Then there is some function $\Psi(.)$ such that for any $T$: 
\beqw
\underset{k=1,\cdots,d_T}{\sup}\Eb[\partial_{\theta_{k,\textcolor{black}{T}}}\ell(y_s,s\leq t;\theta_{0,\textcolor{black}{T}}) \partial_{\theta_{k,\textcolor{black}{T}}}\ell(y_s,s\leq t';\theta_{0,\textcolor{black}{T}})]\leq \Psi(|t-t'|), \;\; \text{and} \;\; \underset{T>0}{\sup}\; \frac{1}{T}\overset{T}{\underset{t,t'=1}{\sum}} \Psi(|t-t'|)<\infty.
\eeqw
\end{assumption}
\begin{assumption}\label{assumption_second_derivative}
Let $\partial^2_{\theta_{k,\textcolor{black}{T}}\theta_{l,\textcolor{black}{T}}}\ell(y_s,s\leq t;\theta_{0,\textcolor{black}{T}}) = \big(Z_{m,t-1}Z^\top_{m,t-1} \otimes I_p\big)_{kl}$ for any $k,l=1,\cdots,d_T$. Let $\zeta_{kl,t} = \partial^2_{\theta_{k,\textcolor{black}{T}} \theta_{l,\textcolor{black}{T}}} \ell(y_s,s\leq t;\theta_{0,\textcolor{black}{T}}) - \Eb[\partial^2_{\theta_{k,\textcolor{black}{T}} \theta_{l,\textcolor{black}{T}}} \ell(y_s,s\leq t;\theta_{0,\textcolor{black}{T}})]$. Then there exists some function $\chi(.)$ such that for any $T$: 
\beqw
|\Eb[\zeta_{kl,t} \zeta_{k'l',t'}]| \leq \chi(|t-t'|), \;\; \text{and} \;\; \underset{T>0}{\sup} \; \frac{1}{T}\overset{T}{\underset{t,t'=1}{\sum}} \chi(|t-t'|) <\infty.
\eeqw
\end{assumption}
These assumptions on the loss deserve a few comments. Assumption \ref{assumption_stationarity} concerns the probabilistic property of the process. Assumptions \ref{assumption_var_cov}-\ref{assumption_var_cov2} are similar to condition (F) of Fan and Peng (2004) and allow for controlling the minimum and maximum eigenvalues of the limits of the empirical Hessian and the score cross-product. Assumptions \ref{assumption_gradient}-\ref{assumption_second_derivative} are moment conditions, which may be somewhat arbitrary but ensure the convergence of the empirical gradient and Hessian to their population level counterparts by an application of the Markov inequality. As for the \textcolor{black}{penalty functions} SCAD and MCP, we consider the following conditions:
\begin{assumption}\label{assumption_folded}
\begin{itemize}
    \item[(i)] Let $\lambda$ a fixed non-negative scalar, let $x \geq 0$, then $\lambda^{-1}\pp(\lambda,x)$ is increasing and concave in $x$, has continuous derivative $\lambda^{-1}\partial_x \pp(\lambda,x)$ with $\underset{x \rightarrow 0^+}{\lim} \lambda^{-1}\partial_x \pp(\lambda,x) > 0$. Moreover, $\underset{x \rightarrow 0^+}{\lim} \lambda^{-1}\partial_x \pp(\lambda,x)$ does not depend on $\lambda$.
    \item[(ii)] Let $A_{1,T} = \underset{ k \in \Ac_{\textcolor{black}{T}}}{\max}|\partial_{\theta_{k,\textcolor{black}{T}}}\pp(\frac{\lambda_T}{T},|\theta_{0,k,\textcolor{black}{T}}|)|$, and $A_{2,T} = \underset{k \in \Ac_{\textcolor{black}{T}}}{\max}|\partial^2_{\theta_{k,\textcolor{black}{T}} \theta_{k,\textcolor{black}{T}}}\pp(\frac{\lambda_T}{T},|\theta_{0,k,\textcolor{black}{T}}|)|$. Then $A_{1,T} = O(T^{-1/2})$ and $A_{2,T} \rightarrow 0$.
    \textcolor{black}{\item[(iii)] There exist constants $M$ and $D$ such that $| \partial^2_{\theta\theta} \pp(\lambda,\theta_1) - \partial^2_{\theta \theta} \pp(\lambda,\theta_2) |\leq D |\theta_1 -\theta_2 |$,
    for any real numbers $\theta_1,\theta_2$ such that $\theta_1,\theta_2 > M\lambda$.
    \item[(iv)] $\frac{\lambda_T}{\sqrt{T d_T}} \rightarrow \infty$ as $T \rightarrow \infty$ holds and $\underset{T \rightarrow \infty}{\lim} \; \underset{x \rightarrow 0^+}{\lim \, \inf} \; \frac{T}{\lambda_T}\partial_x\pp(\frac{\lambda_T}{T},x) > 0$.}
    \item[(v)] $\exists a_1, a_2$ with $0 < a_1 < a_2 < \infty$ and $a_1 \leq \underset{k \in \Ac_{\textcolor{black}{T}}}{\min}|\theta_{0,k,\textcolor{black}{T}}| \leq \underset{k \in \Ac_{\textcolor{black}{T}}}{\max}|\theta_{0,k,\textcolor{black}{T}}| \leq a_2$, and $\underset{k \in \Ac_{\textcolor{black}{T}}}{\min}|\theta_{0,k,\textcolor{black}{T}}| / (\frac{\lambda_T}{T}) \rightarrow \infty$ as $T \rightarrow \infty$.
\end{itemize} 
\end{assumption}
These assumptions are standard in sparse analysis using folded concave \textcolor{black}{penalty functions}: see, e.g., Fan and Peng (2004). Assumption \ref{assumption_folded}-(i), (iii) provide some regularity conditions on the penalty function. Assumption \ref{assumption_folded}-(ii) implies that the penalty has less influence than the non-penalized loss function in the regularised problem. This is a key unbiasedness property for large parameters and ensures the existence of the $\sqrt{T/d_T}$-consistent penalized estimator. Assumption \ref{assumption_folded}-(iv) makes the penalty function singular at the origin so that the
sparse $\widehat{\theta}_{\textcolor{black}{T}}$ satisfies the sparsity property in the oracle Theorem. Assumption \ref{assumption_folded}-(v) is a beta-min assumption on the minimum signal.}

The next assumption may be artificial, but it is key to obtain the oracle property and is in line with Assumption (H) of Fan and Peng (2004).
\begin{assumption}\label{assumption_array}
Let $X_{T,t}=\sqrt{T}Q_T\Vb^{\textcolor{black}{-1/2}}_{\Ac_{\textcolor{black}{T}}\Ac_{\textcolor{black}{T}}}\Hb^{-1}_{\Ac_{\textcolor{black}{T}}\Ac_{\textcolor{black}{T}}}\nabla_{\theta_{\Ac_{\textcolor{black}{T}}}}\Gb_{T,t}(\underline{y};\theta_{0,\textcolor{black}{T}})$ with $(Q_T)$ a sequence of $r \times \normalfont\text{card}(\Ac_{\textcolor{black}{T}})$ matrices such that $ Q_T Q^\top_T \overset{\Pb}{\underset{T \rightarrow \infty}{\longrightarrow}} \mathbb{C}$ for some positive definite symmetric matrix $\mathbb{C}$, $\Vb_{\Ac_{\textcolor{black}{T}}\Ac_{\textcolor{black}{T}}} = \big(\Hb^{-1} \Mb \Hb^{-1}\big)_{\Ac_{\textcolor{black}{T}}\Ac_{\textcolor{black}{T}}}$ and $\nabla_{\theta_{\Ac_{\textcolor{black}{T}}}}\Gb_{T,t}(\underline{y};\theta_{0,\textcolor{black}{T}}) = \frac{1}{T}\big(-\big(Z_{m,t-1}\otimes \{x_t-\Psi_{0,1:m} Z_{m,t-1}\}\big)_{\Ac_{\textcolor{black}{T}}}\big) \in \Rb^{k_T}$. Let $\Fc^T_t = \sigma(X_{T,s},s\leq t)$, then $X_{T,t}$ is a martingale difference and we have
\beqw
\Eb\Big[\underset{1 \leq k,l \leq d_T}{\sup}\Eb[\big\{\partial_{\theta_{k,\textcolor{black}{T}}}\ell(y_s,s\leq t;\theta_{0,\textcolor{black}{T}})\partial_{\theta_{l,\textcolor{black}{T}}}\ell(y_s,s\leq t;\theta_{0,\textcolor{black}{T}})\big\}^2|\Fc^T_{t-1}]\lambda_{\max,t-1}(\Hb^T_{t-1}) \Big] \leq \overline{B}<\infty,
\eeqw
with $\Hb^T_{t-1} = \Eb[\nabla_{\theta_{\textcolor{black}{T}}}\ell(y_s,s\leq t;\theta_{0,\textcolor{black}{T}})\nabla_{\theta^\top_{\textcolor{black}{T}}}\ell(y_s,s\leq t;\theta_{0,\textcolor{black}{T}})|\Fc^T_{t-1}]$ and $\lambda_{\max,t-1}(\Hb^T_{t-1})<\infty$.
\end{assumption}
Assumption \ref{assumption_array} is key to \textcolor{black}{verify} the Lindeberg condition to apply Theorem \ref{lindeberg_shiryaev} of Shiryaev (1991) when dealing with dependent variables.

\begin{proof}[Proof of Theorem \ref{bound_prob_asym}.]
Let $\nu_T = \sqrt{d}_T\big(T^{-1/2} + R_T\big)$, where $R_T$ will be made explicit depending on the penalty case. We would like to prove that for any $\eps > 0$, there exists $C_{\eps} > 0$ such that
\beqw
\Pb(\cfrac{1}{\nu_T}\|\widehat{\theta}_{\textcolor{black}{T}} - \theta_{0,\textcolor{black}{T}}\|_2 > C_{\eps}) < \eps.
\eeqw
Following the reasoning of Fan and Li (2001), Theorem 1, and denoting $\Gb^{\text{pen}}_T(\underline{y};\theta) = \Gb_T(\underline{y};\theta)+\overset{d_T}{\underset{i=1}{\sum}}\pp(\frac{\lambda_T}{T},|\theta_{i,\textcolor{black}{T}}|)$, we have
\beqw
\Pb(\cfrac{1}{\nu_T} \|\widehat{\theta}_{\textcolor{black}{T}} - \theta_{0,\textcolor{black}{T}}\|_2 > C_{\eps}) \leq \Pb(\exists \uu, \|\uu\|_2 = C_{\eps}: \Gb^{\text{pen}}_T(\underline{y};\theta_{0,\textcolor{black}{T}}+\nu_T\uu) \leq \Gb^{\text{pen}}_T(\underline{y};\theta_{0,\textcolor{black}{T}})),
\eeqw
which implies that there is a local minimum in the ball $\{\theta_{0,\textcolor{black}{T}}+\nu_T\uu, \|\uu\|_2 \leq C_{\eps}\}$ so that the minimum $\widehat{\theta}_{\textcolor{black}{T}}$ satisfies $\|\widehat{\theta}_{\textcolor{black}{T}} - \theta_{0,\textcolor{black}{T}}\|_2 = O_p(\nu_T)$. Now by a Taylor expansion of the penalized loss function, we obtain
\begin{eqnarray*}
\lefteqn{\Gb^{\text{pen}}_T(\underline{y};\theta_{0,\textcolor{black}{T}}+\nu_T\uu)-\Gb^{\text{pen}}_T(\underline{y};\theta_{0,\textcolor{black}{T}})}\\
& = & \nu_T \uu^\top \nabla_{\theta_{\textcolor{black}{T}}} \Gb_T(\underline{y};\theta_{0,\textcolor{black}{T}}) + \frac{\nu^2_T}{2} \uu^\top \nabla^2_{\theta \theta^\top_{\textcolor{black}{T}}} \Gb_T(\underline{y};\theta_{0,\textcolor{black}{T}})\uu 
+ \overset{d_T}{\underset{k=1}{\sum}}\big\{\pp(\frac{\lambda_T}{T},|\theta_{0,k,\textcolor{black}{T}}+\nu_T\uu_k|)-\pp(\frac{\lambda_T}{T},|\theta_{0,k,\textcolor{black}{T}}|)\big\},
\end{eqnarray*}
since the third derivative vanishes. We want to prove
\beq \label{bound_obj}
\begin{array}{llll}
\Pb(\exists \uu, \|\uu\|_2 = C_{\eps} &: &\uu^\top \nabla_{\theta_{\textcolor{black}{T}}} \Gb_T(\underline{y};\theta_{0,\textcolor{black}{T}}) + \frac{\nu_T}{2} \uu^\top \Hb \uu + \frac{\nu_T}{2} \Rc_T(\theta_{0,\textcolor{black}{T}}) \\
& + & \nu^{-1}_T\overset{d_T}{\underset{k=1}{\sum}}\big\{\pp(\frac{\lambda_T}{T},|\theta_{0,k,\textcolor{black}{T}}+\nu_T\uu_k|)-\pp(\frac{\lambda_T}{T},|\theta_{0,k,\textcolor{black}{T}}|)\big\} \leq 0) < \eps,
\end{array}
\eeq
where $\Rc_T(\theta_{0,\textcolor{black}{T}}) = \uu^\top \big\{\nabla^2_{\theta \theta^\top_{\textcolor{black}{T}}} \Gb_T(\underline{y};\theta_{0,\textcolor{black}{T}}) - \Hb\big\} \uu$. First, for $a>0$ and the Markov inequality, we have for the score term
\begin{eqnarray*}
\lefteqn{\Pb(\underset{\uu:\|\uu\|_2=C_{\eps}}{\sup}|\uu^{\top}\nabla_{\theta_{\textcolor{black}{T}}} \Gb_T(\underline{y};\theta_{0,\textcolor{black}{T}})| > a) }\\
&\leq &\Pb(\underset{\uu:\|\uu\|_2=C_{\eps}}{\sup}\|\uu\|_2\|\nabla_{\theta_{\textcolor{black}{T}}} \Gb_T(\underline{y};\theta_{0,\textcolor{black}{T}})\|_2 > a) \\
&\leq &\Pb(\|\nabla_{\theta_{\textcolor{black}{T}}} \Gb_T(\underline{y};\theta_{0,\textcolor{black}{T}})\|_2 > \frac{a}{C_{\eps}}) \\
&\leq & \big(\frac{C_{\eps}}{a}\big)^2 \Eb[\|\nabla_{\theta_{\textcolor{black}{T}}} \Gb_T(\underline{y};\theta_{0,\textcolor{black}{T}})\|^2_2] \\
& \leq & \big(\frac{C_{\eps}}{a}\big)^2 \overset{d_T}{\underset{k=1}{\sum}} \Eb[\big(\partial_{\theta_{k,\textcolor{black}{T}}} \Gb_T(\underline{y};\theta_{0,\textcolor{black}{T}})\big)^2] \\
& = & \big(\frac{C_{\eps}}{a}\big)^2 \frac{1}{T^2}\overset{T}{\underset{t,t'=1}{\sum}} \overset{d_T}{\underset{k=1}{\sum}} \Eb[\partial_{\theta_{k,\textcolor{black}{T}}}\ell(y_s,s\leq t;\theta_{0,\textcolor{black}{T}}) \partial_{\theta_{k,\textcolor{black}{T}}}\ell(y_s,s\leq t';\theta_{0,\textcolor{black}{T}})] \\
& \leq & \big(\frac{C_{\eps}}{a}\big)^2 d_T \frac{1}{T^2} \overset{T}{\underset{t,t'=1}{\sum}} \Psi(|t-t'|),
\end{eqnarray*}
where $\partial_{\theta_{k,\textcolor{black}{T}}}\ell(y_s,s\leq t;\theta_{0,\textcolor{black}{T}}) = -\big(Z_{m,t-1}\otimes \{x_t-\Psi_{0,1:m} Z_{m,t-1}\}\big)_k$ for any $k=1,\cdots,d_T$. By Assumption \ref{assumption_gradient}, $\underset{k=1,\cdots,d_T}{\sup}\Eb[\partial_{\theta_{k,\textcolor{black}{T}}}\ell(y_s,s\leq t;\theta_{0,\textcolor{black}{T}}) \partial_{\theta_{k,\textcolor{black}{T}}}\ell(y_s,s\leq t';\theta_{0,\textcolor{black}{T}})]\leq \Psi(|t-t'|)$ and $\frac{1}{T}\overset{T}{\underset{t,t'=1}{\sum}} \Psi(|t-t'|)<\infty$. As a consequence,
\beqw
\Pb(\underset{\uu:\|\uu\|_2=C_{\eps}}{\sup}|\uu^{\top}\nabla_{\theta_{\textcolor{black}{T}}} \Gb_T(\underline{y};\theta_{0,\textcolor{black}{T}})| > a) \leq \frac{K_1 C^2_{\eps}d_T}{a^2T},
\eeqw
for $K_1>0$ a finite constant. We now focus on the hessian quantity that can be rewritten as
\beqw
\uu^\top \nabla^2_{\theta_{\textcolor{black}{T}} \theta^\top_{\textcolor{black}{T}}} \Gb_T(\underline{y};\theta_{0,\textcolor{black}{T}})\uu = \uu^\top \Eb[\nabla^2_{\theta_{\textcolor{black}{T}} \theta^\top_{\textcolor{black}{T}}} \Gb_T(\underline{y};\theta_{0,\textcolor{black}{T}})]\uu + \Rc_T(\theta_{0,\textcolor{black}{T}}), 
\eeqw
where $\Rc_T(\theta_{0,\textcolor{black}{T}})= \overset{d_T}{\underset{k,l=1}{\sum}}\uu_k\uu_l\big\{\partial^2_{\theta_{k,\textcolor{black}{T}}\theta_{l,\textcolor{black}{T}}}\Gb_T(\underline{y};\theta_{0,\textcolor{black}{T}})-\Eb[\partial^2_{\theta_{k,\textcolor{black}{T}}\theta_{l,\textcolor{black}{T}}}\Gb_T(\underline{y};\theta_{0,\textcolor{black}{T}})]\big\}$. Its two first moments satisfy $\Eb[\Rc_T(\theta_{0,\textcolor{black}{T}})]=0$ and 
\beqw
\text{Var}(\Rc_T(\theta_{0,\textcolor{black}{T}})) = \frac{1}{T^{\textcolor{black}{2}}} \overset{T}{\underset{t,t'=1}{\sum}}  \overset{d_T}{\underset{k,k',l,l'=1}{\sum}}\uu_k \uu_l \uu_{k'} \uu_{l'} \Eb[\zeta_{kl,t} \zeta_{k'l',t'}], 
\eeqw
where $\zeta_{kl,t} = \partial^2_{\theta_{k,\textcolor{black}{T}} \theta_{l,\textcolor{black}{T}}} \ell(y_s,s\leq t;\theta_{0,\textcolor{black}{T}}) - \Eb[\partial^2_{\theta_{k,\textcolor{black}{T}} \theta_{l,\textcolor{black}{T}}} \ell(y_s,s\leq t;\theta_{0,\textcolor{black}{T}})]$ with $\partial^2_{\theta_{k,\textcolor{black}{T}} \theta_{l,\textcolor{black}{T}}} \ell(y_s,s\leq t;\theta_{0,\textcolor{black}{T}})=(Z_{m,t-1} Z^\top_{m,t-1} \otimes I_p)_{kl}$. Let $b>0$, we have by the Markov inequality and Assumption \ref{assumption_second_derivative},
\beqw
\Pb(|\Rc_T(\theta_{0,\textcolor{black}{T}})|>b) \leq \frac{1}{b^2} \Eb[\Rc^2_T(\theta_{0,\textcolor{black}{T}})]\leq \frac{K_2 \|\uu\|^4_2 d^2_T}{T b^2} \leq \frac{C^4_{\eps} K_2 d^2_T}{T b^2},
\eeqw
for some constant $K_2>0$. By Assumption \ref{assumption_var_cov},
\beqw
\uu^\top \Eb[\nabla^2_{\theta_{\textcolor{black}{T}} \theta^\top_{\textcolor{black}{T}}} \Gb_T(\underline{y};\theta_{0,\textcolor{black}{T}})]\uu \geq \lambda_{\min}(\Hb_n)\uu^{\top} \uu.
\eeqw
Let us now consider the penalization part. First, note that
\beqw
\overset{d_T}{\underset{k=1}{\sum}}\big\{\pp(\frac{\lambda_T}{T},|\theta_{0,k,\textcolor{black}{T}}+\nu_T\uu_k|)-\pp(\frac{\lambda_T}{T},|\theta_{0,k,\textcolor{black}{T}}|) \big\} \geq \underset{k\in \Ac_T}{\sum}\big\{\pp(\frac{\lambda_T}{T},|\theta_{0,k,\textcolor{black}{T}}+\nu_T\uu_k|)-\pp(\frac{\lambda_T}{T},|\theta_{0,k,\textcolor{black}{T}}|)\big\}.
\eeqw
For $a>2$, the SCAD penalty is
\beqw
\forall k \in \Ac_T, \; \pp(\frac{\lambda_T}{T},|\theta_{0,k,\textcolor{black}{T}}|) = \begin{cases}
\frac{\lambda_T}{T} |\theta_{0,k,\textcolor{black}{T}}|, & \text{for} \;  |\theta_{0,k,\textcolor{black}{T}}| \leq \frac{\lambda_T}{T}, \\
-\frac{1}{(2(a-1))}(\theta_{0,k,\textcolor{black}{T}}^2-2 a\frac{\lambda_T}{T}|\theta_{0,k,\textcolor{black}{T}}|+(\frac{\lambda_T}{T})^2), &  \text{for} \; \frac{\lambda_T}{T} \leq |\theta_{0,k,\textcolor{black}{T}}| \leq a \frac{\lambda_T}{T}, \\
(a+1)(\frac{\lambda_T}{T})^2/2, & \text{for} \;  |\theta_{0,k,\textcolor{black}{T}}| > a\frac{\lambda_T}{T},
\end{cases}
\eeqw
so that the derivative is given as:
\beqw
\forall k \in \Ac_T, \; \partial_{\theta_{k,\textcolor{black}{T}}}\pp(\frac{\lambda_T}{T},|\theta_{0,k,\textcolor{black}{T}}|) = \begin{cases}
\frac{\lambda_T}{T}, & \text{for} \;  |\theta_{0,k,\textcolor{black}{T}}| \leq \frac{\lambda_T}{T}, \\
\frac{a\frac{\lambda_T}{T}-|\theta_{0,k,\textcolor{black}{T}}|}{a-1}, &  \text{for} \; \frac{\lambda_T}{T} \leq |\theta_{0,k,\textcolor{black}{T}}| \leq a \frac{\lambda_T}{T}, \\
0, & \text{for} \;  |\theta_{0,k,\textcolor{black}{T}}| > a\frac{\lambda_T}{T}.
\end{cases}
\eeqw
This derivative can be written compactly as
\beqw
\forall k \in \Ac_T, \; \partial_{\theta_{k,\textcolor{black}{T}}}\pp(\frac{\lambda_T}{T},|\theta_{0,k,\textcolor{black}{T}}|) = \frac{\lambda_T}{T}\big\{\mathbf{1}_{\{|\theta_{0,k,\textcolor{black}{T}}|\leq \frac{\lambda_T}{T}\}}+\frac{(a\frac{\lambda_T}{T}-|\theta_{0,k,\textcolor{black}{T}}|)^+}{(a-1)\frac{\lambda_T}{T}}\mathbf{1}_{\{|\theta_{0,k,\textcolor{black}{T}}|>\frac{\lambda_T}{T}\}}\big\}
\eeqw
As a consequence, the SCAD penalty is twice continuously differentiable, its second derivative is zero unless some components $|\theta_{0,k,\textcolor{black}{T}}|$ take values in $[\frac{\lambda_T}{T},a\frac{\lambda_T}{T}]$ and $\forall k \in \Ac_T,\partial^2_{\theta_{k,\textcolor{black}{T}} \theta_{k,\textcolor{black}{T}}}\pp(\frac{\lambda_T}{T},|\theta_{0,k,\textcolor{black}{T}}|) \rightarrow 0$ as $T\rightarrow \infty$ under the scaling $\lambda_T=o(T)$. The MCP behaves similarly to the SCAD as it is a quadratic spline and is defined as
\beqw
\forall k \in \Ac_T, \; \pp(\frac{\lambda_T}{T},|\theta_{0,k,\textcolor{black}{T}}|) = 
\begin{cases}
\frac{\lambda_T}{T}|\theta_{0,k,\textcolor{black}{T}}| - \frac{\theta^2_{0,k}}{b}, & \text{for} \; |\theta_{0,k,\textcolor{black}{T}}| \leq b\frac{\lambda_T}{T}, \\
\frac{1}{2}b (\frac{\lambda_T}{T})^2 , & \text{for} \; |\theta_{0,k,\textcolor{black}{T}}| > b\frac{\lambda_T}{T},
\end{cases}
\eeqw
so that the derivative is given as
\beqw
\forall k \in \Ac_T, \; \partial_{\theta_{k,\textcolor{black}{T}}}\pp(\frac{\lambda_T}{T},|\theta_{0,k,\textcolor{black}{T}}|) = 
\begin{cases}
(\frac{\lambda_T}{T}-\frac{|\theta_{0,k,\textcolor{black}{T}}|}{b})\text{sgn}(\theta_{0,k,\textcolor{black}{T}}), & \text{for} \; |\theta_{0,k,\textcolor{black}{T}}| \leq b\frac{\lambda_T}{T}, \\
0, & \text{for} \; |\theta_{0,k,\textcolor{black}{T}}| > b\frac{\lambda_T}{T},
\end{cases}
\eeqw
which can be expressed as $\forall k \in \Ac_T, \; \partial_{\theta_{k,\textcolor{black}{T}}}\pp(\frac{\lambda_T}{T},|\theta_{0,k,\textcolor{black}{T}}|) = \frac{1}{b}(b\frac{\lambda_T}{T}-|\theta_{0,k,\textcolor{black}{T}}|)^+$.
Thus, under $\lambda_T = o(T)$, we obtain $\partial^2_{\theta_{k,\textcolor{black}{T}} \theta_{k,\textcolor{black}{T}}}\pp(\frac{\lambda_T}{T},|\theta_{0,k,\textcolor{black}{T}}|) \rightarrow 0$ when $T \rightarrow \infty$. Now for any $k \in \Ac_{\textcolor{black}{T}} \subset \{1,\cdots,d_T\}$, we have
\begin{eqnarray*}
\lefteqn{\pp(\frac{\lambda_T}{T},|\theta_{0,k,\textcolor{black}{T}} + \nu_T  u_k|)-\pp(\frac{\lambda_T}{T},|\theta_{0,k,\textcolor{black}{T}}|)}\\
& = & \nu_T  u_k \textnormal{sgn}(\theta_{0,k,\textcolor{black}{T}}) \partial_{\theta_{k,\textcolor{black}{T}}}\pp(\frac{\lambda_T}{T},|\theta_{0,k,\textcolor{black}{T}}|)
+\frac{\nu^2_T}{2} u^2_k \partial^2_{\theta_{k,\textcolor{black}{T}} \theta_{k,\textcolor{black}{T}}}\pp(\frac{\lambda_T}{T},|\theta_{0,k,\textcolor{black}{T}}|) \big(1+o(1)\big).
\end{eqnarray*}
Hence, we obtain
\beqw
|\underset{k \in \Ac_{\textcolor{black}{T}}}{\sum} \pp(\frac{\lambda_T}{T},|\theta_{0,k,\textcolor{black}{T}}+ \nu_T  u_k|)-\pp(\frac{\lambda_T}{T},|\theta_{0,k,\textcolor{black}{T}}|) | \leq \nu_T \|\uu\|_1 A_{1,\textcolor{black}{T}} + \frac{\nu^2_T}{2} \|\uu\|^2_2 A_{2,\textcolor{black}{T}} \big(1+o(1)\big).
\eeqw
Using $\|\uu\|_1 \leq \sqrt{\textnormal{card}(\Ac_{\textcolor{black}{T}})} \|\uu\|_2$, we obtain
\beqw
|\underset{k \in \Ac_{\textcolor{black}{T}}}{\sum} \pp(\frac{\lambda_T}{T},|\theta_{0,k,\textcolor{black}{T}}+ \nu_T  u_k|)-\pp(\frac{\lambda_T}{T},|\theta_{0,k,\textcolor{black}{T}}|) | \leq \nu_T \sqrt{\textnormal{card}(\Ac_{\textcolor{black}{T}})} \|\uu\|_2 A_{1,\textcolor{black}{T}} + \nu^2_T \|\uu\|^2_2 A_{2,\textcolor{black}{T}}.
\eeqw
Finally, for the LASSO penalty, we have
\beqw
|\underset{k \in \Ac_{\textcolor{black}{T}}}{\sum} \pp(\frac{\lambda_T}{T},|\theta_{0,k,\textcolor{black}{T}}+ \nu_T  u_k|)-\pp(\frac{\lambda_T}{T},|\theta_{0,k,\textcolor{black}{T}}|) | =| \underset{k \in \Ac_{\textcolor{black}{T}}}{\sum}\frac{\lambda_T}{T}\big(|\theta_{0,k,\textcolor{black}{T}}+ \nu_T  u_k|-|\theta_{0,k,\textcolor{black}{T}}|\big)| \leq \nu_T\frac{\lambda_T}{T}\sqrt{\text{card}(\Ac_{\textcolor{black}{T}})}\|\uu\|_2.
\eeqw
Then, denoting $\delta_T = \lambda_{\min}(\Hb) C^2_{\eps} \nu_T/2$, and using $\frac{\nu_T}{2}\Eb[\uu^\top\nabla^2_{\theta \theta^\top_{\textcolor{black}{T}}} \ell(y_s,s\leq t;\theta_{0,\textcolor{black}{T}})\uu] \geq \delta_T$, we deduce that (\ref{bound_obj}) can be bounded as
\begin{eqnarray*}
\lefteqn{\Pb(\exists \uu, \|\uu\|_2 = C_{\eps} :\uu^\top \nabla_{\theta_{\textcolor{black}{T}}} \Gb_T(\underline{y};\theta_{0,\textcolor{black}{T}}) +  \frac{\nu_T}{2} \uu^\top \nabla^2_{\theta_{\textcolor{black}{T}}\theta^\top_{\textcolor{black}{T}}}\Gb_T(\underline{y};\theta_{0,\textcolor{black}{T}})}\\
 && + \nu^{-1}_T\overset{d_T}{\underset{i=1}{\sum}}\big\{\pp(\frac{\lambda_T}{T},|\theta_{0,i,\textcolor{black}{T}}+\nu_T\uu_i|)-\pp(\frac{\lambda_T}{T},|\theta_{0,i,\textcolor{black}{T}}|)\big\} \leq 0) \\
&\leq &\Pb(\exists \uu, \|\uu\|_2 = C_{\eps} : |\nabla_{\theta_{\textcolor{black}{T}}} \Gb_T(\underline{y};\theta_{0,\textcolor{black}{T}})\uu| > \delta_T/4) + \Pb(\exists \uu, \|\uu\|_2 = C_{\eps} : |\frac{\nu_T}{2} \Rc_T(\theta_{0,\textcolor{black}{T}}) |> \delta_T/4) \\
& & + \Pb(\exists \uu, \|\uu\|_2 = C_{\eps} :  |\overset{d_T}{\underset{k=1}{\sum}}\{\pp(\frac{\lambda_T}{T},|\theta_{0,k,\textcolor{black}{T}}+\nu_T\uu_k|)-\pp(\frac{\lambda_T}{T},|\theta_{0,k,\textcolor{black}{T}}|)\}|> \nu_T \delta_T/4) \\
& \leq & \frac{16 K_1 C^2_{\eps} d_T}{\delta^2_T T} + \frac{4 \nu^2_T C^4_{\eps} d^2_T}{T \delta^2_T} + \eps/3 \\
& \leq & \frac{C_1 d_T}{T C^2_{\eps}\nu^2_T} + \frac{C_2 d^2_T}{T} + \eps/3,
\end{eqnarray*}
for $C_1,C_2>0$ some finite constants, where we used for $T$ and $C_{\eps}$ sufficiently large enough
\beqw
\Pb(\exists \uu, \|\uu\|_2 = C_{\eps} :  |\overset{d_T}{\underset{k=1}{\sum}}\{\pp(\frac{\lambda_T}{T},|\theta_{0,k,\textcolor{black}{T}}+\nu_T\uu_k|)-\pp(\frac{\lambda_T}{T},|\theta_{0,k,\textcolor{black}{T}}|)\}|> \nu_T \delta_T/4) < \eps/3.
\eeqw
Moreover, we chose $\nu_T = \sqrt{d}_T \big(T^{-1/2}+R_T\big)$ with $R_T = A_{1,T}$ for the SCAD and MCP, $R_T = \frac{\lambda_T}{T}$ for the LASSO, we obtain
\begin{eqnarray*}
\lefteqn{\Pb(\exists \uu, \|\uu\|_2 = C_{\eps} :\uu^\top \nabla_{\theta_{\textcolor{black}{T}}} \Gb_T(\underline{y};\theta_{0,\textcolor{black}{T}}) +  \frac{\nu_T}{2} \uu^\top \nabla^2_{\theta_{\textcolor{black}{T}}\theta^\top_{\textcolor{black}{T}}}\Gb_T(\underline{y};\theta_{0,\textcolor{black}{T}})}\\
 && + \nu^{-1}_T\overset{d_T}{\underset{k=1}{\sum}}\big\{\pp(\frac{\lambda_T}{T},|\theta_{0,k,\textcolor{black}{T}}+\nu_T\uu_k|)-\pp(\frac{\lambda_T}{T},|\theta_{0,k,\textcolor{black}{T}}|)\big\} \leq 0) \\
& \leq & \frac{C_1}{C^2_{\eps}} + \frac{C_2 d^2_T}{T} + \eps/3,
\end{eqnarray*}
Now for $C_{\eps}$ sufficiently large, $\frac{C_1}{C^2_{\eps}}<\eps/3$. Now there exists $T_0$ such that for $T>T_0$ and a fixed $C_{\eps}$, $\frac{C_2 d^2_T}{T}<\eps/3$ under $d^2_T=o(T)$. Hence, 
\begin{eqnarray*}
\lefteqn{\Pb(\exists \uu, \|\uu\|_2 = C_{\eps} :\uu^\top \nabla_{\theta_{\textcolor{black}{T}}} \Gb_T(\underline{y};\theta_{0,\textcolor{black}{T}}) +  \frac{\nu_T}{2} \uu^\top \nabla^2_{\theta_{\textcolor{black}{T}}\theta^\top_{\textcolor{black}{T}}}\Gb_T(\underline{y};\theta_{0,\textcolor{black}{T}})}\\
&&+ \nu^{-1}_T\overset{d_T}{\underset{k=1}{\sum}}\big\{\pp(\frac{\lambda_T}{T},|\theta_{0,k,\textcolor{black}{T}}+\nu_T\uu_k|)-\pp(\frac{\lambda_T}{T},|\theta_{0,k,\textcolor{black}{T}}|)\big\} \leq 0) < \eps.
\end{eqnarray*}
We deduce $\|\widehat{\theta}_{\textcolor{black}{T}}-\theta_{0,\textcolor{black}{T}}\|_2 = O_p(\nu_T)$.
\end{proof}

\begin{proof}[Proof of Theorem \ref{oracle_theorem}.]
Let us define $\theta = (\theta^\top_{\Ac_{\textcolor{black}{T}}},\theta^\top_{\Ac^c_{\textcolor{black}{T}}})^\top$. To prove the support recovery consistency, we show with probability tending to one when $T \rightarrow \infty$, under $\|\theta_{\Ac_{\textcolor{black}{T}}}-\theta_{0,\Ac_{\textcolor{black}{T}}}\| = O_p(T^{-1/2})$ and suitable regularisation rates depending on the penalty, that
\beq \label{correct_min}
\Gb^{\text{pen}}_T(\underline{y};\theta_{\Ac_{\textcolor{black}{T}}},\mathbf{0}_{\Ac^c_{\textcolor{black}{T}}}) = \underset{\|\theta_{\Ac^c_{\textcolor{black}{T}}}\|_2 \leq C \sqrt{d_T/T}}{\min} \;\big\{\Gb^{\text{pen}}_T(\underline{y};\theta_{\Ac_{\textcolor{black}{T}}},\theta_{\Ac^c_{\textcolor{black}{T}}}) \big\}.
\eeq
To prove (\ref{correct_min}), for any $\sqrt{T/d_T}$-consistent $\theta_{\Ac_{\textcolor{black}{T}}}$, we show that over the set $\{i \in \Ac^c_{\textcolor{black}{T}}, \theta_{i,\textcolor{black}{T}}:|\theta_{i,\textcolor{black}{T}}| \leq \sqrt{d_T/T}C\}$ for $C>0$,
\beq \label{gradient_sign}
\begin{array}{llll}
\partial_{\theta_{i,\textcolor{black}{T}}} \Gb^{\text{pen}}_T(\underline{y};\theta_{\textcolor{black}{T}}) > 0 & \text{when} & 0 < \theta_{i,\textcolor{black}{T}} < \sqrt{d_T/T}C, \\
\partial_{\theta_{i,\textcolor{black}{T}}} \Gb^{\text{pen}}_T(\underline{y};\theta_{\textcolor{black}{T}}) < 0 & \text{when} & -\sqrt{d_T/T}C < \theta_{i,\textcolor{black}{T}} < 0,
\end{array}
\eeq
with probability converging to $1$. For any index $i \in \Ac^c_{\textcolor{black}{T}}$, by a Taylor expansion around the true parameter, we have
\begin{eqnarray*}
\lefteqn{ \partial_{\theta_{i,\textcolor{black}{T}}} \Gb^{\text{pen}}_T(\underline{y};\theta_{\textcolor{black}{T}}) = \partial_{\theta_{i,\textcolor{black}{T}}} \Gb_T(\underline{y};\theta_{\textcolor{black}{T}}) + \partial_{\theta_{i,\textcolor{black}{T}}} \pp(\frac{\lambda_T}{T},|\theta_{i,\textcolor{black}{T}}|) \text{sgn}(\theta_{i,\textcolor{black}{T}})} \\
& = & \partial_{\theta_{i,\textcolor{black}{T}}} \Gb_T(\underline{y};\theta_{0,\textcolor{black}{T}}) + \overset{d_T}{\underset{k=1}{\sum}} \partial^2_{\theta_{i,\textcolor{black}{T}} \theta_{k,\textcolor{black}{T}}} \Gb_T(\underline{y};\theta_{0,\textcolor{black}{T}}) \big(\theta_{k,\textcolor{black}{T}}-\theta_{0,k,\textcolor{black}{T}}\big) + \partial_{\theta_{i,\textcolor{black}{T}}} \pp(\frac{\lambda_T}{T},|\theta_{i,\textcolor{black}{T}}|) \text{sgn}(\theta_{i,\textcolor{black}{T}}).
\end{eqnarray*}
Now using Assumption \ref{assumption_gradient}, we have the bound $\partial_{\theta_{i,\textcolor{black}{T}}} \Gb_T(\underline{y};\theta_{0,\textcolor{black}{T}}) = O_p(\sqrt{\frac{d_T}{T}})$. The second order term can be developed as
\beqw
\overset{d_T}{\underset{k=1}{\sum}} \partial^2_{\theta_{i,\textcolor{black}{T}} \theta_{k,\textcolor{black}{T}}} \Gb_T(\underline{y};\theta_{0,\textcolor{black}{T}}) \big(\theta_{k,\textcolor{black}{T}}-\theta_{0,k,\textcolor{black}{T}}\big) = \overset{d_T}{\underset{k=1}{\sum}} \Eb[\partial^2_{\theta_{i,\textcolor{black}{T}}\theta_{k,\textcolor{black}{T}}} \ell(y_s,s\leq t;\theta_{0,\textcolor{black}{T}})]\big(\theta_{k,\textcolor{black}{T}}-\theta_{0,k,\textcolor{black}{T}}\big) + \Kc_T(\theta_{\textcolor{black}{T}},\theta_{0,\textcolor{black}{T}}),
\eeqw
with 
\beqw
\Kc_T(\theta_{\textcolor{black}{T}},\theta_{0,\textcolor{black}{T}}) = \overset{d_T}{\underset{k=1}{\sum}}\big(\partial^2_{\theta_{i,\textcolor{black}{T}}\theta_{k,\textcolor{black}{T}}}\Gb_T(\underline{y};\theta_{0,\textcolor{black}{T}})-\Eb[\partial^2_{\theta_{i,\textcolor{black}{T}}\theta_{k,\textcolor{black}{T}}} \ell(y_s,s\leq t;\theta_{0,\textcolor{black}{T}})]\big)\big(\theta_{k,\textcolor{black}{T}}-\theta_{0,k,\textcolor{black}{T}}\big).
\eeqw
Now using $\|\theta-\theta_{0,\textcolor{black}{T}}\|_2=O_p(\sqrt{d_T/T})$, denoting $\Pc_T(\theta_{0,\textcolor{black}{T}}) = \text{vec}\big(\partial^2_{\theta_{i,\textcolor{black}{T}}\theta_{k,\textcolor{black}{T}}}\Gb_T(\underline{y};\theta_{0,\textcolor{black}{T}})-\Eb[\partial^2_{\theta_{i,\textcolor{black}{T}}\theta_{k,\textcolor{black}{T}}} \ell(y_s,s\leq t;\theta_{0,\textcolor{black}{T}})],k=1,\cdots,\textcolor{black}{d_T}\big)$ the \textcolor{black}{$d_T$}-dimensional vector, we have for any $a>0$ and the Cauchy-Schwarz inequality
\begin{eqnarray*}
\lefteqn{\Pb(|\Kc_T(\theta_{\textcolor{black}{T}},\theta_{0,\textcolor{black}{T}})|>a)\leq \Pb(\|\theta-\theta_{0,\textcolor{black}{T}}\|_2\|\Pc_T(\theta_{0,\textcolor{black}{T}})\|_2>a)}\\
& \leq & \Pb(C_0\sqrt{\frac{d_T}{T}}\|\Pc_T(\theta_{0,\textcolor{black}{T}})\|_2|>a) + \Pb(\|\theta-\theta_{0,\textcolor{black}{T}}\|_2>C_0\sqrt{\frac{d_T}{T}}).
\end{eqnarray*}
where $C_0>0$ is a fixed constant. We have
\beqw
\Eb[\|\Pc_T(\theta_{0,\textcolor{black}{T}})\|^2_2] = \frac{1}{T^2}\overset{T}{\underset{t,t'=1}{\sum}} \overset{\textcolor{black}{d_T}}{\underset{k,\textcolor{black}{k'}=1}{\sum}} \Eb[\zeta_{i,k,t} \zeta_{i,k',t'}],
\eeqw
where $\zeta_{i,k,t} = \partial^2_{\theta_{i,\textcolor{black}{T}}\theta_{k,\textcolor{black}{T}}}\ell(y_s,s\leq t;\theta_{0,\textcolor{black}{T}})-\Eb[\partial^2_{\theta_{i,\textcolor{black}{T}}\theta_{k,\textcolor{black}{T}}}\ell(y_s,s\leq t;\theta_{0,\textcolor{black}{T}})]$. As a consequence, by Assumption \ref{assumption_second_derivative}, we deduce
\beqw
\Pb(C_0\sqrt{\frac{d_T}{T}}\|\Pc_T(\theta_{0,\textcolor{black}{T}})\|_2|>a) \leq \frac{L C^2_0 d^{\textcolor{black}{3}}_T}{T^2 a^2},
\eeqw
for a finite constant $L>0$. Hence, taking $a = \sqrt{d_T/T}$, we obtain for $L_0>0$ finite
\beqw
\Pb(|\Kc_T(\theta_{\textcolor{black}{T}},\theta_{0,\textcolor{black}{T}})|>L_0\sqrt{\frac{d_T}{T}}) \leq \frac{C_{st}d^2_T}{T} + \frac{\eps}{2},
\eeqw
where $C_{st}$ is a generic constant. Under the scaling assumption on $(d_T,T)$, we deduce $|\Kc_T(\theta_{\textcolor{black}{T}},\theta_{0,\textcolor{black}{T}})|=O_p(\sqrt{d_T/T})$. Moreover, using $\|\theta-\theta_{0,\textcolor{black}{T}}\|_2=O_p(\sqrt{d_T/T})$, 
\beqw
|\overset{d_T}{\underset{k=1}{\sum}} \Eb[\partial^2_{\theta_{i,\textcolor{black}{T}}\theta_{k,\textcolor{black}{T}}} \ell(y_s,s\leq t;\theta_{0,\textcolor{black}{T}})]\big(\theta_{k,\textcolor{black}{T}}-\theta_{0,k,\textcolor{black}{T}}\big)| \leq \big(\overset{d_T}{\underset{k=1}{\sum}} \Eb[\partial^2_{\theta_{i,\textcolor{black}{T}}\theta_{k,\textcolor{black}{T}}} \ell(y_s,s\leq t;\theta_{0,\textcolor{black}{T}})]^2\big)^{1/2} C_0\sqrt{\frac{d_T}{T}} \leq C_{st} \sqrt{\frac{d_T}{T}},
\eeqw
using the bound assumption on the eigenvalues of the Hessian matrix of Assumption \ref{assumption_var_cov}. Hence, 
\beqw
|\overset{d_T}{\underset{k=1}{\sum}} \Eb[\partial^2_{\theta_{i,\textcolor{black}{T}}\theta_{k,\textcolor{black}{T}}} \ell(y_s,s\leq t;\theta_{0,\textcolor{black}{T}})]\big(\theta_{k,\textcolor{black}{T}}-\theta_{0,k,\textcolor{black}{T}}\big)| = O_p(\sqrt{\frac{d_T}{T}}).
\eeqw
Thus, putting the pieces together, we obtain
\beqw
\overset{d_T}{\underset{k=1}{\sum}} \partial^2_{\theta_{i,\textcolor{black}{T}} \theta_{k,\textcolor{black}{T}}} \Gb_T(\underline{y};\theta_{0,\textcolor{black}{T}}) \big(\theta_{k,\textcolor{black}{T}}-\theta_{0,k,\textcolor{black}{T}}\big) = O_p(\sqrt{\frac{d_T}{T}}).
\eeqw
We thus obtain for the SCAD and MCP penalty functions
\beqw
\begin{array}{llll}
\partial_{\theta_{i,\textcolor{black}{T}}} \Gb_T(\underline{y};\theta_{\textcolor{black}{T}}) & = & O_p(\sqrt{\frac{d_T}{T}}) + \partial_{\theta_{i,\textcolor{black}{T}}} \pp(\frac{\lambda_T}{T},|\theta_{i,\textcolor{black}{T}}|) \text{sgn}(\theta_{i,\textcolor{black}{T}}) \\
& = & \frac{\lambda_T}{T}\big[\frac{T}{\lambda_T} \partial_{\theta_{i,\textcolor{black}{T}}} \pp(\frac{\lambda_T}{T},|\theta_{i,\textcolor{black}{T}}|) \text{sgn}(\theta_{i,\textcolor{black}{T}}) + O_p(\frac{\sqrt{T d_T}}{\lambda_T}) \big].
\end{array}
\eeqw
As a consequence, under Assumption \ref{assumption_folded}-(iv), that is $\underset{T \rightarrow \infty}{\lim} \; \underset{x \rightarrow 0^+}{\lim \, \inf} \; \frac{T}{\lambda_T}\nabla_{x}\pp(\frac{\lambda_T}{T},x) > 0$ and if the regularisation parameter satisfies $\frac{\lambda_T}{\sqrt{d_T T}} \rightarrow \infty$, we deduce that the sign of the gradient entirely depends on the sign of $\widehat{\theta}_{\textcolor{black}{T}}$. This this proves (\ref{gradient_sign}).

\medskip

We now turn to the asymptotic distribution. We proved that $\widehat{\theta}_{\Ac^c_{\textcolor{black}{T}}}$ becomes $\mathbf{0}_{\Ac^c_{\textcolor{black}{T}}}$ with probability approaching one. Now by a Taylor expansion around $\theta_{0,i,\textcolor{black}{T}}$, for each $i \in \Ac_{\textcolor{black}{T}}$, we have  
\begin{eqnarray*}
\lefteqn{\partial_{\theta_{i,\textcolor{black}{T}}} \Gb_T(\underline{y};\widehat{\theta}_{\textcolor{black}{T}})+\partial_{\theta_{i,\textcolor{black}{T}}}\pp(\frac{\lambda_T}{T},|\widehat{\theta}_{i,\textcolor{black}{T}}|)\text{sgn}(\widehat{\theta}_{i,\textcolor{black}{T}}) }\\
&=& \partial_{\theta_{i,\textcolor{black}{T}}} \Gb_T(\underline{y};\theta_{0,\textcolor{black}{T}}) + \underset{j\in \Ac_{\textcolor{black}{T}}}{\sum}\partial^2_{\theta_{i,\textcolor{black}{T}} \theta_{j,\textcolor{black}{T}}} \Gb_T(\underline{y};\theta_{0,\textcolor{black}{T}}) (\widehat{\theta}_{j,\textcolor{black}{T}}-\theta_{0,j,\textcolor{black}{T}}) + \partial_{\theta_{i,\textcolor{black}{T}}}\pp(\frac{\lambda_T}{T},|\theta_{0,i,\textcolor{black}{T}}|)\text{sgn}(\theta_{0,i,\textcolor{black}{T}})\\
&+ &\partial^2_{\theta_{i,\textcolor{black}{T}} \theta_{i,\textcolor{black}{T}}}\pp(\frac{\lambda_T}{T},|\textcolor{black}{\widetilde{\theta}}_{i,T}|)(\widehat{\theta}_{i,\textcolor{black}{T}}-\theta_{0,i,\textcolor{black}{T}}),
\end{eqnarray*}
where $\textcolor{black}{\widetilde{\theta}}_{\textcolor{black}{T}}$ is such that $\|\textcolor{black}{\widetilde{\theta}}_{\textcolor{black}{T}}-\theta_{0,\textcolor{black}{T}}\|_2\leq \|\widehat{\theta}_{\textcolor{black}{T}}-\theta_{0,\textcolor{black}{T}}\|_2$.
Then inverting this relationship and multiplying by $\sqrt{T}$, we obtain in vector form with respect to the elements in $\Ac_{\textcolor{black}{T}}$
\begin{eqnarray*}
\lefteqn{\Hb_{\Ac_{\textcolor{black}{T}}\Ac_{\textcolor{black}{T}}}\big(\widehat{\theta}_{\textcolor{black}{T}}-\theta_{0,\textcolor{black}{T}}\big)_{\Ac_{\textcolor{black}{T}}} + \mathbf{b}(\theta_{0,\textcolor{black}{T}})_{\Ac_{\textcolor{black}{T}}} + \mathbf{S}(\theta_{0,\textcolor{black}{T}})_{\Ac_{\textcolor{black}{T}}\Ac_{\textcolor{black}{T}}}\big(\widehat{\theta}_{\textcolor{black}{T}}-\theta_{0,\textcolor{black}{T}}\big)_{\Ac_{\textcolor{black}{T}}}}\\
&=& -\nabla_{\theta_{\Ac_{\textcolor{black}{T}}}}\Gb_T(\underline{y};0_{\Ac^c_{\textcolor{black}{T}}},\theta_{0,\Ac_{\textcolor{black}{T}}}) - \Pc(\theta_{0,\textcolor{black}{T}})_{\Ac_{\textcolor{black}{T}}\Ac_{\textcolor{black}{T}}}\big(\widehat{\theta}_{\textcolor{black}{T}}-\theta_{0,\textcolor{black}{T}}\big)_{\Ac_{\textcolor{black}{T}}} - \big\{\mathbf{S}(\widetilde{\theta}_{\textcolor{black}{T}})_{\Ac_{\textcolor{black}{T}}\Ac_{\textcolor{black}{T}}}-\mathbf{S}(\theta_{0,\textcolor{black}{T}})_{\Ac_{\textcolor{black}{T}}\Ac_{\textcolor{black}{T}}}\big\}\big(\widehat{\theta}_{\textcolor{black}{T}}-\theta_{0,\textcolor{black}{T}}\big)_{\Ac_{\textcolor{black}{T}}},
\end{eqnarray*}
where $\Pc(\theta_{0,\textcolor{black}{T}})_{\Ac_{\textcolor{black}{T}}\Ac_{\textcolor{black}{T}}} = \nabla^2_{\theta_{\Ac_{\textcolor{black}{T}}}\theta^\top_{\Ac_{\textcolor{black}{T}}}}\Gb_T(\underline{y};0_{\Ac^c_{\textcolor{black}{T}}},\theta_{0,\Ac_{\textcolor{black}{T}}}) - \Hb_{\Ac_{\textcolor{black}{T}}\Ac_{\textcolor{black}{T}}}$ and
\beqw
\begin{array}{llll}
\mathbf{b}(\theta_{0,\textcolor{black}{T}})_{\Ac_{\textcolor{black}{T}}} &=& \big(\partial_{\theta_{j,\textcolor{black}{T}}} \pp(\frac{\lambda_T}{T},|\theta_{0,j}|)\text{sgn}(\theta_{0,j}),j=1,\cdots,\textcolor{black}{k_T}\big)^\top \in \Rb^{\textcolor{black}{k_T}},\\ \mathbf{S}(\theta_{0,\textcolor{black}{T}})_{\Ac_{\textcolor{black}{T}}\Ac_{\textcolor{black}{T}}} &=& \text{diag}(\partial^2_{\theta_{i,\textcolor{black}{T}} \theta_{i,\textcolor{black}{T}}} \pp(\frac{\lambda_T}{T},|\theta_{0,i,\textcolor{black}{T}}|),i =1,\cdots,\textcolor{black}{k_T}) \in \Mc_{\textcolor{black}{k_T} \times \textcolor{black}{k_T}}(\Rb).
\end{array}
\eeqw
Let $\mathbf{K}_{\Ac_{\textcolor{black}{T}} \Ac_{\textcolor{black}{T}}}= I_{\Ac_{\textcolor{black}{T}}}+\Hb^{-1}_{\Ac_{\textcolor{black}{T}}\Ac_{\textcolor{black}{T}}} \mathbf{S}(\theta_{0,\textcolor{black}{T}})_{\Ac_{\textcolor{black}{T}}\Ac_{\textcolor{black}{T}}}$, multiplying both sides by $\Hb^{-1}_{\Ac_{\textcolor{black}{T}}\Ac_{\textcolor{black}{T}}}$, and multiplying by $\sqrt{T}Q_T \Vb^{-1/2}_{\Ac_{\textcolor{black}{T}}\Ac_{\textcolor{black}{T}}}$, we obtain
\begin{eqnarray*}
\lefteqn{\sqrt{T} Q_T \Vb^{-1/2}_{\Ac_{\textcolor{black}{T}}\Ac_{\textcolor{black}{T}}} \mathbf{K}_{\Ac_{\textcolor{black}{T}} \Ac_{\textcolor{black}{T}}} \Big[\big(\widehat{\theta}_{\textcolor{black}{T}}-\theta_{0,\textcolor{black}{T}}\big)_{\Ac_{\textcolor{black}{T}}}  + \big(\Hb_{\Ac_{\textcolor{black}{T}}\Ac_{\textcolor{black}{T}}} \mathbf{K}_{\Ac_{\textcolor{black}{T}} \Ac_{\textcolor{black}{T}}}\big)^{-1}  \mathbf{b}(\theta_{0,\textcolor{black}{T}})_{\Ac_{\textcolor{black}{T}}}\Big] }\\
& = &  -\sqrt{T} Q_T \Vb^{-1/2}_{\Ac_{\textcolor{black}{T}}\Ac_{\textcolor{black}{T}}}\Hb^{-1}_{\Ac_{\textcolor{black}{T}}\Ac_{\textcolor{black}{T}}} \Big[\nabla_{\theta_{\Ac_{\textcolor{black}{T}}}}\Gb_T(\underline{y};0_{\Ac^c_{\textcolor{black}{T}}},\theta_{0,\Ac_{\textcolor{black}{T}}}) + \Pc(\theta_{0,\textcolor{black}{T}})_{\Ac_{\textcolor{black}{T}},\Ac_{\textcolor{black}{T}}}\big(\widehat{\theta}_{\textcolor{black}{T}}-\theta_{0,\textcolor{black}{T}}\big)_{\Ac_{\textcolor{black}{T}}}\\
&&+ \big\{\mathbf{S}(\widetilde{\theta}_{\textcolor{black}{T}})_{\Ac_{\textcolor{black}{T}}\Ac_{\textcolor{black}{T}}}-\mathbf{S}(\theta_{0,\textcolor{black}{T}})_{\Ac_{\textcolor{black}{T}}\Ac_{\textcolor{black}{T}}}\big\}\big(\widehat{\theta}_{\textcolor{black}{T}}-\theta_{0,\textcolor{black}{T}}\big)_{\Ac_{\textcolor{black}{T}}}\Big].
\end{eqnarray*}
First, let us treat the term involving the second order derivative. \textcolor{black}{We have:
\begin{eqnarray*}
\lefteqn{\Pc(\theta_{0,\textcolor{black}{T}})_{\Ac_{\textcolor{black}{T}},\Ac_{\textcolor{black}{T}}}\big(\widehat{\theta}_{\textcolor{black}{T}}-\theta_{0,\textcolor{black}{T}}\big)_{\Ac_{\textcolor{black}{T}}}}\\
& \leq & \|\nabla^2_{\theta_{\Ac_{\textcolor{black}{T}}}\theta^\top_{\Ac_{\textcolor{black}{T}}}}\Gb_T(\underline{y};0_{\Ac^c_{\textcolor{black}{T}}},\theta_{0,\Ac_{\textcolor{black}{T}}}) - \Hb_{\Ac_{\textcolor{black}{T}}\Ac_{\textcolor{black}{T}}}\|_F\|\big(\widehat{\theta}_{\textcolor{black}{T}}-\theta_{0,\textcolor{black}{T}}\big)_{\Ac_{\textcolor{black}{T}}}\|_2 =  O_p(\frac{d_T}{\sqrt{T}}) O_p(\sqrt{\frac{d_T}{T}}) = \textcolor{black}{o_p(\frac{1}{\sqrt{T}})}, 
\end{eqnarray*}
using Assumption \ref{assumption_second_derivative}. Since $\|\Hb^{-1}_{\Ac_T\Ac_T} \xx\|_2 \leq \lambda_{\min}^{-1}(\Hb_{\Ac_T\Ac_T})\|\xx\|_2$ for any vector $\xx \in \Rb^{k_T}$, we deduce:
\begin{eqnarray*}
| Q_T \Vb^{-1/2}_{\Ac_{\textcolor{black}{T}}\Ac_{\textcolor{black}{T}}}\Hb^{-1}_{\Ac_{\textcolor{black}{T}}\Ac_{\textcolor{black}{T}}} \Pc(\theta_{0,\textcolor{black}{T}})_{\Ac_{\textcolor{black}{T}},\Ac_{\textcolor{black}{T}}}\big(\widehat{\theta}_{\textcolor{black}{T}}-\theta_{0,\textcolor{black}{T}}\big)_{\Ac_{\textcolor{black}{T}}}| \leq | Q_T \Vb^{-1/2}_{\Ac_{\textcolor{black}{T}}\Ac_{\textcolor{black}{T}}}| \lambda_{\min}^{-1}(\Hb_{\Ac_T\Ac_T})O_p(\frac{d_T}{\sqrt{T}}) O_p(\sqrt{\frac{d_T}{T}}) = \textcolor{black}{o_p(\frac{1}{\sqrt{T}})}.
\end{eqnarray*}}
% \begin{eqnarray*}
% \Pb(\|\nabla^2_{\theta_{\Ac_{\textcolor{black}{T}}}\theta^\top_{\Ac_{\textcolor{black}{T}}}}\Gb_T(\underline{y};0_{\Ac^c_{\textcolor{black}{T}}},\theta_{0,\Ac_{\textcolor{black}{T}}}) - \Hb_{\Ac_{\textcolor{black}{T}}\Ac_{\textcolor{black}{T}}}\|_2 > (\eps/d_T))\leq \frac{ d^2_T}{\eps^2} \frac{1}{T^2} \overset{T}{\underset{t,t'=1}{\sum}} \underset{k,k',l,l'\in \Ac_T}{\sum} \Eb[\zeta_{kl,t} \zeta_{k'l',t'}] \leq C_0 \frac{d^2_T}{\eps^2} \frac{d^2_T}{T^2}
% \end{eqnarray*}
% with $C_0>0$ a finite constant, $\zeta_{kl,t} = \partial^2_{\theta_{k,\textcolor{black}{T}} \theta_{l,\textcolor{black}{T}}} \ell(y_s,s\leq t;\theta_{0,\textcolor{black}{T}}) - \Eb[\partial^2_{\theta_{k,\textcolor{black}{T}} \theta_{l,\textcolor{black}{T}}} \ell(y_s,s\leq t;\theta_{0,\textcolor{black}{T}})]$, and by Assumption \ref{assumption_second_derivative}.} 
As for the expansion with respect to the penalty term, using Assumption \ref{assumption_folded}-(iii), element-by-element, we obtain
{\small{\beqw
\forall k \in \Ac_T, |\partial^2_{\theta_{k,\textcolor{black}{T}}\theta_{k,\textcolor{black}{T}}}\pp(\frac{\lambda_T}{T},|\widetilde{\theta}_{k,\textcolor{black}{T}}|)-\partial^2_{\theta_{k,\textcolor{black}{T}}\theta_{k,\textcolor{black}{T}}}\pp(\frac{\lambda_T}{T},|\theta_{0,k,\textcolor{black}{T}}|)||\widehat{\theta}_{k,\textcolor{black}{T}}-\theta_{0,k,\textcolor{black}{T}}| \leq K |\widetilde{\theta}_{k,\textcolor{black}{T}}-\theta_{0,k,\textcolor{black}{T}}| |\widehat{\theta}_{k,\textcolor{black}{T}}-\theta_{0,k,\textcolor{black}{T}}| = O_p(\frac{d_T}{T}).
\eeqw}}
\textcolor{black}{We deduce $|\sqrt{T} Q_T \Vb^{-1/2}_{\Ac_{\textcolor{black}{T}}\Ac_{\textcolor{black}{T}}}\Hb^{-1}_{\Ac_{\textcolor{black}{T}}\Ac_{\textcolor{black}{T}}}  \big\{\mathbf{S}(\widetilde{\theta}_{\textcolor{black}{T}})_{\Ac_{\textcolor{black}{T}}\Ac_{\textcolor{black}{T}}}-\mathbf{S}(\theta_{0,\textcolor{black}{T}})_{\Ac_{\textcolor{black}{T}}\Ac_{\textcolor{black}{T}}}\big\}\big(\widehat{\theta}_{\textcolor{black}{T}}-\theta_{0,\textcolor{black}{T}}\big)_{\Ac_{\textcolor{black}{T}}}|=o_p(1)$.}

\medskip

We now prove that $X_{T,t} = \sqrt{T}Q_T \Vb^{\textcolor{black}{-}1/2}_{\Ac_{\textcolor{black}{T}}\Ac_{\textcolor{black}{T}}}\Hb^{-1}_{\Ac_{\textcolor{black}{T}}\Ac_{\textcolor{black}{T}}}\nabla_{\theta_{\Ac_{\textcolor{black}{T}}}}\Gb_{T,t}(\underline{y};0_{\Ac^c_{\textcolor{black}{T}}},\theta_{0,\Ac_{\textcolor{black}{T}}}), t = 1,\cdots,T$ is asymptotically normal by checking the Lindeberg condition for applying Theorem \ref{lindeberg_shiryaev} of Shiryaev.  Here, $\nabla_{\theta_{\textcolor{black}{T}}}\Gb_{T,t}(\underline{y};0_{\Ac^c_{\textcolor{black}{T}}},\theta_{0,\Ac_{\textcolor{black}{T}}})$ is the $t$-th point of the score of the empirical criterion. Now let $\beta >0$, we need to prove that for any $\eps>0$
\beqw
\Pb(\overset{T}{\underset{t=1}{\sum}} \Eb[\|X_{T,t}\|^2_2\mathbf{1}_{\|X_{T,t}\|_2>\beta}|\Fc^T_{t-1}] > \eps) \underset{T \rightarrow \infty}{\longrightarrow} 0,
\eeqw
where $\Fc^T_{t} = \sigma(X_{T,s},s\leq t)$. By the Markov inequality, we have
\begin{eqnarray*}
\lefteqn{\Pb(\overset{T}{\underset{t=1}{\sum}} \Eb[\|X_{T,t}\|^2_2\mathbf{1}_{\|X_{T,t}\|_2>\beta}|\Fc^T_{t-1}] > \eps)  \leq \frac{1}{\eps} \overset{T}{\underset{t=1}{\sum}} \Eb[\Eb[\|X_{T,t}\|^2_2\mathbf{1}_{\|X_{T,t}\|_2>\beta}|\Fc^T_{t-1}]]}  \\
& \leq & \frac{1}{\eps} \overset{T}{\underset{t=1}{\sum}} \Eb[\Eb[\|X_{T,t}\|^4_2|\Fc^T_{t-1}] \Pb(\|X_{T,t}\|_2>\beta|\Fc^T_{t-1})^{1/2}] \\
& \leq &  \frac{1}{\eps} \overset{T}{\underset{t=1}{\sum}} \Eb[ \big\{\frac{C_{st}}{T^2}\Eb[\|\nabla_{\theta_{\Ac_{\textcolor{black}{T}}}}\ell(y_s,s\leq t; \theta_{0,\textcolor{black}{T}}) \nabla_{\theta^\top_{\Ac_{\textcolor{black}{T}}}}\ell(y_s,s\leq t; \theta_{0,\textcolor{black}{T}})\|^2_2|\Fc^T_{t-1}]\big\}^{1/2} \\
& & \times \frac{1}{\beta}\Eb[\|\sqrt{T}Q_T \Vb^{1/2}_{\Ac_{\textcolor{black}{T}}\Ac_{\textcolor{black}{T}}}\Hb^{-1}_{\Ac_{\textcolor{black}{T}}\Ac_{\textcolor{black}{T}}}\nabla_{\theta_{\Ac_{\textcolor{black}{T}}}}\Gb_{T,t}(\underline{y};0_{\Ac^c_{\textcolor{black}{T}}},\theta_{0,\Ac_{\textcolor{black}{T}}})\|^2_2|\Fc^T_{t-1}]^{1/2}],
\end{eqnarray*}
with $C_{st}>0$. Then, let $\Kb_T = \Vb^{-1/2}_{\Ac_{\textcolor{black}{T}}\Ac_{\textcolor{black}{T}}}\Hb^{-1}_{\Ac_{\textcolor{black}{T}}\Ac_{\textcolor{black}{T}}}$. We have
\begin{eqnarray*}
\lefteqn{\Eb[\|\sqrt{T}\Kb_T\textcolor{black}{\nabla_{\theta_{\Ac_T}}}\Gb_{T,t}(\underline{y};0_{\Ac^c_{\textcolor{black}{T}}},\theta_{0,\Ac_{\textcolor{black}{T}}})\|^2_2|\Fc^T_{t-1}] = \frac{1}{T}\Eb[\nabla_{\theta_{\Ac^\top_T}}\ell(y_s,s\leq t; \theta_{0,\textcolor{black}{T}})\Kb^\top_T \Kb_T \nabla_{\theta_{\Ac_{\textcolor{black}{T}}}}\ell(y_s,s\leq t; \theta_{0,\textcolor{black}{T}})|\Fc^T_{t-1}]} \\
& = & \frac{1}{T}\Eb[\text{tr}\big(\nabla_{\theta_{\Ac^\top_T}}\ell(y_s,s\leq t; \theta_{0,\textcolor{black}{T}})\Kb^\top_T \Kb_T \nabla_{\theta_{\Ac_{\textcolor{black}{T}}}}\ell(y_s,s\leq t; \theta_{0,\textcolor{black}{T}})\big)|\Fc^T_{t-1}] \\
& = & \frac{1}{T}\text{tr}\big(\Eb[\nabla_{\theta_{\Ac^\top_T}}\ell(y_s,s\leq t; \theta_{0,\textcolor{black}{T}}) \nabla_{\theta_{\Ac_{\textcolor{black}{T}}}}\ell(y_s,s\leq t; \theta_{0,\textcolor{black}{T}})|\Fc^T_{t-1}]\Kb^\top_T \Kb_T\big) \leq \frac{1}{T}\lambda_{\max}(\Hb^T_{t-1}) \tilde{C}_{st},
\end{eqnarray*}
where $\tilde{C}_{st}>0$ is a finite constant. Moreover, we have
\begin{eqnarray*}
\lefteqn{\Eb[\|\nabla_{\theta_{\Ac_{\textcolor{black}{T}}}}\ell(y_s,s\leq t; \theta_{0,\textcolor{black}{T}}) \nabla_{\theta^\top_{\Ac_{\textcolor{black}{T}}}}\ell(y_s,s\leq t; \theta_{0,\textcolor{black}{T}})\|^2_2|\Fc^T_{t-1}]}\\
& = & \Eb[\overset{d_T}{\underset{k,l=1}{\sum}} \big\{\partial_{\theta_{k,\textcolor{black}{T}}}\ell(y_s,s\leq t;\theta_{0,\textcolor{black}{T}})\partial_{\theta_{l,\textcolor{black}{T}}}\ell(y_s,s\leq t;\theta_{0,\textcolor{black}{T}})\big\}\textcolor{black}{^2}|\Fc^T_{t-1}] \\
& \leq & d^2_T \underset{k,l=1,\cdots,d_T}{\sup} \Eb[\big\{\partial_{\theta_{k,\textcolor{black}{T}}}\ell(y_s,s\leq t;\theta_{0,\textcolor{black}{T}})\partial_{\theta_{l,\textcolor{black}{T}}}\ell(y_s,s\leq t;\theta_{0,\textcolor{black}{T}})\big\}\textcolor{black}{^2}|\Fc^T_{t-1}].
\end{eqnarray*}
Now by Assumption \ref{assumption_array}, we have
\begin{eqnarray*}
\lefteqn{\Pb(\overset{T}{\underset{t=1}{\sum}} \Eb[\|X_{T,t}\|^2_2\mathbf{1}_{\|X_{T,t}\|_2>\beta}|\Fc^T_{t-1}] > \eps) }\\
& \leq & \frac{C^{1/2}_{st}\tilde{C}^{1/2}_{st}d_T}{T^{3/2}} \overset{T}{\underset{t=1}{\sum}} \Eb[\underset{k,l=1,\cdots,d_T}{\sup} \Eb[\big\{\partial_{\theta_{k,\textcolor{black}{T}}}\ell(y_s,s\leq t;\theta_{0,\textcolor{black}{T}})\partial_{\theta_{l,\textcolor{black}{T}}}\ell(y_s,s\leq t;\theta_{0,\textcolor{black}{T}})\big\}^2|\Fc^T_{t-1}] \lambda_{\max}(\Hb^T_{t-1})]\\
& \leq & \frac{C^{1/2}_{st}\tilde{C}^{1/2}_{st} \overline{B}Td_T}{T^{3/2}}\cdot
\end{eqnarray*}
Hence, $\overset{T}{\underset{t=1}{\sum}} \Eb[\|X_{T,t}\|^2_2\mathbf{1}_{\|X_{T,t}\|_2>\beta}|\Fc^T_{t-1}]=o_p(1)$. Thus, $X_{T,t}$ satisfies the Lindeberg condition, and by Theorem \ref{lindeberg_shiryaev}, $\sqrt{T}Q_T \Vb^{-1/2}_{\Ac_{\textcolor{black}{T}}\Ac_{\textcolor{black}{T}}}\Hb^{-1}_{\Ac_{\textcolor{black}{T}}\Ac_{\textcolor{black}{T}}}\nabla_{\theta_{\Ac_{\textcolor{black}{T}}}}\Gb_{T}(\underline{y};0_{\Ac^c_{\textcolor{black}{T}}},\theta_{0,\Ac_{\textcolor{black}{T}}})$ is asymptotically normally distributed. Finally, for $T$ large enough, $\mathbf{b}(\theta_{0,\textcolor{black}{T}})_{\Ac_{\textcolor{black}{T}}} = 0 \in \Rb^{k_T}$ and $\mathbf{S}(\theta_{0,\textcolor{black}{T}})_{\Ac_{\textcolor{black}{T}}\Ac_{\textcolor{black}{T}}} = 0 \in \Mc_{k_T \times k_T}(\Rb)$. 
\end{proof}

{To establish the large sample properties of the second step estimator $\widehat{\gamma}$, we mainly rely on moment assumptions, detailed as follows.
\textcolor{black}{\begin{assumption}\label{assumption_first_2ndstep}
Let $\partial_{\gamma_{k}}f(y_s,s\leq t;\theta_{0,T},\gamma_0) = -\big(K_{t-1}(\theta_{0,T})\otimes \{y^{\ell}_t-\Gamma_0 K_{t-1}(\theta_{0,T})\}\big)_k$, for any $k=1,\cdots,d_2$. Then there is some function $\kappa(.)$ such that for any $T$: 
\beqw
\underset{k=1,\cdots,d_2}{\sup}\Eb[\partial_{\gamma_{k}}f(y_s,s\leq t;\theta_{0,T},\gamma_0) \partial_{\gamma_{k}}f(y_s,s\leq t';\theta_{0,T},\gamma_0)]\leq \kappa(|t-t'|), \;\; \text{and} \;\; \underset{T>0}{\sup}\; \frac{1}{T}\overset{T}{\underset{t,t'=1}{\sum}} \kappa(|t-t'|)<\infty.
\eeqw
\end{assumption}}

%\begin{assumption}\label{assumption_second_cross_expectation}
%The $d_2 \times d_T$ matrix $\textcolor{black}{\Upsilon_{\gamma \theta_T} = \Eb[\partial^2_{\gamma_l \theta_{k,T}}f(y_s,s\leq t;\theta_{0,\textcolor{black}{T}},\gamma_0)]_{1\leq l \leq d_2,1 \leq k \leq d_T}}$ exists and for any $T$, there is a constant $l_1<\infty$ such that $\|\textcolor{black}{\Upsilon_{\gamma \theta_T}}\|_s<l_1$.
%\end{assumption}
\begin{assumption}\label{assumption_second_cross_2ndstep}
Let $\upsilon_{kl,t}= \partial^2_{\gamma_k\theta_{l,\textcolor{black}{T}}}f(y_s,s\leq t;\theta_{0,\textcolor{black}{T}},\gamma_0)$, for any $k=1,\cdots,d_2$ and $l=1,\cdots,d_T$, where $\partial^2_{\gamma_k\theta_{l,\textcolor{black}{T}}}f(y_s,s\leq t;\theta_{0,\textcolor{black}{T}},\gamma_0)$ is provided in Appendix \ref{appendix_derivative}. There is some function $\xi(.)$ such that for any $T$: $|\Eb[\upsilon_{kl,t}\upsilon_{k'l',t'}] | \leq \xi(|t-t'|), \; \text{and} \;\underset{T>0}{\sup} \; \frac{1}{T}\overset{T}{\underset{t,t'=1}{\sum}}\xi(|t-t'|)<\infty$.
\end{assumption}
\textcolor{black}{\begin{assumption}\label{assumption_second_derivative_2nd_step}
Let $\partial^2_{\gamma_{k}\gamma_{l}}f(y_s,s\leq t;\theta_{0,T},\gamma_0) = \big(I_p\otimes K_{t-1}(\theta_{0,T})K^\top_{t-1}(\theta_{0,T})\big)_{kl}$ for any $k,l=1,\cdots,d_T$. Let $\mu_{kl,t} = \partial^2_{\gamma_{k} \gamma_{l}} f(y_s,s\leq t;\theta_{0,T},\gamma_0) - \Eb[\partial^2_{\gamma_{k} \gamma_{l}} f(y_s,s\leq t;\theta_{0,T},\gamma_0)]$. There exists some function $\varphi(.)$ such that for any $T$: $|\Eb[\mu_{kl,t} \mu_{k'l',t'}]| \leq \varphi(|t-t'|), \;\; \text{and} \;\; \underset{T>0}{\sup} \; \frac{1}{T}\overset{T}{\underset{t,t'=1}{\sum}} \varphi(|t-t'|) <\infty$.
\end{assumption}}
For the next assumption, the third order partial derivatives $\partial^3_{\gamma_k\theta_{l,\textcolor{black}{T}}\theta_{j,\textcolor{black}{T}}}f(.), \partial^3_{\gamma_k\gamma_l\theta_{j,\textcolor{black}{T}}}f(.)$ are provided in Appendix \ref{appendix_derivative}.
\begin{assumption}\label{assumption_third_cross_1_2ndstep}
For almost all observations $(y_t)$, the derivatives $\partial^3_{\gamma_k \theta_{l,\textcolor{black}{T}}\theta_{j,\textcolor{black}{T}}}\ell(y_s,s\leq t;\theta_{\textcolor{black}{T}},\gamma)$ and $\partial^3_{\gamma_k \gamma_l\theta_{j,\textcolor{black}{T}}}\ell(y_s,s\leq t;\theta_{\textcolor{black}{T}},\gamma)$ exist. Let 
\beqw
\begin{array}{llll}
\kappa_t(L) & = & \underset{1\leq k \leq d_2, 1\leq l,j \leq d_T}{\sup} \big\{\underset{\theta_{\textcolor{black}{T}}:\|\theta_{\textcolor{black}{T}}-\theta_{0,\textcolor{black}{T}}\|_2\leq L \sqrt{d_T/T}}{\sup} \partial^3_{\gamma_k\theta_{l,\textcolor{black}{T}}\theta_{j,\textcolor{black}{T}}}f(y_s,s\leq t;\theta_{\textcolor{black}{T}},\gamma_0) \big\},\\
\rho_t(L) & = & \underset{1\leq k,l \leq d_2, 1\leq j \leq d_T}{\sup} \big\{\underset{\theta_{\textcolor{black}{T}}:\|\theta_{\textcolor{black}{T}}-\theta_{0,\textcolor{black}{T}}\|_2\leq L \sqrt{d_T/T}}{\sup} \partial^3_{\gamma_k\gamma_l\theta_{j,\textcolor{black}{T}}}f(y_s,s\leq t;\theta_{\textcolor{black}{T}},\gamma_0) \big\},
\end{array}
\eeqw
where $0 < L < \infty$. Then 
\beqw
\eta(L) = \frac{1}{T^2}\overset{T}{\underset{t,t'=1}{\sum}}\Eb[\kappa_t(L) \kappa_{t'}(L) ]<\infty, \;\; \zeta(L) = \frac{1}{T^2}\overset{T}{\underset{t,t'=1}{\sum}}\Eb[\rho_t(L) \rho_{t'}(L) ]<\infty.
\eeqw
\end{assumption}}
\begin{proof}[Proof of Theorem \ref{consistency_second}.]
Under the Theorem's assumptions, the first step estimator satisfies the rate $\|\widehat{\theta}_{\textcolor{black}{T}}-\theta_{0,\textcolor{black}{T}}\|_2=O_p(\sqrt{d_T/T})$. Now let us denote $\nu_T = T^{-1/2}$. We would like to prove that for any $\eps>0$, there exists $C_{\eps}>0$ such that 
\beq \label{obj_second}
\Pb(\nu^{-1}_T\|\widehat{\gamma}-\gamma_0\|>C_{\eps}) \leq \Pb(\exists \uu \in \Rb^{d_2}, \|\uu\|_2 \geq C_{\eps}: \Lb_T(\underline{y};\widehat{\theta}_{\textcolor{black}{T}},\gamma_0+\uu \nu_T)\leq \Lb_T(\underline{y};\widehat{\theta}_{\textcolor{black}{T}},\gamma_0) ),
\eeq
where $d_2 = p(1+2p)$ the second step parameter dimension size. Using the convexity of the objective function, we have
\begin{eqnarray}
\lefteqn{\big\{\exists \uu^* \in \Rb^{d_2}, \|\uu^*\| \geq C_{\eps}: \Lb_T(\underline{y};\widehat{\theta}_{\textcolor{black}{T}},\gamma_0+\uu^* \nu_T)\leq \Lb_T(\underline{y};\widehat{\theta}_{\textcolor{black}{T}},\gamma_0) \big\}}\nonumber\\
&\subset & \big\{\exists \overline{\uu} \in \Rb^{d_2}, \|\overline{\uu}\| \geq C_{\eps}: \Lb_T(\underline{y};\widehat{\theta}_{\textcolor{black}{T}},\gamma_0+\overline{\uu} \nu_T)\leq \Lb_T(\underline{y};\widehat{\theta}_{\textcolor{black}{T}},\gamma_0)\big\},\label{convex}
\end{eqnarray}
a relationship that allows us to work with a fixed $\|\uu\|_2$. Let $\gamma^*_1= \gamma_0+\nu_T \uu^*$ such that $\Lb_T(\underline{y};\widehat{\theta}_{\textcolor{black}{T}},\gamma_1)\leq \Lb_T(\underline{y};\widehat{\theta}_{\textcolor{black}{T}},\gamma_0)$. By convexity of $\Lb_T(\underline{y};\widehat{\theta}_{\textcolor{black}{T}},.)$, for a fixed $\widehat{\theta}_{\textcolor{black}{T}}$, we have for $\alpha \in (0,1)$ and $\gamma = \alpha \gamma_1 + (1-\alpha) \gamma_0$ that
\beqw
\Lb_T(\underline{y};\widehat{\theta}_{\textcolor{black}{T}},\gamma) \leq \alpha \Lb_T(\underline{y};\widehat{\theta}_{\textcolor{black}{T}},\gamma_1) + (1-\alpha) \Lb_T(\underline{y};\widehat{\theta}_{\textcolor{black}{T}},\gamma_0) \leq \Lb_T(\underline{y};\widehat{\theta}_{\textcolor{black}{T}},\gamma_0).
\eeqw
Now we choose $\alpha$ such that $\|\overline{\uu}\|_2=C_{\eps}$ with $\overline{\uu}=\alpha \gamma_1+(1-\alpha)\gamma_0$. Hence (\ref{convex}) holds and
\begin{eqnarray*}
\lefteqn{\Pb(\|\widehat\gamma-\gamma_0\|>C_{\eps}\nu_T)}\\
& \leq & \Pb(\exists \uu \in \Rb^{d_2}, \|\uu\|_2 \geq C_{\eps}: \Lb_T(\underline{y};\widehat{\theta}_{\textcolor{black}{T}},\gamma_0+\uu \nu_T)\leq \Lb_T(\underline{y};\widehat{\theta}_{\textcolor{black}{T}},\gamma_0) ) \\
& \leq & \Pb(\exists \overline{\uu} \in \Rb^{d_2}, \|\overline{\uu}\|_2 = C_{\eps}: \Lb_T(\underline{y};\widehat{\theta}_{\textcolor{black}{T}},\gamma_0+\overline{\uu} \nu_T)\leq \Lb_T(\underline{y};\widehat{\theta}_{\textcolor{black}{T}},\gamma_0) ).
\end{eqnarray*}
Thus we choose $\uu$ such that $\|\uu\|_2=C_{\eps}$. Now by a Taylor expansion, we want to prove
\beq\label{objective_loss}
\Pb(\exists \uu: \|\uu\|_2=C_{\eps}:\nabla_{\gamma}\Lb_T(\underline{y};\widehat{\theta}_{\textcolor{black}{T}},\gamma_0) \uu + \frac{\nu_T}{2}\uu^\top\nabla^2_{\gamma\gamma^\top}\Lb_T(\underline{y};\widehat{\theta}_{\textcolor{black}{T}},\gamma_0)\uu \leq 0) <\eps,
\eeq
where the third order derivative vanishes. Let us consider the first order term. By a Taylor expansion, we have for any $\uu$ such that $\|\uu\|_2=C_{\eps}$,
\begin{eqnarray}
\lefteqn{\uu^\top\nabla_{\gamma}\Lb_T(\underline{y};\widehat{\theta}_{\textcolor{black}{T}},\gamma_0)}\nonumber\\
& = & \overset{d_2}{\underset{k=1}{\sum}} \uu_k \nabla_{\gamma_k}\Lb_T(\underline{y};\theta_{0,\textcolor{black}{T}},\gamma_0)
+ \textcolor{black}{\overset{d_2}{\underset{k=1}{\sum}} \overset{d_T}{\underset{l=1}{\sum}} \uu_k \nabla^2_{\gamma_k\theta_{l,\textcolor{black}{T}}}\Lb_T(\underline{y};\theta_{0,\textcolor{black}{T}},\gamma_0)(\widehat{\theta}_{l,\textcolor{black}{T}}-\theta_{0,l,\textcolor{black}{T}})} \nonumber\\
%& & +  \overset{d_2}{\underset{k=1}{\sum}} \overset{d_T}{\underset{l=1}{\sum}} \uu_k\Eb[\nabla^2_{\gamma_k\theta_{l,\textcolor{black}{T}}}\Lb_T(\underline{y};\theta_{0,\textcolor{black}{T}},\gamma_0)](\widehat{\theta}_{l,\textcolor{black}{T}}-\theta_{0,l,\textcolor{black}{T}}) \nonumber\\
 &  & + \overset{d_2}{\underset{k=1}{\sum}} \overset{d_T}{\underset{l,j=1}{\sum}} \uu_k \nabla^3_{\gamma_k\theta_{l,\textcolor{black}{T}}\theta_{j,\textcolor{black}{T}}}\Lb_T(\underline{y};\overline{\theta}_{\textcolor{black}{T}},\gamma_0)(\widehat{\theta}_{l,\textcolor{black}{T}}-\theta_{0,l,\textcolor{black}{T}}) (\widehat{\theta}_{j,\textcolor{black}{T}}-\theta_{0,j,\textcolor{black}{T}}), \label{taylor_first_order}
\end{eqnarray}
where $\overline{\theta}_{\textcolor{black}{T}}$ is such that $\|\overline{\theta}_{\textcolor{black}{T}}-\theta_{0,\textcolor{black}{T}}\|_2 \leq \|\widehat{\theta}_{\textcolor{black}{T}}-\theta_{0,\textcolor{black}{T}}\|_2$. Thus, for any $a>0$, we have
{\footnotesize{\begin{eqnarray*}
\lefteqn{\Pb(\underset{\uu:\|\uu\|_2=C_{\eps}}{\sup}|\uu^\top\nabla_{\gamma}\Lb_T(\underline{y};\widehat{\theta}_{\textcolor{black}{T}},\gamma_0)|>a)}\\
&\leq & \Pb(\underset{\uu:\|\uu\|_2=C_{\eps}}{\sup}|\uu^\top\nabla_{\gamma}\Lb_T(\underline{y};\theta_{0,\textcolor{black}{T}},\gamma_0)|>a/3) 
+ \Pb(\underset{\uu:\|\uu\|_2=C_{\eps}}{\sup}|\overset{d_2}{\underset{k=1}{\sum}} \overset{d_T}{\underset{l=1}{\sum}} \uu_k \nabla^2_{\gamma_k\theta_{l,\textcolor{black}{T}}}\Lb_T(\underline{y};\theta_{0,\textcolor{black}{T}},\gamma_0)(\widehat{\theta}_{l,\textcolor{black}{T}}-\theta_{0,l,\textcolor{black}{T}})|>a/3) \\
& & + \Pb(\underset{\uu:\|\uu\|_2=C_{\eps}}{\sup}|\overset{d_2}{\underset{k=1}{\sum}} \overset{d_T}{\underset{l,j=1}{\sum}} \uu_k \nabla^3_{\gamma_k\theta_{l,\textcolor{black}{T}}\theta_{j,\textcolor{black}{T}}}\Lb_T(\underline{y};\overline{\theta}_{\textcolor{black}{T}},\gamma_0)(\widehat{\theta}_{l,\textcolor{black}{T}}-\theta_{0,l,\textcolor{black}{T}}) (\widehat{\theta}_{j,\textcolor{black}{T}}-\theta_{0,j,\textcolor{black}{T}})|>a/3).
\end{eqnarray*}}}
%\begin{eqnarray*}
%\lefteqn{\Pb(\underset{\uu:\|\uu\|_2=C_{\eps}}{\sup}|\uu^\top\Lb_T(\underline{y};\widehat{\theta}_{\textcolor{black}{T}},\gamma_0)|>a) \leq  \Pb(\underset{\uu:\|\uu\|_2=C_{\eps}}{\sup}|\uu^\top\Lb_T(\underline{y};\theta_{0,\textcolor{black}{T}},\gamma_0)|>a/4) }\\
%& & + \Pb(\underset{\uu:\|\uu\|_2=C_{\eps}}{\sup}|\overset{d_2}{\underset{k=1}{\sum}} \overset{d_T}{\underset{l=1}{\sum}} \uu_k \big\{\nabla^2_{\gamma_k\theta_{l,\textcolor{black}{T}}}\Lb_T(\underline{y};\theta_{0,\textcolor{black}{T}},\gamma_0)-\Eb[\nabla^2_{\gamma_k\theta_{l,\textcolor{black}{T}}}\Lb_T(\underline{y};\theta_{0,\textcolor{black}{T}},\gamma_0)]\big\}(\widehat{\theta}_{l,\textcolor{black}{T}}-\theta_{0,l,\textcolor{black}{T}})|>a/4) \\
%&& + \Pb(\underset{\uu:\|\uu\|_2=C_{\eps}}{\sup}|\overset{d_2}{\underset{k=1}{\sum}} \overset{d_T}{\underset{l=1}{\sum}} \uu_k \Eb[\nabla^2_{\gamma_k\theta_{l,\textcolor{black}{T}}}\Lb_T(\underline{y};\theta_{0,\textcolor{black}{T}},\gamma_0)](\widehat{\theta}_{l,\textcolor{black}{T}}-\theta_{0,l,\textcolor{black}{T}})|>a/4) \\
%& & + \Pb(\underset{\uu:\|\uu\|_2=C_{\eps}}{\sup}|\overset{d_2}{\underset{k=1}{\sum}} \overset{d_T}{\underset{l,j=1}{\sum}} \uu_k \nabla^3_{\gamma_k\theta_{l,\textcolor{black}{T}}\theta_{j,\textcolor{black}{T}}}\Lb_T(\underline{y};\overline{\theta}_{\textcolor{black}{T}},\gamma_0)(\widehat{\theta}_{l,\textcolor{black}{T}}-\theta_{0,l,\textcolor{black}{T}}) (\widehat{\theta}_{j,\textcolor{black}{T}}-\theta_{0,j,\textcolor{black}{T}})|>a/4).
%\end{eqnarray*}
\textcolor{black}{By the Markov inequality, for any $a>0$:
\begin{eqnarray*}
\Pb(\underset{\uu:\|\uu\|_2=C_{\eps}}{\sup}|\uu^\top\nabla_{\gamma}\Lb_T(\underline{y};\theta_{0,\textcolor{black}{T}},\gamma_0)|>a/3) \leq\big(\frac{C_{\eps}}{a}\big)^2\Eb[\|\nabla_{\gamma}\Lb_T(\underline{y};\theta_{0,\textcolor{black}{T}},\gamma_0)\|^2_2] \leq C_{st}\big(\frac{C_{\eps}}{a}\big)^2\frac{1}{T^2}\overset{T}{\underset{t,t'=1}{\sum}}\kappa(|t-t'|).
\end{eqnarray*}
Then, by Assumption \ref{assumption_first_2ndstep}, we deduce
\beqw
\Pb(\underset{\uu:\|\uu\|_2=C_{\eps}}{\sup}|\uu^\top\nabla_{\gamma}\Lb_T(\underline{y};\theta_{0,\textcolor{black}{T}},\gamma_0)|>a) \leq \frac{L_1C^2_{\eps}}{a^2T}.
\eeqw}
Now we have by the Cauchy-Schwarz inequality and for $L>0$ large enough
\textcolor{black}{\begin{eqnarray*}
\lefteqn{\Pb(\underset{\uu:\|\uu\|_2=C_{\eps}}{\sup}|\overset{d_2}{\underset{k=1}{\sum}} \overset{d_T}{\underset{l=1}{\sum}} \uu_k \nabla^2_{\gamma_k\theta_{l,\textcolor{black}{T}}}\Lb_T(\underline{y};\theta_{0,\textcolor{black}{T}},\gamma_0)(\widehat{\theta}_{l,\textcolor{black}{T}}-\theta_{0,l,\textcolor{black}{T}})|>a/3) }\\
& \leq & \Pb(\underset{\uu:\|\uu\|_2=C_{\eps}}{\sup}\|\widehat\theta_{\textcolor{black}{T}}-\theta_{0,\textcolor{black}{T}}\|_2\big(\overset{d_T}{\underset{l=1}{\sum}} \Big[\overset{d_2}{\underset{k=1}{\sum}} \uu_k \nabla^2_{\gamma_k\theta_{l,\textcolor{black}{T}}}\Lb_T(\underline{y};\theta_{0,\textcolor{black}{T}},\gamma_0)\Big]^2\big)^{1/2}|>a/3) \\
& \leq & \Pb(\underset{\uu:\|\uu\|_2=C_{\eps}}{\sup} L \sqrt{\frac{d_T}{T}}\big(\overset{d_T}{\underset{l=1}{\sum}} \Big[\overset{d_2}{\underset{k=1}{\sum}} \uu_k \nabla^2_{\gamma_k\theta_{l,\textcolor{black}{T}}}\Lb_T(\underline{y};\theta_{0,\textcolor{black}{T}},\gamma_0)\Big]^2\big)^{1/2}|>a/3) \\
&& + \Pb(\|\widehat\theta_{\textcolor{black}{T}}-\theta_{0,\textcolor{black}{T}}\|_2>L\sqrt{\frac{d_T}{T}}).
\end{eqnarray*}}

%\begin{eqnarray*}
%\lefteqn{\Pb(\underset{\uu:\|\uu\|_2=C_{\eps}}{\sup}|\overset{d_2}{\underset{k=1}{\sum}} \overset{d_T}{\underset{l=1}{\sum}} \uu_k \big\{\nabla^2_{\gamma_k\theta_{l,\textcolor{black}{T}}}\Lb_T(\underline{y};\theta_{0,\textcolor{black}{T}},\gamma_0)-\Eb[\nabla^2_{\gamma_k\theta_{l,\textcolor{black}{T}}}\Lb_T(\underline{y};\theta_{0,\textcolor{black}{T}},\gamma_0)]\big\}(\widehat{\theta}_{l,\textcolor{black}{T}}-\theta_{0,l,\textcolor{black}{T}})|>a/4) }\\
%& \leq & \Pb(\underset{\uu:\|\uu\|_2=C_{\eps}}{\sup}\|\widehat\theta_{\textcolor{black}{T}}-\theta_{0,\textcolor{black}{T}}\|_2\big(\overset{d_T}{\underset{l=1}{\sum}} \Big[\overset{d_2}{\underset{k=1}{\sum}} \uu_k \big\{\nabla^2_{\gamma_k\theta_{l,\textcolor{black}{T}}}\Lb_T(\underline{y};\theta_{0,\textcolor{black}{T}},\gamma_0)-\Eb[\nabla^2_{\gamma_k\theta_{l,\textcolor{black}{T}}}\Lb_T(\underline{y};\theta_{0,\textcolor{black}{T}},\gamma_0)]\big\}\Big]^2\big)^{1/2}|>a/4) \\
%& \leq & \Pb(\underset{\uu:\|\uu\|_2=C_{\eps}}{\sup} L \sqrt{\frac{d_T}{T}}\big(\overset{d_T}{\underset{l=1}{\sum}} \Big[\overset{d_2}{\underset{k=1}{\sum}} \uu_k \big\{\nabla^2_{\gamma_k\theta_{l,\textcolor{black}{T}}}\Lb_T(\underline{y};\theta_{0,\textcolor{black}{T}},\gamma_0)-\Eb[\nabla^2_{\gamma_k\theta_{l,\textcolor{black}{T}}}\Lb_T(\underline{y};\theta_{0,\textcolor{black}{T}},\gamma_0)]\big\}\Big]^2\big)^{1/2}|>a/4) \\
%&& + \Pb(\|\widehat\theta_{\textcolor{black}{T}}-\theta_{0,\textcolor{black}{T}}\|_2>L\sqrt{\frac{d_T}{T}}).
%\end{eqnarray*}
By the Markov inequality, we obtain
\textcolor{black}{\begin{eqnarray*}
\lefteqn{\Pb(\underset{\uu:\|\uu\|_2=C_{\eps}}{\sup} L \sqrt{\frac{d_T}{T}}\big(\overset{d_T}{\underset{l=1}{\sum}} \Big[\overset{d_2}{\underset{k=1}{\sum}} \uu_k \nabla^2_{\gamma_k\theta_{l,\textcolor{black}{T}}}\Lb_T(\underline{y};\theta_{0,\textcolor{black}{T}},\gamma_0)\Big]^2\big)^{1/2}|>a/3)} \\
& \leq & \frac{9L^2d_T}{Ta^2} \Eb[\overset{d_T}{\underset{l=1}{\sum}} \Big[\overset{d_2}{\underset{k=1}{\sum}} \uu_k \nabla^2_{\gamma_k\theta_{l,\textcolor{black}{T}}}\Lb_T(\underline{y};\theta_{0,\textcolor{black}{T}},\gamma_0)\Big]^2] \\
& \leq & \frac{9L^2d_T}{Ta^2} \frac{1}{T^2}\overset{T}{\underset{t,t'=1}{\sum}}\overset{d_T}{\underset{l=1}{\sum}} \overset{d_2}{\underset{k,k'=1}{\sum}} \uu_k\uu_{k'} \Eb[\upsilon_{kl,t}\upsilon_{k'l,t'}] \\
& \leq & \frac{9 L^2d_T}{Ta^2} \frac{d_T d_2 C^2_{\eps} K_1}{T},
\end{eqnarray*}}
%\begin{eqnarray*}
%\lefteqn{\Pb(\underset{\uu:\|\uu\|_2=C_{\eps}}{\sup} L \sqrt{\frac{d_T}{T}}\big(\overset{d_T}{\underset{l=1}{\sum}} \Big[\overset{d_2}{\underset{k=1}{\sum}} \uu_k \big\{\nabla^2_{\gamma_k\theta_{l,\textcolor{black}{T}}}\Lb_T(\underline{y};\theta_{0,\textcolor{black}{T}},\gamma_0)-\Eb[\nabla^2_{\gamma_k\theta_{l,\textcolor{black}{T}}}\Lb_T(\underline{y};\theta_{0,\textcolor{black}{T}},\gamma_0)]\big\}\Big]^2\big)^{1/2}|>a/4)} \\
%& \leq & \frac{16L^2d_T}{Ta^2} \Eb[\overset{d_T}{\underset{l=1}{\sum}} \Big[\overset{d_2}{\underset{k=1}{\sum}} \uu_k \big\{\nabla^2_{\gamma_k\theta_{l,\textcolor{black}{T}}}\Lb_T(\underline{y};\theta_{0,\textcolor{black}{T}},\gamma_0)-\Eb[\nabla^2_{\gamma_k\theta_{l,\textcolor{black}{T}}}\Lb_T(\underline{y};\theta_{0,\textcolor{black}{T}},\gamma_0)]\big\}\Big]^2] \\
%& \leq & \frac{16L^2d_T}{Ta^2} \frac{1}{T^2}\overset{T}{\underset{t,t'=1}{\sum}}\overset{d_T}{\underset{l=1}{\sum}} \overset{d_2}{\underset{k,k'=1}{\sum}} \uu_k\uu_{k'} \Eb[\upsilon_{kl,t}\upsilon_{k'l,t'}] \\
%& \leq & \frac{16 L^2d_T}{Ta^2} \frac{d_T d_2 C^2_{\eps} K_1}{T},
%\end{eqnarray*}
using Assumption \ref{assumption_second_cross_2ndstep} for $K_1>0$. 
%Note that using assumption \ref{assumption_second_cross_expectation} and $\|\widehat{\theta}_{\textcolor{black}{T}}-\theta_{0,\textcolor{black}{T}}\|_2=O_p(\sqrt{\frac{d_T}{T}})$, by the Cauchy-Schwarz inequality, then
%\begin{eqnarray*}
%\lefteqn{\Pb(\underset{\uu:\|\uu\|_2=C_{\eps}}{\sup}|\overset{d_2}{\underset{k=1}{\sum}} \overset{d_T}{\underset{l=1}{\sum}} \uu_k \Eb[\nabla^2_{\gamma_k\theta_{l,\textcolor{black}{T}}}\Lb_T(\underline{y};\theta_{0,\textcolor{black}{T}},\gamma_0)](\widehat{\theta}_{l,\textcolor{black}{T}}-\theta_{0,l,\textcolor{black}{T}})|>a/4) }\\
%& \leq & \Pb(\underset{\uu:\|\uu\|_2=C_{\eps}}{\sup} L \sqrt{\frac{d_T}{T}}\big(\overset{d_T}{\underset{l=1}{\sum}} \Big[\overset{d_2}{\underset{k=1}{\sum}} \uu_k \Eb[\nabla^2_{\gamma_k\theta_{l,\textcolor{black}{T}}}\Lb_T(\underline{y};\theta_{0,\textcolor{black}{T}},\gamma_0)]\Big]^2\big)^{1/2}|>a/4) + \Pb(\|\widehat\theta_{\textcolor{black}{T}}-\theta_{0,\textcolor{black}{T}}\|_2>L\sqrt{\frac{d_T}{T}}) \\
%& \leq &  \frac{16 L^2d_T}{Ta^2} \frac{d_T d_2 C^2_{\eps} K_2}{T},
%\end{eqnarray*}
%with $K_2>0$. 
For the third order term of (\ref{taylor_first_order}), we have
\begin{eqnarray*}
\lefteqn{\Pb(\underset{\uu:\|\uu\|_2=C_{\eps}}{\sup}|\overset{d_2}{\underset{k=1}{\sum}} \overset{d_T}{\underset{l,j=1}{\sum}} \uu_k \nabla^3_{\gamma_k\theta_{l,\textcolor{black}{T}}\theta_{j,\textcolor{black}{T}}}\Lb_T(\underline{y};\overline{\theta}_{\textcolor{black}{T}},\gamma_0)(\widehat{\theta}_{l,\textcolor{black}{T}}-\theta_{0,l,\textcolor{black}{T}}) (\widehat{\theta}_{j,\textcolor{black}{T}}-\theta_{0,j,\textcolor{black}{T}})|>a/\textcolor{black}{3})} \\
& \leq & \Pb(\underset{\uu:\|\uu\|_2=C_{\eps}}{\sup}\|\widehat{\theta}_{\textcolor{black}{T}}-\theta_{0,\textcolor{black}{T}}\|^2_2 \big(\overset{d_T}{\underset{l,j=1}{\sum}}\Big[\overset{d_2}{\underset{k=1}{\sum}}  \uu_k \nabla^3_{\gamma_k\theta_{l,\textcolor{black}{T}}\theta_{j,\textcolor{black}{T}}}\Lb_T(\underline{y};\overline{\theta}_{\textcolor{black}{T}},\gamma_0)\Big]^2\big)^{1/2}>a/\textcolor{black}{3}) \\
& \leq & \Pb(\underset{\uu:\|\uu\|_2=C_{\eps}}{\sup}L^2 \frac{d_T}{T} \big(\overset{d_T}{\underset{l,j=1}{\sum}}\Big[\overset{d_2}{\underset{k=1}{\sum}}  \uu_k \nabla^3_{\gamma_k\theta_{l,\textcolor{black}{T}}\theta_{j,\textcolor{black}{T}}}\Lb_T(\underline{y};\overline{\theta}_{\textcolor{black}{T}},\gamma_0)\Big]^2\big)^{1/2}>a/\textcolor{black}{3}) + \Pb(\|\widehat{\theta}_{\textcolor{black}{T}}-\theta_{0,\textcolor{black}{T}}\|_2>L\sqrt{\frac{d_T}{T}}).
\end{eqnarray*}
Now by the Markov inequality, we have
\begin{eqnarray*}
\lefteqn{\Pb(\underset{\uu:\|\uu\|_2=C_{\eps}}{\sup}L^2 \frac{d_T}{T} \big(\overset{d_T}{\underset{l,j=1}{\sum}}\Big[\overset{d_2}{\underset{k=1}{\sum}}  \uu_k \nabla^3_{\gamma_k\theta_{l,\textcolor{black}{T}}\theta_{j,\textcolor{black}{T}}}\Lb_T(\underline{y};\overline{\theta}_{\textcolor{black}{T}},\gamma_0)\Big]^2\big)^{1/2}>a/\textcolor{black}{3})}\\
& \leq & \frac{\textcolor{black}{9}L^4d^2_T}{T^2} \Eb[\overset{d_T}{\underset{l,j=1}{\sum}}\Big[\overset{d_2}{\underset{k=1}{\sum}}  \uu_k \nabla^3_{\gamma_k\theta_{l,\textcolor{black}{T}}\theta_{j,\textcolor{black}{T}}}\Lb_T(\underline{y};\overline{\theta}_{\textcolor{black}{T}},\gamma_0)\Big]^2] \\
& \leq & \frac{\textcolor{black}{9}L^4d^2_T}{T^2a^2} \frac{1}{T^2}\overset{T}{\underset{t,t'=1}{\sum}}\overset{d_T}{\underset{l,j=1}{\sum}} \overset{d_2}{\underset{k,k'=1}{\sum}} \uu_k\uu_{k'} \Eb[\partial^3_{\gamma_k\theta_{l,\textcolor{black}{T}}\theta_{j,\textcolor{black}{T}}}f(y_s,s\leq t;\theta_{0,\textcolor{black}{T}},\gamma_0)\partial^3_{\gamma_{k'}\theta_{l,\textcolor{black}{T}}\theta_{j,\textcolor{black}{T}}}f(y_s,s\leq t';\theta_{0,\textcolor{black}{T}},\gamma_0)] \\
& \leq & \frac{\textcolor{black}{9}L^2d^2_T}{T^2a^2} d^2_T d_2 C^2_{\eps} K_3,
\end{eqnarray*}
using Assumption \ref{assumption_third_cross_1_2ndstep} with $K_3>0$. We thus have managed (\ref{taylor_first_order}). Now, let us focus on the second order derivative $\nabla^2_{\gamma\gamma^\top}\Lb_T(\underline{y};\widehat{\theta}_{\textcolor{black}{T}},\gamma_0)$. By a Taylor expansion, we have
\beq\label{taylor_second_order}
\uu^\top\nabla^2_{\gamma\gamma^\top}\Lb_T(\underline{y};\widehat{\theta}_{\textcolor{black}{T}},\gamma_0)\uu = \uu^\top\nabla^2_{\gamma\gamma^\top}\Lb_T(\underline{y};\theta_{0,\textcolor{black}{T}},\gamma_0)\uu +  \nabla_{\theta_{\textcolor{black}{T}}}\big\{\uu^\top\nabla^2_{\gamma\gamma^\top}\Lb_T(\underline{y};\overline{\theta}_{\textcolor{black}{T}},\gamma_0)\uu\big\}(\widehat{\theta}_{\textcolor{black}{T}}-\theta_{0,\textcolor{black}{T}}),
\eeq
where $\|\overline{\theta}_{\textcolor{black}{T}}-\theta_{0,\textcolor{black}{T}}\|_2 \leq \|\widehat{\theta}_{\textcolor{black}{T}}-\theta_{0,\textcolor{black}{T}}\|_2$. First, we have
\beqw
\uu^\top\nabla^2_{\gamma\gamma^\top}\Lb_T(\underline{y};\theta_{0,\textcolor{black}{T}},\gamma_0)\uu = \Eb[\uu^\top\nabla^2_{\gamma\gamma^\top}\Lb_T(\underline{y};\theta_{0,\textcolor{black}{T}},\gamma_0)\uu] + \Rc_T(\theta_{0,\textcolor{black}{T}},\gamma_0),
\eeqw
where $\Rc_T(\theta_{0,\textcolor{black}{T}},\gamma_0) = \overset{d_2}{\underset{k,l=1}{\sum}}\uu_k\uu_l\big\{\nabla^2_{\gamma_k\gamma_l}\Lb_T(\underline{y};\theta_{0,\textcolor{black}{T}},\gamma_0) - \Eb[\nabla^2_{\gamma_k\gamma_l}\Lb_T(\underline{y};\theta_{0,\textcolor{black}{T}},\gamma_0)]\big\}$. Under \textcolor{black}{Assumption \ref{assumption_second_derivative_2nd_step}}:
\beqw
\|\nabla^2_{\gamma\gamma^\top}\Lb_T(\underline{y};\theta_{0,\textcolor{black}{T}},\gamma_0) - \Eb[\nabla^2_{\gamma\gamma^\top}f(y_s,s\leq t;\theta_{0,\textcolor{black}{T}},\gamma_0)]\|=o_p(1).
\eeqw
Hence, $\Rc_T(\theta_{0,\textcolor{black}{T}},\gamma_0)=o_p(1)$. As for the third order term, the derivatives of the form $\partial^3_{\gamma_k \gamma_l \theta_{j,\textcolor{black}{T}}}f(y_s,s\leq t;\theta_{\textcolor{black}{T}},\gamma)$ are provided in {Appendix} \ref{appendix_derivative}-(ii). For any $b>0$,
\begin{eqnarray*}
\lefteqn{\Pb(\underset{\uu:\|\uu\|_2=C_{\eps}}{\sup}|\nabla_{\theta_{\textcolor{black}{T}}}\big\{\uu^\top\nabla^2_{\gamma\gamma^\top}\Lb_T(\underline{y};\overline{\theta}_{\textcolor{black}{T}},\gamma_0)\uu\big\}(\widehat{\theta}_{\textcolor{black}{T}}-\theta_{0,\textcolor{black}{T}})|>b)}\\
& \leq & \Pb(\underset{\uu:\|\uu\|_2=C_{\eps}}{\sup}\|\widehat{\theta}_{\textcolor{black}{T}}-\theta\|_2 \big(\overset{d_T}{\underset{j=1}{\sum}} \Big[\overset{d_2}{\underset{k,l=1}{\sum}}\uu_k\uu_l\nabla^3_{\gamma_k\gamma_l\theta_{j,\textcolor{black}{T}}}\Lb_T(\underline{y};\overline{\theta}_{\textcolor{black}{T}},\gamma_0)\Big]^2\big)^{1/2}>b) \\
& \leq & \Pb(\underset{\uu:\|\uu\|_2=C_{\eps}}{\sup}L \sqrt{\frac{d_T}{T}}\big(\overset{d_T}{\underset{j=1}{\sum}} \Big[\overset{d_2}{\underset{k,l=1}{\sum}}\uu_k\uu_l\nabla^3_{\gamma_k\gamma_l\theta_{j,\textcolor{black}{T}}}\Lb_T(\underline{y};\overline{\theta}_{\textcolor{black}{T}},\gamma_0)\Big]^2\big)^{1/2}>b) + \Pb(\|\widehat{\theta}_{\textcolor{black}{T}}-\theta_{0,\textcolor{black}{T}}\|_2>L\sqrt{\frac{d_T}{T}}) \\
& \leq & \frac{L^2d_T}{Tb^2}\Eb[\overset{d_T}{\underset{j=1}{\sum}} \Big[\overset{d_2}{\underset{k,l=1}{\sum}}\uu_k\uu_l\nabla^3_{\gamma_k\gamma_l\theta_{j,\textcolor{black}{T}}}\Lb_T(\underline{y};\overline{\theta}_{\textcolor{black}{T}},\gamma_0)\Big]^2] + \Pb(\|\widehat{\theta}_{\textcolor{black}{T}}-\theta_{0,\textcolor{black}{T}}\|_2>L\sqrt{\frac{d_T}{T}})\\
& \leq & \frac{L^2d_T}{Tb^2}\frac{1}{T^2}\overset{T}{\underset{t,t'=1}{\sum}}\overset{d_T}{\underset{j=1}{\sum}}\overset{d_2}{\underset{k,l,k',l'=1}{\sum}}\uu_k\uu_l \uu_{k'}\uu_{l'} \Eb[\nabla^3_{\gamma_k\gamma_l\theta_{j,\textcolor{black}{T}}}f(y_s,s\leq t;\overline{\theta}_{\textcolor{black}{T}},\gamma_0)\nabla^3_{\gamma_{k'}\gamma_{l'}\theta_{j,\textcolor{black}{T}}}f(y_s,s\leq t';\overline{\theta}_{\textcolor{black}{T}},\gamma_0)]\\
& & + \Pb(\|\widehat{\theta}_{\textcolor{black}{T}}-\theta_{0,\textcolor{black}{T}}\|_2>L\sqrt{\frac{d_T}{T}})\\
& \leq & \frac{L^2d_T}{Tb^2} d_T d^2_2 C^4_{\eps} K_3 + \Pb(\|\widehat{\theta}_{\textcolor{black}{T}}-\theta_{0,\textcolor{black}{T}}\|_2>L\sqrt{\frac{d_T}{T}}),
\end{eqnarray*}
for $L,K_3>0$ and using Assumption \ref{assumption_third_cross_1_2ndstep}. We have thus controlled for (\ref{taylor_second_order}). Putting the pieces together, we are in a position to bound probability (\ref{objective_loss}). Denoting $\delta_T = \lambda_{\min}(\Eb[\nabla^2_{\gamma\gamma^\top}f(y_s,s\leq t;\theta_{0,\textcolor{black}{T}},\gamma_0)])C^2_{\eps}\nu_T/2$ and using $\frac{\nu_T}{2}\Eb[\uu^\top\nabla^2_{\gamma\gamma^\top}f(y_s,s\leq t;\theta_{0,\textcolor{black}{T}},\gamma_0)\uu] \geq \delta_T$, we have
{\small{\begin{eqnarray}
\lefteqn{\Pb(\exists \uu: \|\uu\|_2=C_{\eps}:\nabla_{\gamma}\Lb_T(\underline{y};\widehat{\theta}_{\textcolor{black}{T}},\gamma_0) \uu + \frac{\nu_T}{2}\uu^\top\nabla^2_{\gamma\gamma^\top}\Lb_T(\underline{y};\widehat{\theta}_{\textcolor{black}{T}},\gamma_0)\uu \leq 0)}\nonumber\\
& \hspace*{-0.5cm}\leq & \hspace*{-0.3cm}\Pb(\exists \uu: \|\uu\|_2=C_{\eps}:|\nabla_{\gamma}\Lb_T(\underline{y};\widehat{\theta}_{\textcolor{black}{T}},\gamma_0)\uu|>\delta_T/4) + \Pb(\exists \uu: \|\uu\|_2=C_{\eps}:|\frac{\nu_T}{2}\Rc_T(\theta_{0,\textcolor{black}{T}},\gamma_0)|>\delta_T/4) \nonumber\\
& \hspace*{-0.8cm}& \hspace*{-1.1cm}+ \Pb(\exists \uu: \|\uu\|_2=C_{\eps}:\underset{\overline{\theta}_{\textcolor{black}{T}}:\|\overline{\theta}_{\textcolor{black}{T}}-\theta_{0,\textcolor{black}{T}}\|_2\leq L(\frac{d_T}{T})^{1/2}}{\sup}|\frac{\nu_T}{2}\nabla_{\theta_{\textcolor{black}{T}}}\big\{\uu^\top\nabla^2_{\gamma\gamma^\top}\Lb_T(\underline{y};\overline{\theta}_{\textcolor{black}{T}},\gamma_0)\uu\big\}(\widehat{\theta}_{\textcolor{black}{T}}-\theta_{0,\textcolor{black}{T}})|>\delta_T/4). \label{upper_prob}
\end{eqnarray}}}
First, we have 
\textcolor{black}{\begin{eqnarray*}
\lefteqn{\Pb(\exists \uu: \|\uu\|_2=C_{\eps}:|\nabla_{\gamma}\Lb_T(\underline{y};\widehat{\theta}_{\textcolor{black}{T}},\gamma_0)\uu|>\delta_T/4) \leq \Pb(\underset{\uu:\|\uu\|_2=C_{\eps}}{\sup}|\uu^\top\Lb_T(\underline{y};\theta_{0,\textcolor{black}{T}},\gamma_0)|>\delta_T/12)}\\
&  &  + \Pb(\underset{\uu:\|\uu\|_2=C_{\eps}}{\sup}|\overset{d_2}{\underset{k=1}{\sum}} \overset{d_T}{\underset{l=1}{\sum}} \uu_k \nabla^2_{\gamma_k\theta_{l,\textcolor{black}{T}}}\Lb_T(\underline{y};\theta_{0,\textcolor{black}{T}},\gamma_0)(\widehat{\theta}_{l,\textcolor{black}{T}}-\theta_{0,l,\textcolor{black}{T}})|>\delta_T/12) \\
& & + \Pb(\underset{\uu:\|\uu\|_2=C_{\eps}}{\sup}|\overset{d_2}{\underset{k=1}{\sum}} \overset{d_T}{\underset{l,j=1}{\sum}} \uu_k \nabla^3_{\gamma_k\theta_{l,\textcolor{black}{T}}\theta_{j,\textcolor{black}{T}}}\Lb_T(\underline{y};\overline{\theta}_{\textcolor{black}{T}},\gamma_0)(\widehat{\theta}_{l,\textcolor{black}{T}}-\theta_{0,l,\textcolor{black}{T}}) (\widehat{\theta}_{j,\textcolor{black}{T}}-\theta_{0,j,\textcolor{black}{T}})|>\delta_T/12) \\
& \leq & \frac{C^2_{\eps}L_1}{T\delta^2_T} + \frac{d^2_Td_2C^2_{\eps}L_2}{T^2\delta^2_T} + \eps/8 +  \frac{C^2_{\eps}d^4_T d_2 L_3}{T^2 \delta^2_T} + \eps/8.
\end{eqnarray*}}
Moreover, $\Pb(\exists \uu: \|\uu\|_2=C_{\eps}:|\frac{\nu_T}{2}\Rc_T(\theta_{0,\textcolor{black}{T}},\gamma_0)|>\delta_T/4) < \eps/\textcolor{black}{8}$, and
{\footnotesize{\begin{eqnarray*}
\Pb(\exists \uu: \|\uu\|_2=C_{\eps}:\underset{\overline{\theta}:\|\overline{\theta}_{\textcolor{black}{T}}-\theta_{0,\textcolor{black}{T}}\|_2\leq L\sqrt{d_T/T}}{\sup}|\frac{\nu_T}{2}\nabla_{\theta_{\textcolor{black}{T}}}\big\{\uu^\top\nabla^2_{\gamma\gamma^\top}\Lb_T(\underline{y};\overline{\theta}_{\textcolor{black}{T}},\gamma_0)\uu\big\}(\widehat{\theta}_{\textcolor{black}{T}}-\theta_{0,\textcolor{black}{T}})|>\delta_T/4)
 \leq \frac{d^2_T d^2_2 L_4 C^4_{\eps}\nu^2_T}{T \delta^2_T} + \eps/\textcolor{black}{8}.
\end{eqnarray*}}}
As a consequence, for $L_1,L_2,L_3,L_4>0$, with $\nu_T=\frac{1}{\sqrt{T}}$, (\ref{upper_prob}) can be bounded as
\begin{eqnarray*}
\lefteqn{\Pb(\exists \uu: \|\uu\|_2=C_{\eps}:\nabla_{\gamma}\Lb_T(\underline{y};\widehat{\theta}_{\textcolor{black}{T}},\gamma_0) \uu + \frac{\nu_T}{2}\uu^\top\nabla^2_{\gamma\gamma^\top}\Lb_T(\underline{y};\widehat{\theta}_{\textcolor{black}{T}},\gamma_0)\uu \leq 0)}\\
& \leq &  \frac{C^2_{\eps}L_1}{T\delta^2_T} + \frac{d^2_Td_2C^2_{\eps}L_2}{T^2\delta^2_T} + \eps/\textcolor{black}{8} +  \frac{C^2_{\eps}d^4_T d_2 L_3}{T^3\delta^2_T} + \eps/\textcolor{black}{8} + \eps/\textcolor{black}{8} + \frac{d^2_T d^2_2 L_4 C^4_{\eps}\nu^2_T}{T \delta^2_T} + \eps/\textcolor{black}{8}\\
& \leq & \frac{C_{1}}{C^2_{\eps}} + \frac{d^2_Td_2 C_{2}}{T C^2_{\eps}} + \eps/\textcolor{black}{8} +\frac{d^4_T d_2 C_{3}}{T^2C^2_{\eps}} + \eps/\textcolor{black}{8} + \eps/\textcolor{black}{8} + \frac{d^2_T d^2_2 C_4}{T} + \eps/\textcolor{black}{8},
\end{eqnarray*}
where $C_1,C_2,C_3,C_4$ are strictly positive constants. Under the scaling behaviour of Theorem \ref{oracle_theorem}, we deduce that for $C_{\eps}$ large enough, $T$ large enough,
\beqw
\Pb(\exists \uu: \|\uu\|_2=C_{\eps}:\nabla_{\gamma}\Lb_T(\underline{y};\widehat{\theta}_{\textcolor{black}{T}},\gamma_0) \uu + \frac{\nu_T}{2}\uu^\top\nabla^2_{\gamma\gamma^\top}\Lb_T(\underline{y};\widehat{\theta}_{\textcolor{black}{T}},\gamma_0)\uu \leq 0) < \eps.
\eeqw
\end{proof}

{To derive the asymptotic distribution of the second step estimator, we assume the following conditions.
\begin{assumption}\label{assumption_varcov_2ndstep}
The $d_2$ square matrices $\textcolor{black}{\Ub}:=\Eb[\big(I_p\otimes K_{t-1}(\theta_{0,\textcolor{black}{T}})K_{t-1}(\theta_{0,\textcolor{black}{T}})^\top\big)]$ and $\textcolor{black}{\Wb}:=\Eb[\big(K_{t-1}(\theta_{0,\textcolor{black}{T}})\otimes \big\{\textcolor{black}{y_{t}^{\ell}}-\Gamma_0K_{t-1}(\theta_{0,\textcolor{black}{T}})\big\}\big)\big(K_{t-1}(\theta_{0,\textcolor{black}{T}})\otimes \big\{\textcolor{black}{y_{t}^{\ell}}-\Gamma_0K_{t-1}(\theta_{0,\textcolor{black}{T}})\big\}\big)^\top]$ exist and are positive definite. Moreover, the $d_2 \times k_T$ matrix $\textcolor{black}{\Upsilon_{\gamma \Ac_T} = \Eb[\partial^2_{\gamma_l \theta_{k,T}}f(y_s,s\leq t;\theta_{0,\textcolor{black}{T}},\gamma_0)]_{1\leq l \leq d_2,k \in \Ac_T }}$ and \textcolor{black}{the $k_T \times d_2$ matrix $\Jb_{\Ac_T\gamma} = \Eb[\nabla_{\theta_{\Ac_T}}\ell(y_s,s\leq t;\theta_{0,T})\nabla_{\gamma^\top}f(y_s,s\leq t;\theta_{0,T},\gamma_0)]$ exist.}
\end{assumption}

\textcolor{black}{\begin{assumption}\label{assumption_second_cross_2ndstep_centered}
Let $\lambda_{kl,t}= \partial^2_{\gamma_k\theta_{l,\textcolor{black}{T}}}f(y_s,s\leq t;\theta_{0,\textcolor{black}{T}},\gamma_0)-\Eb[\partial^2_{\gamma_k\theta_{l,\textcolor{black}{T}}}f(y_s,s\leq t;\theta_{0,\textcolor{black}{T}},\gamma_0)]$, for any $k=1,\cdots,d_2$ and $l=1,\cdots,d_T$. There is some function $\psi(.)$ such that for any $T$: $|\Eb[\lambda_{kl,t}\lambda_{k'l',t'}] | \leq \psi(|t-t'|), \; \text{and} \;\underset{T>0}{\sup} \; \frac{1}{T}\overset{T}{\underset{t,t'=1}{\sum}}\psi(|t-t'|)<\infty$.
\end{assumption}}

% Moreover, the $d_2 \times k_T$ matrix $\textcolor{black}{\Upsilon_{\gamma \Ac_T} = \Eb[\partial^2_{\gamma_l \theta_{k,T}}f(y_s,s\leq t;\theta_{0,\textcolor{black}{T}},\gamma_0)]_{1\leq l \leq d_2,k \in \Ac_T }}$ exists with \textcolor{black}{$\sigma_{1}(\Upsilon_{\gamma \Ac_T} )<\infty$, and the $k_T \times d_2$ matrix $\Jb_{\Ac_T\gamma} = \Eb[\nabla_{\theta_{\Ac_T}}\ell(y_s,s\leq t;\theta_{0,T})\nabla_{\gamma^\top}f(y_s,s\leq t;\theta_{0,T},\gamma_0)]$ exists with $\sigma_{1}(\Jb_{\Ac_T\gamma} )<\infty$, where $\sigma_{1}(.)$ is the maximum singular value.

\textcolor{black}{\begin{assumption}\label{assumption_array_2nd}
Let $Z_{T,t}=\sqrt{T}\big(\big(\Upsilon_{\gamma\Ac_T}\Hb^{-1}_{\Ac_{\textcolor{black}{T}}\Ac_{\textcolor{black}{T}}}\nabla_{\theta_{\Ac_{\textcolor{black}{T}}}}\Gb_{T,t}(\underline{y};\theta_{0,\textcolor{black}{T}})\big)^\top,\big(\nabla_{\gamma}\Lb_{T,t}(\underline{y};\theta_{0,T},\gamma_0)\big)^\top\big)^\top$ with  $\nabla_{\theta_{\Ac_{\textcolor{black}{T}}}}\Gb_{T,t}(\underline{y};\theta_{0,\textcolor{black}{T}}) = - \frac{1}{T}\big(Z_{m,t-1}\otimes \{x_t-\Psi_{0,1:m} Z_{m,t-1}\}\big)_{\Ac_{\textcolor{black}{T}}} \in \Rb^{k_T}$ and $\nabla_{\gamma}\Lb_{T,t}(\underline{y};\theta_{0,T},\gamma_0) = \frac{1}{T}\nabla_{\gamma}f(y_s,s\leq t;\theta_{0,T},\gamma_0) = -\frac{1}{T}\big(K_{t-1}(\theta_{0,T})\otimes \{y^{\ell}_{t}-\Gamma_0K_{t-1}(\theta_{0,T})\}\big) \in \Rb^{d_2}$, $\Upsilon_{\gamma\Ac_T}$ the $d_2 \times \normalfont\text{card}(\Ac_{\textcolor{black}{T}})$ matrix defined in Assumption \ref{assumption_varcov_2ndstep}. Let $\Fc^T_{t} = \sigma(Z_{T,s},s\leq t)$, then $Z_{T,t}$ is a martingale difference and we have
\beqw
\Eb\Big[\underset{1 \leq k,l \leq d_T}{\sup}\Eb[\big\{\partial_{\theta_{k,\textcolor{black}{T}}}\ell(y_s,s\leq t;\theta_{0,\textcolor{black}{T}})\partial_{\theta_{l,\textcolor{black}{T}}}\ell(y_s,s\leq t;\theta_{0,\textcolor{black}{T}})\big\}^2|\Fc^T_{1,t-1}]\lambda_{\max,t-1}(\Hb^T_{t-1}) \Big] \leq \overline{B}_1<\infty,
\eeqw
\beqw
\Eb\Big[\underset{1 \leq k,l \leq d_2}{\sup}\Eb[\big\{\partial_{\gamma_{k}}f(y_s,s\leq t;\theta_{0,\textcolor{black}{T}},\gamma_0)\partial_{\gamma_{l}}f(y_s,s\leq t;\theta_{0,\textcolor{black}{T}},\gamma_0)\big\}^2|\Fc^T_{t-1}]\lambda_{\max,t-1}(\Hb^T_{t-1}) \Big] \leq \overline{B}_2<\infty,
\eeqw
\beqw
\Eb\Big[\underset{1 \leq k \leq d_T, 1 \leq l \leq d_2}{\sup}\Eb[\big\{\partial_{\theta_{k,\textcolor{black}{T}}}\ell(y_s,s\leq t;\theta_{0,\textcolor{black}{T}})\partial_{\gamma_{l}}f(y_s,s\leq t;\theta_{0,\textcolor{black}{T}},\gamma_0)\big\}^2|\Fc^T_{t-1}]\lambda_{\max,t-1}(\Hb^T_{t-1}) \Big] \leq \overline{B}_3<\infty,
\eeqw
with $\Hb^T_{t-1} = \Eb[\alpha_T\alpha^\top_T|\Fc^T_{t-1}], \; \alpha_T = \big(\big(\nabla_{\theta_{\textcolor{black}{\Ac_T}}}\ell(y_s,s\leq t;\theta_{0,\textcolor{black}{T}})\big)^\top,\big(\nabla_{\gamma}f(y_s,s\leq t;\theta_{0,\textcolor{black}{T}},\gamma_0)\big)^\top\big)^\top$,
and $\lambda_{\max,t-1}(\Hb^T_{t-1})<\infty$.
\end{assumption}}
\textcolor{black}{As in Assumption \ref{assumption_array}, the conditions stated in Assumption \ref{assumption_array_2nd} may be artificial but are necessary to verify the Lindeberg condition since both losses are triangular arrays.}
\begin{assumption}\label{assumption_third_cross_2_2ndstep}
For almost all observations, the derivatives $\partial^3_{\gamma_j \theta_{k,\textcolor{black}{T}}\theta_{l,\textcolor{black}{T}}}\ell(y_s,s\leq t;\theta_{\textcolor{black}{T}},\gamma)$, $\partial^3_{\gamma_j \theta_{k,\textcolor{black}{T}}\gamma_i}\ell(y_s,s\leq t;\theta_{\textcolor{black}{T}},\gamma)$ and $\partial^3_{\gamma_i \gamma_j\theta_{k,\textcolor{black}{T}}}\ell(y_s,s\leq t;\theta_{\textcolor{black}{T}},\gamma)$ exist. Let $\eta>0$ and define
\beqw
\begin{array}{llll}
H_t(L)&=&\underset{\gamma:\|\gamma-\gamma_0\|<\eta}{\sup}\big\{\underset{\theta: \|\theta-\theta_{0,\textcolor{black}{T}}\|_2\leq L (d_T/T)^{1/2}}{\sup}|\partial^3_{\gamma_j \theta_{k,\textcolor{black}{T}}\theta_{l,\textcolor{black}{T}}}\ell(y_s,s\leq t;\theta_{\textcolor{black}{T}},\gamma)|\big\}, 1 \leq j \leq d_2, 1 \leq k,l \leq d_T, \\
R_t(L)&=&\underset{\gamma:\|\gamma-\gamma_0\|<\eta}{\sup}\big\{\underset{\theta: \|\theta-\theta_{0,\textcolor{black}{T}}\|_2\leq L (d_T/T)^{1/2}}{\sup}|\partial^3_{\gamma_j \theta_{k,\textcolor{black}{T}}\gamma_i}\ell(y_s,s\leq t;\theta_{\textcolor{black}{T}},\gamma)|\big\}, 1 \leq j,i \leq d_2, 1 \leq k \leq d_T, \\
M_t(L)&=&\underset{\gamma:\|\gamma-\gamma_0\|<\eta}{\sup}\big\{\underset{\theta: \|\theta-\theta_{0,\textcolor{black}{T}}\|_2\leq L (d_T/T)^{1/2}}{\sup}|\partial^3_{\gamma_i\gamma_j \theta_{k,\textcolor{black}{T}}}\ell(y_s,s\leq t;\theta_{\textcolor{black}{T}},\gamma)|\big\}, 1 \leq j,i \leq d_2, 1 \leq k \leq d_T,
\end{array}
\eeqw
where $0 < L < \infty$. Then
\beqw
\frac{1}{T^2}\overset{T}{\underset{t,t'=1}{\sum}}\Eb[H_{t}(L)H_{t'}(L)]<\infty, \;\; \frac{1}{T^2}\overset{T}{\underset{t,t'=1}{\sum}}\Eb[R_{t}(L)R_{t'}(L)]<\infty, \;\; \frac{1}{T^2}\overset{T}{\underset{t,t'=1}{\sum}}\Eb[M_{t}(L)M_{t'}(L)]<\infty.
\eeqw
\end{assumption}
These moment assumptions are similar to Assumption G of Fan and Peng (2004), but they are adapted to the dependent case. The derivatives can be found in Appendix \ref{appendix_derivative}.}

\begin{proof}[Proof of Theorem \ref{distr_second}.]
Through a Taylor expansion, we obtain for the $\widehat{\gamma}$ component
\beqw
0=\nabla_{\gamma} \Lb_T(\underline{y};\widehat{\theta}_{\Ac_T},\widehat{\gamma}) = \nabla_{\gamma} \Lb_T(\underline{y};\theta_{0,\Ac_T},\gamma_0) + \nabla^2_{\gamma \theta^\top_{\Ac_T}} \Lb_T(\underline{y};\overline{\theta}_{\Ac_T},\overline{\gamma}) (\widehat{\theta}_{\textcolor{black}{T}}-\theta_{0,\textcolor{black}{T}})_{\Ac_{\textcolor{black}{T}}} + \nabla^2_{\gamma\gamma^\top} \Lb_T(\underline{y};\overline{\theta}_{\Ac_T},\overline{\gamma}) (\widehat{\gamma}-\gamma_0),
\eeqw
where $\|\overline{\theta}_{\textcolor{black}{T}}-\theta_{0,\textcolor{black}{T}}\|_2 \leq \|\widehat{\theta}_{\textcolor{black}{T}}-\theta_{0,\textcolor{black}{T}}\|_2$ and $\|\overline{\gamma}-\gamma\| \leq \|\widehat{\gamma}-\gamma\|$. Here, $\widehat{\theta}_{\Ac_T} \in \Rb^{k_T}$. Then inverting this relationship, multiplying by $\sqrt{T}$ and using the asymptotic expansion of the first step estimator, we obtain
{\footnotesize{\begin{eqnarray*}
\lefteqn{\sqrt{T}(\widehat{\gamma}-\gamma_0)}\\
&=& (-\nabla^2_{\gamma\gamma^\top} \Lb_T(\underline{y};\overline{\theta}_{\Ac_T},\overline{\gamma}) )^{-1} \nabla^2_{\gamma\theta^\top_{\Ac_T}} \Lb_T(\underline{y};\overline{\theta}_{\Ac_T},\overline{\gamma}) \sqrt{T}(\widehat{\theta}_{\textcolor{black}{T}}-\theta_{0,\textcolor{black}{T}})_{\Ac_{\textcolor{black}{T}}} + (-\nabla^2_{\gamma\gamma^\top} \Lb_T(\underline{y};\overline{\theta}_{\Ac_T},\overline{\gamma}) )^{-1} \sqrt{T}\nabla_{\gamma} \Lb_T(\underline{y};\theta_{0,\Ac_T},\gamma_0).
\end{eqnarray*}}}
\textcolor{black}{To ease the notations, we omit the $\Ac_T$ index with respect to the arguments in $\Lb_T(\underline{y};.)$. Let us control for $\|\nabla^2_{\gamma\theta^\top_{\Ac_T}} \Lb_T(\underline{y};\overline{\theta}_T,\overline{\gamma}) - \nabla^2_{\gamma\theta^\top_{\Ac_T}} \Lb_T(\underline{y};\theta_{0,T},\gamma_0)\|_F$. We have the expansion:
\begin{eqnarray}
\lefteqn{\nabla^2_{\gamma\theta^\top_{\Ac_T}} \Lb_T(\underline{y};\overline{\theta}_{\textcolor{black}{T}},\overline{\gamma}) = \nabla^2_{\gamma\theta^\top_{\Ac_T}} \Lb_T(\underline{y};\theta_{0,\textcolor{black}{T}},\gamma_0) }\nonumber \\
& &  +(\overline{\theta}_{\textcolor{black}{T}}-\theta_{0,\textcolor{black}{T}})^\top_{\Ac_T} \nabla_{\theta_{\Ac_T}}\big\{\nabla^2_{\gamma\theta^\top_{\Ac_T}} \Lb_T(\underline{y};\widetilde{\theta}_{\textcolor{black}{T}},\widetilde{\gamma}) \big\} + (\overline{\gamma}-\gamma_0)^\top \nabla_{\gamma}\big\{\nabla^2_{\gamma\theta^\top_{\Ac_T}}\Lb_T(\underline{y};\widetilde{\theta}_{\textcolor{black}{T}},\widetilde{\gamma})\big\}, \label{expansion1}
\end{eqnarray}
where $\|\widetilde{\theta}_{\textcolor{black}{T}}-\theta_{0,\textcolor{black}{T}}\|_2\leq \|\overline{\theta}_{\textcolor{black}{T}}-\theta_{0,\textcolor{black}{T}}\|_2$ and $\|\widetilde{\gamma}-\gamma_0\|\leq \|\overline{\gamma}-\gamma_0\|$. The second term of (\ref{expansion1}) can be bounded as:
\begin{eqnarray*}
\|(\overline{\theta}_{\textcolor{black}{T}}-\theta_{0,\textcolor{black}{T}})^\top_{\Ac_T} \nabla_{\theta_{\Ac_T}}\big\{\nabla^2_{\gamma\theta^\top_{\Ac_T}} \Lb_T(\underline{y};\widetilde{\theta}_{\textcolor{black}{T}},\widetilde{\gamma})\big\}\|^2_F \leq \underset{1\leq j \leq d_2,k,l\in \Ac_{\textcolor{black}{T}}}{\sum} \{\partial^3_{\gamma_j \theta_{k,\textcolor{black}{T}}\theta_{l,\textcolor{black}{T}}}\Lb_T(\underline{y};\widetilde{\theta}_{\textcolor{black}{T}},\widetilde{\gamma})\}^2\|\overline{\theta}_{\textcolor{black}{T}}-\theta_{0,\textcolor{black}{T}}\|^2_2.
\end{eqnarray*}
Under Assumption \ref{assumption_third_cross_2_2ndstep}, using $\|\widehat{\theta}_{\textcolor{black}{T}}-\theta_{0,\textcolor{black}{T}}\|_2=O_p(\sqrt{\frac{d_T}{T}})$, we deduce 
\beqw
\|(\overline{\theta}_{\textcolor{black}{T}}-\theta_{0,\textcolor{black}{T}})^\top_{\Ac_T} \nabla_{\theta_{\Ac_T}}\big\{\nabla^2_{\gamma\theta^\top_{\Ac_T}} \Lb_T(\underline{y};\widetilde{\theta}_{\textcolor{black}{T}},\widetilde{\gamma})\big\}\|_F \leq O_p(\big(d^2_T\frac{d_T}{T}\big)^{1/2}).
\eeqw
As for the third term, using a similar reasoning, for any $j=1,\cdots,d_2$, we have
\begin{eqnarray*}
\|(\overline{\gamma}-\gamma_0)^\top\nabla_{\gamma}\big\{\nabla^2_{\gamma\theta^\top_{\Ac_T}} \Lb_T(\underline{y};\widetilde{\theta}_{\textcolor{black}{T}},\widetilde{\gamma})\big\}\|^2_F \leq \underset{k\in \Ac_{\textcolor{black}{T}}}{\sum}\overset{d_2}{\underset{j,l=1}{\sum}}\big\{\partial^3_{\gamma_j \gamma_k\gamma_l}\Lb_T(\underline{y};\widetilde{\theta}_{\textcolor{black}{T}},\widetilde{\gamma})\big\}^2 \|\overline{\gamma}-\gamma_0\|^2_2.
\end{eqnarray*}
Under Assumption \ref{assumption_third_cross_2_2ndstep}, using $\|\widehat{\gamma}-\gamma_0\|_2=O_p(\frac{1}{\sqrt{T}})$, we deduce
\beqw
\|(\overline{\gamma}-\gamma_0)^\top\nabla_{\gamma}\big\{\nabla^2_{\gamma\theta^\top_{\Ac_T}} \Lb_T(\underline{y};\widetilde{\theta}_{\textcolor{black}{T}},\widetilde{\gamma})\big\}\|_F\leq O_p(\big(\frac{d_T}{T}\big)^{1/2}).
\eeqw
Finally, using Assumption \ref{assumption_second_cross_2ndstep_centered}:
\begin{eqnarray*}
\lefteqn{\big(\nabla^2_{\gamma\theta^\top_{\Ac_T}}\Lb_T(\underline{y};\theta_{0,\textcolor{black}{T}},\gamma_0)-\Eb[\nabla^2_{\gamma\theta^\top_{\Ac_T}}\Lb_T(\underline{y};\theta_{0,\textcolor{black}{T}},\gamma_0)]\big)(\widehat{\theta}_T-\theta_{0,T})_{\Ac_T}}\\
& \leq & \|\big\{\nabla^2_{\gamma\theta^\top_{\Ac_T}}\Lb_T(\underline{y};\theta_{0,\textcolor{black}{T}},\gamma_0)-\Eb[\nabla^2_{\gamma\theta^\top_{\Ac_T}}\Lb_T(\underline{y};\theta_{0,\textcolor{black}{T}},\gamma_0)]\big\}\|_F\|(\widehat{\theta}_T-\theta_{0,T})_{\Ac_T}\|_2 = O_p(\frac{d_T}{\sqrt{T}})O_p(\sqrt{\frac{d_T}{T}}).
\end{eqnarray*}
Putting the pieces together, we deduce
\beqw
\nabla^2_{\gamma\theta^\top_{\Ac_T}} \Lb_T(\underline{y};\overline{\theta}_{\textcolor{black}{T}},\overline{\gamma}) \sqrt{T}(\widehat{\theta}_{\textcolor{black}{T}}-\theta_{0,\textcolor{black}{T}})_{\Ac_{\textcolor{black}{T}}}  =  \Eb[\nabla^2_{\gamma\theta^\top_{\Ac_T}} \Lb_T(\underline{y};\theta_{0,\textcolor{black}{T}},\gamma_0)] \sqrt{T}(\widehat{\theta}_{\textcolor{black}{T}}-\theta_{0,\textcolor{black}{T}})_{\Ac_{\textcolor{black}{T}}} + o_p(1). 
\eeqw}
To prove $\|\nabla^2_{\gamma\gamma^\top}\Lb_T(\underline{y};\overline{\theta}_{\textcolor{black}{T}},\overline{\gamma})-\Eb[\nabla^2_{\gamma\gamma^\top} f(y_s,s\leq t;\theta_{0,\textcolor{black}{T}},\gamma_0)]\|=o_p(1)$, we can use a similar reasoning. We have the expansion
\begin{eqnarray*}
\lefteqn{\nabla^2_{\gamma\gamma^\top} \Lb_T(\underline{y};\overline{\theta}_{\textcolor{black}{T}},\overline{\gamma}) }\\
& = & \nabla^2_{\gamma\gamma^\top} \Lb_T(\underline{y};\theta_{0,\textcolor{black}{T}},\gamma_0)+ (\overline{\theta}_{\textcolor{black}{T}}-\theta_{0,\textcolor{black}{T}})^\top_{\Ac_T}\nabla_{\theta_{\Ac_T}}\big\{\nabla^2_{\gamma\gamma^\top} \Lb_T(\underline{y};\widetilde{\theta}_{\textcolor{black}{T}},\widetilde{\gamma}) \big\} + (\overline{\gamma}-\gamma_0)^\top\nabla_{\gamma}\big\{\nabla^2_{\gamma\gamma^\top}\Lb_T(\underline{y};\widetilde{\theta}_{\textcolor{black}{T}},\widetilde{\gamma})\alpha_T\big\} \\
& = & \nabla^2_{\gamma\gamma^\top} \Lb_T(\underline{y};\theta_{0,\textcolor{black}{T}},\gamma_0)+(\overline{\theta}_{\textcolor{black}{T}}-\theta_{0,\textcolor{black}{T}})^\top_{\Ac_T} \nabla_{\theta_{\Ac_T}}\big\{\nabla^2_{\gamma\gamma^\top} \Lb_T(\underline{y};\widetilde{\theta}_{\textcolor{black}{T}},\widetilde{\gamma}) \big\},
\end{eqnarray*}
since $\forall k,l,m \leq d_2, \partial^3_{\gamma_k\gamma_l\gamma_m}f(y_s,s\leq t;\theta_{\textcolor{black}{T}},\gamma)=0$, and
where $\|\widetilde{\theta}_{\textcolor{black}{T}}-\theta_{0,\textcolor{black}{T}}\|_2\leq \|\overline{\theta}_{\textcolor{black}{T}}-\theta_{0,\textcolor{black}{T}}\|_2$ and $\|\widetilde{\gamma}-\gamma_0\|\leq \|\overline{\gamma}-\gamma_0\|$. Using the $\sqrt{T/d_T}$-consistency of $\|\widehat{\theta}_{\textcolor{black}{T}}-\theta_{0,\textcolor{black}{T}}\|_2$, the Cauchy-Schwarz inequality and Assumption \ref{assumption_third_cross_2_2ndstep}, we deduce $\|\nabla^2_{\gamma\gamma^\top}\Lb_T(\underline{y};\overline{\theta}_{\textcolor{black}{T}},\overline{\gamma})-\Eb[\nabla^2_{\gamma\gamma^\top} f(y_s,s\leq t;\theta_{0,\textcolor{black}{T}},\gamma_0)]\| =o_p(1)$ and we denote $\textcolor{black}{\Ub}:=\Eb[\nabla^2_{\gamma\gamma^\top} f(y_s,s\leq t;\theta_{0,\textcolor{black}{T}},\gamma_0)]$. \textcolor{black}{Thus, by Lemma 11 of Loh and Wainwright (2017), we deduce $\|\big(\nabla^2_{\gamma\gamma^\top}\Lb_T(\underline{y};\overline{\theta}_{\textcolor{black}{T}},\overline{\gamma})\big)^{-1}-\Ub^{-1}\| =o_p(1)$. }

\medskip

\textcolor{black}{Now, let $\textcolor{black}{\Upsilon_{\gamma \Ac_T} = \Eb[\partial^2_{\gamma_l\theta_{k,T}}f(y_s,s\leq t;\theta_{0,\textcolor{black}{T}},\gamma_0)]_{1 \leq l \leq d_2, k \in \Ac_T}}$ with dimension $d_2 \times k_T$, for $T$ large enough:
$\sqrt{T}\Upsilon_{\gamma \Ac_T}(\widehat{\theta}_{\textcolor{black}{T}}-\theta_{0,\textcolor{black}{T}})_{\Ac_{\textcolor{black}{T}}} = \sqrt{T}\Upsilon_{\gamma \Ac_T}\Hb^{-1}_{\Ac_T\Ac_T} \nabla_{\theta_{\Ac_T}}\Gb_T(\underline{y};\theta_{0,T})$. Under Assumptions \ref{assumption_array_2nd},  $Z_{T,t}=\sqrt{T}\big(\big(\Upsilon_{\gamma \Ac_T}\Hb^{-1}_{\Ac_T\Ac_T} \nabla_{\theta_{\Ac_T}}\Gb_{T,t}(\underline{y};\theta_{0,T})\big)^\top,\big(\nabla_{\gamma}\Lb_{T,t}(\underline{y};\theta_{0,T},\gamma_0)\big)^\top\big)^\top$, is a martingale difference. Moreover, we have:
\beqw
Z_{T,t} = \sqrt{T}\Kb_T\begin{pmatrix}
\nabla_{\theta_{\Ac_T}}\Gb_{T,t}(\underline{y};\theta_{0,T}) \\  \nabla_{\gamma}\Lb_{T,t}(\underline{y};\theta_{0,T},\gamma_0)
\end{pmatrix}, \; \Kb_T = \begin{pmatrix}
\Upsilon_{\gamma \Ac_T}\Hb^{-1}_{\Ac_T\Ac_T} & \mathbf{0} \\ \mathbf{0} & I_{d_2}
\end{pmatrix}.
\eeqw
Proceeding in the manner as in the proof of the asymptotic distribution of the oracle estimator $\widehat{\theta}_{\Ac_T}$, under the conditions of Assumption \ref{assumption_array_2nd}, we deduce that $Z_{T,t}$ satisfies the Lindeberg condition. By Theorem \ref{lindeberg_shiryaev}, we deduce  $\sqrt{T}\big(\big(\Upsilon_{\gamma \Ac_T}\Hb^{-1}_{\Ac_T\Ac_T} \nabla_{\theta_{\Ac_T}}\Gb_{T}(\underline{y};\theta_{0,T})\big)^\top,\big(\nabla_{\gamma}\Lb_{T}(\underline{y};\theta_{0,T},\gamma_0)\big)^\top\big)^\top$ is asymptotically normal.} Then, by Slutsky's Theorem, we deduce 
\beqw
\sqrt{T}\big(\widehat{\gamma}-\gamma\big) \overset{d}{\underset{T \rightarrow \infty}{\longrightarrow}} \Nc_{\Rb^{d_2}}(0,\Vb_{\gamma}),
\eeqw
where the variance-covariance $\Vb_{\gamma}$ is
\beqw
\Vb_{\gamma} = \textcolor{black}{\Ub^{-1}}\textcolor{black}{\Upsilon_{\gamma \Ac_T}}\Vb_{\Ac_{\textcolor{black}{T}}\Ac_{\textcolor{black}{T}}} \textcolor{black}{\Upsilon^\top_{\gamma \Ac_T}}  \textcolor{black}{\Ub^{-1}} +  \textcolor{black}{\Ub^{-1}} \textcolor{black}{\Wb} \textcolor{black}{\Ub^{-1}} \textcolor{black}{+\Ub^{-1} \Upsilon_{\gamma \Ac_T} \Hb^{-1}_{\Ac_T\Ac_T}\Jb_{\Ac_T\gamma}\Ub^{-1} +\Ub^{-1} \Jb^\top_{\Ac_T\gamma}  \Hb^{-1}_{\Ac_T\Ac_T}\Upsilon^\top_{\gamma \Ac_T}\Ub^{-1}}.
\eeqw
\end{proof}

\section{Some competing M-GARCH models}
\label{competing_MGARCH}

The BEKK model directly generates a variance-covariance process. Developed by Baba, Engle, Kraft and Kroner, in a preliminary version of Engle and Kroner (1995), the BEKK is specified for a $p$-dimensional random vector $\eps_t$ as
\beqw
\left\{\begin{array}{clll}
\eps_t & = & H^{1/2}_t \eta_t, \; \text{with} \; H_t := \Eb[\eps_t \eps^\top_t | \Fc_{t-1}]\succ 0 \; \text{so that} \\
H_t & = & \Omega + \overset{q}{\underset{k = 1}{\sum}} \overset{K}{\underset{j = 1}{\sum}} A_{kj} \eps_{t-k} \eps^\top_{t-k} A^\top_{kj} + \overset{r}{\underset{i = 1}{\sum}} \overset{K}{\underset{i = 1}{\sum}} B_{ij} H_{t-i} B^\top_{ij},
\end{array}\right.
\eeqw
where $K$ is an integer, $\Omega$, $A_{kj}$ and $B_{kj}$ are square $p \times p$ matrices and $\Omega \succ 0$. One advantage of the BEKK model is there is no positive semi-definite constraint on the $A_{kj}$ and $B_{kj}$ matrices. However, it imposes highly artificial constraints on the volatilities and covariances of the components. As a consequence, the coefficients of a BEKK representation are difficult to interpret. In our application, a scalar BEKK was considered, where $A_{kj}$ and $B_{kj}$ are scalar with $K = 1$, $q = r = 1$, together with a Gaussian QMLE estimation.

\medskip

Factor models provide rather natural alternatives to BEKK type dynamics. The O-GARCH assumes the decomposition $H_t = P \Lambda_t P^\top$, where $\Lambda_t = \text{diag}(\lambda_{1,t},\cdots,\lambda_{K,t})$, with $K$ the number of factors. Here, we choose $K=p$ factors and each $\lambda_t$ is supposed to follow a univariate GARCH(1,1) process that is estimated by maximum likelihood. The matrix $P$ is nonsingular and it is estimated by PCA on the empirical variance-covariance matrix of $\epsilon_t$: see Alexander (2001), e.g.

\medskip

Rather than a direct specification of the covariance matrices $(H_t)$ dynamics, an alternative approach is to split the task into two parts: individual volatility dynamics on one side, and correlation dynamics on the other side. The most commonly used correlation process is the Dynamic Conditional Correlation (DCC) of Engle (2002). In its BEKK form, the general DCC model is specified as 
\beqw \label{DCC}
\left\{\begin{array}{llll}
\eps_t & = & H^{1/2}_t \eta_t, \; \text{with} \; H_t := \Eb[\eps_t \eps^\top_t | \Fc_{t-1}]\succ 0 \; \text{so that} \\
H_t & = & D_t R_t D_t, \; R_t =  Q^{\star-1/2}_t Q_t Q^{\star-1/2}_t, \\
Q_t & = & \Omega + \overset{q}{\underset{k=1}{\sum}} M_k Q_{t-k} M^\top_k + \overset{r}{\underset{l=1}{\sum}} W_l u_{t-l} u^\top_{t-l} W^\top_l,
\end{array}\right. 
\eeqw
where $D_t = \text{diag}\left(\sqrt{h_{11,t}},\sqrt{h_{22,t}},\ldots,\sqrt{h_{pp,t}}\right)$, $u_t = \left(u_{1,t},\ldots,u_{p,t}\right)$ with $u_{i,t} = \eps_{i,t}/\sqrt{h_{ii,t}}$, $Q_t = \left[q_{ij,t}\right]$, $Q^{\star}_t = \text{diag}\left(q_{11,t},q_{22,t},\ldots,q_{pp,t}\right)$. The model is parameterized by some deterministic matrices $(M_k)_{k = 1,\cdots,q}$, $(W_l)_{l = 1,\cdots,r}$ and a positive definite $p \times p$ matrix $\Omega$. 
Alternatively, Engle (2002) considered a VEC-type specification too. Denoting by $\odot$ the Hadamard matrix product, the $(Q_t)$-dynamics become
\beqw \label{DCC_hadamard}
Q_t = \Omega^* + \overset{q}{\underset{k=1}{\sum}} B_k \odot Q_{t-k} + \overset{r}{\underset{l=1}{\sum}} A_l \odot u_{t-l} u^\top_{t-l},
\eeqw
where the deterministic matrices $(B_k)_{k = 1,\cdots,q}$ and $(A_l)_{l = 1,\cdots,r}$ must be positive semi-definite. 

\medskip

Since the number of parameters of the latter models is of order $O(p^2)$, the matrices $M_k$ and $W_l$ (resp. $B_k$'s and $A_l$) are often assumed to be scalar. 
This is typically a strong and questionable constraint, particularly when the dimension $p$ increases or when the variables in $(\eps_t)$ are heterogeneous.
Furthermore, their inference is usually carried out trough the QML method, based on a Gaussian or Student quasi likelihood function. Under this methodology, applying a regularisation method, even possible, is numerically arduous and no general asymptotic results exist in this case (to the best of our knowledge), due to the non-convexity of the QML criterion. 

\medskip

If $R_t = R$ a constant correlation matrix, then (\ref{DCC}) becomes the Constant Conditional Correlation (CCC) model.

\end{document}